\newif\ifnotes
\notesfalse

\newif\iffocs
\focsfalse

\documentclass[11pt,letterpaper]{article}
\ifnotes
  \usepackage[draft,notes=true,later=true]{dtrt}
\else
  \usepackage[notes=false,later=false]{dtrt}
\fi
\usepackage[T1]{fontenc} %
\usepackage{lmodern} %
\usepackage[utf8]{inputenc} %
\usepackage{cmap} %
\usepackage{multirow}
\usepackage{float}
\usepackage{amsmath}
\usepackage{amsthm}
\usepackage{amssymb}
\usepackage{bm} %
\usepackage{graphicx}
\usepackage{fullpage}
\usepackage{caption}
\usepackage{tikz}
\usetikzlibrary{angles,quotes}
\usepackage{physics}
\usepackage{mathtools}
\usepackage[capitalize]{cleveref}
  \crefname{step}{Step}{Steps}
  \crefname{claim}{Claim}{Claims}
\usepackage{xspace} %
\usepackage{braket}
\usepackage{paralist}
\usepackage{enumitem}
  \setlist[itemize]{leftmargin=*}
  \setlist[enumerate]{leftmargin=*}
\allowdisplaybreaks

\newtheorem{theorem}{Theorem}[section]
\newtheorem{proposition}[theorem]{Proposition}
\newtheorem{claim}[theorem]{Claim}
\newtheorem{lemma}[theorem]{Lemma}

\newtheorem{corollary}[theorem]{Corollary}

\newtheorem{remark}[theorem]{Remark}

\theoremstyle{definition}
\newtheorem{definition}[theorem]{Definition}

\newtheorem{construction}[theorem]{Construction}

\theoremstyle{remark}

\newcommand{\doclearpage}{%
  \iffocs\else\clearpage\fi
}
\newcommand{\FormatAuthor}[3]{
  \begin{tabular}{c}
  #1 \\ {\small\texttt{#2}} \\ {\small #3}
  \end{tabular}
}
\newcommand{\keywords}[1]{\bigskip\par\noindent{\footnotesize\textbf{Keywords\/}: #1}}

\def\bbC{{\mathbb C}}

\def\bbN{{\mathbb N}}

\def\bbR{{\mathbb R}}

\newcommand{\RegA}{\mathcal{A}}
\newcommand{\RegB}{\mathcal{B}}

\newcommand{\RegH}{\mathcal{H}}
\newcommand{\RegI}{\mathcal{I}}

\newcommand{\RegK}{\mathcal{K}}
\newcommand{\RegL}{\mathcal{L}}

\newcommand{\RegO}{\mathcal{O}}

\newcommand{\RegR}{\mathcal{R}}
\newcommand{\RegS}{\mathcal{S}}

\newcommand{\RegV}{\mathcal{V}}
\newcommand{\RegW}{\mathcal{W}}
\newcommand{\RegX}{\mathcal{X}}
\newcommand{\RegY}{\mathcal{Y}}
\newcommand{\RegZ}{\mathcal{Z}}

\newcommand{\defemph}[1]{\textbf{#1}}

\newcommand{\secp}{\lambda}

\newcommand{\eqdef}{\coloneqq}

\newcommand{\Naturals}{\mathbb{N}}
\newcommand{\Bits}{\{0,1\}}
\DeclareMathOperator*{\Expectation}{\mathbb{E}}

\newcommand{\Bin}{\mathsf{Bin}}
\newcommand{\ChernoffN}{n}

\newcommand{\NP}{\mathsf{NP}}
\newcommand{\QMA}{\mathsf{QMA}}
\def\poly{{\rm poly}}
\def\negl{{\rm negl}}

\newcommand{\Malicious}[1]{\tilde{#1}}

\let\oldketbra\ketbra
\renewcommand{\ketbra}[1]{\oldketbra*{#1}}

\newcommand{\contra}[1]{#1^{\dagger}}

\newcommand{\CProj}{\mathsf{CProj}}
\newcommand{\ConProj}{\Game}

\newcommand{\ApproxEig}{\mathsf{ValEst}}

\newcommand{\Id}{\mathbf{I}}
\newcommand{\Meas}[1]{\mathsf{M}_{#1}}
\newcommand{\BProj}[1]{\Pi_{#1}}
\newcommand{\BMeas}[1]{\left( #1, \Id - #1 \right)}

\newcommand{\SProj}[2][]{\Pi_{#2}^{#1}}
\newcommand{\Hermitians}[1]{\mathbf{S}(#1)}

\newcommand{\Projector}{\Pi}
\newcommand{\Measurement}{\mathsf{M}}
\newcommand{\ProjMeasurement}{\mathsf{P}}
\newcommand{\MeasA}{\mathsf{A}}
\newcommand{\MeasB}{\mathsf{B}}
\newcommand{\MeasC}{\mathsf{P}}
\newcommand{\DMatrix}{\bm{\rho}}
\newcommand{\DMatrixB}{\bm{\sigma}}
\newcommand{\CPTPUnitary}{U_{T}}

\newcommand{\NumAlternations}{T}
\newcommand{\Subspace}{\RegS}
\newcommand{\Fractions}[1]{Q_{#1}}

\DeclareMathOperator{\spanset}{span}
\DeclareMathOperator{\image}{image}

\newcommand{\MeasUnitary}{U_{\Measurement}}

\newcommand{\UniState}[1]{\ket{+_{#1}}}
\newcommand{\UniProj}[1]{\ketbra{+_{#1}}}
\newcommand{\Jor}{\mathrm{Jor}}

\newcommand{\MeasJor}{\Meas{\Jor}}
\newcommand{\JorSymb}[3]{#1_{#2,#3}}
\newcommand{\JorKetA}[2]{\ket{\JorSymb{v}{#1}{#2}^{\MeasA}}}
\newcommand{\JorKetB}[2]{\ket{\JorSymb{v}{#1}{#2}^{\MeasB}}}
\newcommand{\JorBraA}[2]{\bra{\JorSymb{v}{#1}{#2}^{\MeasA}}}

\newcommand{\JorBraKetAB}[2]{\braket{\JorSymb{v}{#1}{#2}^{\MeasA} | \JorSymb{v}{#1}{#2}^{\MeasB}}}
\newcommand{\JorKetBraA}[2]{\ketbra{\JorSymb{v}{#1}{#2}^{\MeasA}}}
\newcommand{\JorKetBraB}[2]{\ketbra{\JorSymb{v}{#1}{#2}^{\MeasB}}}
\newcommand{\JorProjA}[2]{\JorKetBraA{#1}{#2}}
\newcommand{\JorProjB}[2]{\JorKetBraB{#1}{#2}}

\newcommand{\Hybrid}{\mathbf{H}}
\newcommand{\Adversary}{\mathsf{Adv}}
\newcommand{\ProverState}{\mathsf{state}}

\newcommand{\MixM}{\mathsf{MixM}}
\newcommand{\Test}{\mathsf{Equals}_{\ket{\psi}}}

\newcommand{\MWDist}{\mathsf{MWDist}}
\newcommand{\MWDistFull}[2]{\MWDist(#2,#1)}
\newcommand{\TestRep}{\mathsf{Test}_{\Advantage}}
\newcommand{\ApproxTest}{\mathsf{ApproxTest}_{\Advantage,t}}

\newcommand{\CRHF}{\mathsf{CRHF}}
\newcommand{\Merkle}{\mathsf{Merkle}}
\newcommand{\MerkleCommit}{\Merkle.\mathsf{Commit}}

\newcommand{\GetKilian}[2]{\mathsf{Kilian}[#1,#2]}

\newcommand{\Language}{\mathcal{L}}
\newcommand{\Relation}{\mathfrak{R}}
\newcommand{\Instance}{x}
\newcommand{\Witness}{w}

\newcommand{\PCPSystem}{\mathsf{PCP}}
\newcommand{\PCPProver}{\mathbf{P}_\PCPSystem}
\newcommand{\PCPVerifier}{\mathbf{V}_\PCPSystem}
\newcommand{\PCPExtractor}{\mathbf{E}}

\newcommand{\PCPTuple}{(\PCPProver,\PCPVerifier)}
\newcommand{\ProofLength}{\ell}
\newcommand{\PCPError}{\varepsilon_{\scriptscriptstyle\PCPSystem}}
\newcommand{\PCPQuery}{\mathsf{qc}}
\newcommand{\PCPRandComplexity}{\mathsf{rc}}
\newcommand{\PCPKnowledge}{\kappa_{\scriptscriptstyle\PCPSystem}}
\newcommand{\PCPProof}{\pi}
\newcommand{\Alphabet}{\Sigma}
\newcommand{\PCPVRandomness}{r}
\newcommand{\PCPQuerySet}{Q}

\newcommand{\ARGSystem}{\mathsf{ARG}}
\newcommand{\ARGProver}{P}
\newcommand{\ARGVerifier}{V}
\newcommand{\ARGExtractor}{E}
\newcommand{\SSExtractor}{\mathsf{Ext}}
\newcommand{\Interact}[2]{\langle #1, #2 \rangle}
\newcommand{\ProtocolCollapseExp}{\mathsf{CollapseExp}}
\newcommand{\NumRounds}{m}
\newcommand{\SoundnessError}{s}
\newcommand{\KnowledgeError}{\kappa}
\newcommand{\UPrvRound}[1]{U^{(#1)}}
\newcommand{\RegChal}{\RegR}
\newcommand{\RegInt}{\RegI}
\newcommand{\RegResp}{\RegZ}
\newcommand{\Response}{z}

\newcommand{\Transcript}{\tau}
\newcommand{\Kilian}{\mathsf{Kilian}}
\newcommand{\WIN}[3]{\mathsf{win}[#1,#2](#3)}

\newcommand{\LastMsg}{z}

\newcommand{\VCScheme}{\mathsf{VC}}
\newcommand{\VCTuple}{(\mathsf{Gen},\mathsf{Commit},\mathsf{Open},\mathsf{Verify})}
\newcommand{\VCGen}{\VCScheme.\mathsf{Gen}}
\newcommand{\VCCommit}{\VCScheme.\mathsf{Commit}}
\newcommand{\VCOpen}{\VCScheme.\mathsf{Open}}
\newcommand{\VCVerify}{\VCScheme.\mathsf{Verify}}
\newcommand{\ck}{\mathsf{ck}}

\newcommand{\veclen}{\ell}
\newcommand{\Message}{m}
\newcommand{\VCcm}{\mathsf{cm}}
\newcommand{\VCaux}{\mathsf{aux}}
\newcommand{\VCauth}{\mathsf{pf}}
\newcommand{\VCQuerySet}{Q}
\newcommand{\VCval}{v}
\newcommand{\VCCollapsingExp}[5]{\mathsf{VCCollapseExp}(#1,#2,#3,#4,#5)}

\newcommand{\GetMerkle}[1]{\mathsf{Merkle}[#1]}
\newcommand{\MRoot}{\mathsf{rt}}
\newcommand{\MTree}{\mathsf{tr}}
\newcommand{\MPath}{\mathsf{path}}
\newcommand{\MHeight}{d}
\newcommand{\MTUnitary}{U}

\newcommand{\HashFamily}{\mathcal{H}}
\newcommand{\HashDistribution}{H}
\newcommand{\HashFunction}{h}
\newcommand{\hinlen}{n}
\newcommand{\houtlen}{\ell}
\newcommand{\HCollapsingExp}[3]{\mathsf{HCollapseExp}(#1,#2,#3)}

\newcommand{\SuccProb}{\eta}
\newcommand{\RWLoss}{\eta_0}

\newcommand{\Repair}{\mathsf{Repair}}
\newcommand{\RepairProb}{\mathsf{ValRepair}}

\newcommand{\Advantage}{\varepsilon}

\newcommand{\AEigReps}{t}
\newcommand{\RSet}{R}
\newcommand{\ZSet}{Z}
\newcommand{\PEst}{\tilde{p}}
\newcommand{\Pred}{f}

\renewcommand{\Game}{\mathcal{G}}
\newcommand{\Value}[1]{\omega_{#1}}
\newcommand{\StrategyMap}{S}
\begin{document}
\title{Post-Quantum Succinct Arguments:\\
Breaking the Quantum Rewinding Barrier}
\author{
  \begin{tabular}[h!]{cc}
      \FormatAuthor{Alessandro Chiesa}{alexch@berkeley.edu}{UC Berkeley}
    & \FormatAuthor{Fermi Ma}{fermima@alum.mit.edu}{Princeton and NTT Research}
  \end{tabular}\\
  \begin{tabular}[h!]{cc}
    \\
      \FormatAuthor{Nicholas Spooner}{nspooner@bu.edu}{Boston University}
    & \FormatAuthor{Mark Zhandry}{mzhandry@gmail.com}{Princeton and NTT Research} 
  \end{tabular}
}
\date{\today}
\maketitle

\begin{abstract}

We prove that Kilian’s four-message succinct argument system is post-quantum secure in the standard model when instantiated with any probabilistically checkable proof and any collapsing hash function (which in turn exist based on the post-quantum hardness of Learning with Errors). This yields the first post-quantum succinct argument system from any falsifiable assumption.

At the heart of our proof is a new quantum rewinding procedure that enables a reduction to repeatedly query a quantum adversary for accepting transcripts as many times as desired. Prior techniques were limited to a \emph{constant} number of accepting transcripts.

\keywords{succinct arguments; post-quantum cryptography; quantum rewinding}

\end{abstract}

\iffocs\else
\clearpage
\setcounter{tocdepth}{2}
\tableofcontents
\clearpage
\fi

\doclearpage
\section{Introduction}
\label{sec:intro}

Quantum computers pose a growing threat to cryptography. Fully realized, quantum computers would enable an attacker to break the computational assumptions underlying many of today's public-key cryptosystems \cite{Shor94}. Fortunately, a number of plausibly \emph{quantum-secure} computational assumptions have emerged (e.g., lattice assumptions \cite{Regev05}) providing a foundation for secure cryptography in a post-quantum era. But post-quantum cryptography requires more than quantum-safe assumptions: it also needs \emph{security reductions} compatible with quantum attackers. While some classical security reductions directly translate to the quantum setting, many other security reductions do not translate because they are not compatible with quantum attackers.

Kilian's protocol \cite{Kilian92} is a fundamental result in cryptography for which no security reduction compatible with quantum attackers is known. Kilian's protocol is the canonical construction of a \emph{succinct argument}: it uses a collision-resistant hash function to transform any probabilistically checkable proof (PCP) into an interactive protocol that achieves an exponential improvement in communication complexity over just sending the PCP. This comes at the cost of \emph{computational} soundness, i.e., fooling the verification procedure of the protocol is intractable, not impossible. The security reduction against a classical attacker is via a \emph{rewinding argument}: the attacker's state is saved midway through the protocol execution, and the attacker is run from this state many times to obtain many (succinct) protocol executions, from which the (long) PCP string can be extracted.

Alarmingly, Kilian's security reduction completely falls apart if the attacker has a quantum computer! The reduction has access to only a single copy of the attacker's state, due to the \emph{no-cloning theorem}. Moreover, since quantum measurements are \emph{destructive}, any attempt to measure the attacker's response may irreversibly damage the attacker's state, potentially rendering it useless.

Translating rewinding-based security reductions to the quantum setting has proved difficult (see e.g.,~\cite{AmbainisRU14}). While there has been some progress on developing quantum techniques tailored to specific use cases \cite{Watrous06,Unruh12,Unruh16-eurocrypt}, these techniques are not broadly applicable. Importantly, existing quantum rewinding techniques are limited to recording a \emph{constant} number of attacker responses. This is particularly problematic for Kilian's protocol and beyond: all known techniques for reducing security of a succinct argument to an underlying (falsifiable) assumption require the reduction to record a super-constant (and typically polynomial) number of attacker responses.\footnote{Even if a classical security proof relies on an explicitly post-quantum assumption (e.g., \cite{BaumBCPGL18,BootleLNS20}) this does not translate to \emph{provable} post-quantum security as the rewinding security reduction is not quantum-compatible.}

One way to avoid rewinding security reductions for succinct arguments is to rely on strong cryptographic assumptions. Kilian's protocol can be proved secure via a straightline (non-rewinding) extractor when ported to the random oracle model, and its security in the quantum random oracle model \cite{BonehDFLSZ11} follows from prior work \cite{ChiesaMS19}. Beyond Kilian's protocol, there are constructions of succinct arguments that are proved secure directly from underlying post-quantum ``knowledge'' assumptions \cite{BonehISW17,BonehISW18,GennaroMNO18}, but these assumptions are not falsifiable.\footnote{See~\cite{Naor03,GentryW11} for further discussion on falsifiable assumptions.}

In sum, the following question remains open:
\begin{center}
\emph{Do post-quantum succinct arguments exist under standard assumptions?}
\end{center}

\subsection{Our results}
\label{subsec:results}

We answer the question affirmatively by proving that Kilian's protocol is post-quantum secure, provided the underlying hash function is \emph{collapsing}~\cite{Unruh16-eurocrypt}. 

\begin{theorem}[Kilian's protocol is post-quantum secure]
\label{thm:informal-kilian}
Kilian's protocol is a post-quantum succinct argument when instantiated with any PCP and any collapsing hash function. Moreover, if the underlying PCP is a proof of knowledge, Kilian's protocol is a post-quantum succinct argument of knowledge.
\end{theorem}

Since collapsing hash functions are implied by post-quantum lossy functions~\cite{Unruh16-asiacrypt}, which exist assuming the quantum hardness of Learning with Errors (QLWE), we obtain post-quantum succinct arguments for all of $\NP$ from the same assumption. This is the first construction of post-quantum succinct arguments from \emph{any} falsifiable assumption.

\begin{corollary}[Post-quantum succinct arguments from QLWE]
Assuming quantum hardness of LWE (QLWE), there exist post-quantum succinct arguments (of knowledge) for all of $\NP$.
\end{corollary}

The core of our proof is a new quantum extraction procedure that enables a reduction to record the prover's responses for an \emph{arbitrary} number of random challenges. This significantly improves over prior work, which was limited to recording the prover's responses for a \emph{constant} number of random challenges \cite{Unruh12,Unruh16-eurocrypt,DonFMS19}. 

Our extraction procedure applies not only to Kilian's protocol, but any \emph{collapsing protocol} \cite{Unruh16-eurocrypt,LiuZ19,DonFMS19}. A collapsing protocol refers to any public-coin interactive argument with the guarantee, roughly, that any (unitary) prover that only gives accepting responses cannot detect if its last response is measured. We show Kilian's protocol has this guarantee if it is instantiated with a collapsing hash function.

\begin{theorem}[Quantum rewinding, informal] Given black-box access to any quantum adversary for a collapsing protocol, there is an efficient procedure to repeatedly query the adversary on random challenges and record an arbitrary number of accepting transcripts.
\end{theorem}

Beyond our primary application to Kilian's protocol, our quantum rewinding procedure also implies that any $k$-special sound collapsing protocol is a post-quantum argument of knowledge, for any polynomially-bounded $k$.

\parhead{Optimal knowledge error}
Our rewinding technique achieves asymptotically optimal knowledge error. As an immediate application, our technique improves a previous result due to \cite{Unruh12,Unruh16-eurocrypt}, who showed that if a quantum attacker in a $2$-special sound collapsing sigma protocol has success probability $\Advantage$, then there is an extractor that can output a witness with probability $\Advantage \cdot (\Advantage^2-1/C)$, where $C$ is the size of the challenge space. In particular, there was previously no guarantee for $1/C \leq \Advantage \leq 1/\sqrt{C}$. Our techniques yield an extractor running in time $\poly(\secp,1/\Advantage)$ that (given $\Advantage$ as input) outputs a witness with probability $\Omega(\Advantage)$ provided that $\Advantage \geq (1+\delta)/C$ for any constant $\delta > 0$.

\doclearpage
\section{Technical overview}
\label{sec:techniques}

\subsection{Kilian's protocol}
\label{subsec:kilian}

Kilian's protocol compiles any \emph{probabilistically checkable proof} (PCP) into an interactive protocol using a Merkle tree built from a collision-resistant hash function. Recall that a PCP is a type of $\NP$ proof $\PCPProof$ that can be verified by reading only a few random positions \cite{BabaiFLS91,FeigeGLSS91,AroraS98,AroraLMSS98}. The collision-resistant hash function enables the argument prover to send a \emph{succinct} Merkle tree commitment to the PCP $\PCPProof$ that it can later open on any subset of positions $\PCPQuerySet$ with a short opening proof.

\parhead{The protocol}
Let $(\PCPProver,\PCPVerifier)$ be a PCP proof system for an $\NP$ relation $\Relation$, and let $\{H_{\secp}\}_{\secp}$ be a family of collision-resistant hash functions. The argument prover $\ARGProver$ and argument verifier $\ARGVerifier$ both receive as input the security parameter $\secp$ and an instance $\Instance$, while the prover additionally receives a corresponding witness $\Witness$ (such that $(\Instance,\Witness) \in \Relation$). They interact as follows.
\begin{enumerate}[noitemsep]
\item $\ARGVerifier$ samples a collision-resistant hash function $h_{\CRHF} \gets H_{\secp}$ and sends it to $\ARGProver$.
\item $\ARGProver$ computes a PCP string $\PCPProof \leftarrow \PCPProver(\Instance,\Witness)$, uses $h_{\CRHF}$ to generate a Merkle tree commitment $\VCcm \gets \MerkleCommit(h_{\CRHF},\PCPProof)$ to $\PCPProof$, and sends $\VCcm$ to $\ARGVerifier$.
\item $\ARGVerifier$ samples random coins $r \leftarrow R$ for the PCP verifier $\PCPVerifier$ and sends them to $\ARGProver$.
\item $\ARGProver$ computes the PCP indices $\PCPQuerySet$ that $\PCPVerifier(\Instance;r)$ would query, generates a Merkle opening proof $\VCauth$ for $\PCPProof[\PCPQuerySet]$, and sends the response $z \eqdef (\PCPProof[\PCPQuerySet],\VCauth)$ to $\ARGVerifier$.\footnote{The Merkle opening for a PCP index $q$ consists of the hash values of every vertex adjacent to the path from $q$ to the root; the Merkle opening proof $\VCauth$ for a set of PCP indices $Q$ consists of the Merkle openings for each $q \in Q$.}
\end{enumerate}
Once the interaction is complete, $\ARGVerifier$ accepts if:
\begin{inparaenum}[(1)]
  \item $\VCauth$ is a valid Merkle opening of $\VCcm$ to $\PCPProof[\PCPQuerySet]$ on indices $\PCPQuerySet$; and
  \item $\PCPProof[\PCPQuerySet]$ is accepted by the PCP verifier $\PCPVerifier(\Instance;r)$. 
\end{inparaenum}
Kilian's protocol is \emph{publicly verifiable}: one can compute whether $V$ accepts given only the instance $x$ and the four-message transcript $(h_{\CRHF},\VCcm,r,z)$.

\parhead{The classical security reduction}
Kilian's protocol ensures that an efficient extractor, given a malicious classical prover $\Malicious{\ARGProver}$ that convinces $\ARGVerifier$ with success probability $2\Advantage$, can output with overwhelming probability a PCP $\PCPProof$ such that $\Pr[\PCPVerifier^{\PCPProof}(\Instance)] \geq \Advantage/2$; the particular constants here are chosen to simplify the presentation in the following steps.

The extractor works by running $\Malicious{\ARGProver}$ through the first round of the protocol, obtaining a transcript prefix $\Transcript = (h_{\CRHF},\VCcm)$ and $\Malicious{\ARGProver}$'s intermediate state $\ProverState_{\Transcript}$. Call $\ProverState_{\Transcript}$ ``$\Advantage$-good'' if
\begin{equation*}
\Pr\left[\ARGVerifier(\Transcript,r,z)=1 \, \middle\vert \begin{array}{r}
r \gets R \\
z \gets \Malicious{\ARGProver}(\ProverState_{\Transcript},r) \end{array} \right] \geq \Advantage \enspace.
\end{equation*}
By Markov's inequality, $\ProverState_{\Transcript}$ is $\Advantage$-good with probability at least $\Advantage$. If $\ProverState_{\Transcript}$ is $\Advantage$-good, the extractor constructs a PCP proof $\PCPProof$ as follows.
\begin{itemize}[noitemsep]
\item[] Start with $\PCPProof \eqdef 0^{\ProofLength}$ where $\ell$ is the PCP proof length. Repeat the loop:
\begin{enumerate}[nolistsep]
    \item Choose $r \gets R$ uniformly at random.
    \item Run $z \gets \Malicious\ARGProver(\ProverState_{\Transcript}, r)$.
    \item If $\ARGVerifier(\Transcript,r,z) = 1$, parse $z$ as $(\PCPProof'[\PCPQuerySet],\VCauth)$. Update $\PCPProof$ to match $\PCPProof'$ at the positions in $\PCPQuerySet$.
\end{enumerate}
\end{itemize}

If the PCP has alphabet $\Alphabet$ and proof length $\ProofLength$, one can show that if the extractor records $k = 6 \ProofLength \cdot \log(2|\Alphabet|)$ challenge-response pairs $(r_1,\LastMsg_1),\dots,(r_k,\LastMsg_k)$ for distinct challenges $r_i$, then with probability $1 - \negl(\secp)$ the PCP string $\PCPProof$ satisfies $\Pr[\PCPVerifier^{\PCPProof}(\Instance)] \geq \Advantage/2$.

This guarantee implies the \emph{classical} security of Kilian's protocol. For instance, if the PCP system has negligible soundness error then the interactive argument has negligible soundness error.

\subsection{Our approach to post-quantum security of Kilian's protocol}

In this work, we prove that if the collision-resistant hash function $h_{\CRHF}$ is a \emph{collapsing hash function}~\cite{Unruh16-eurocrypt}, then Kilian's protocol, without any additional modifications, is secure against malicious quantum provers. At a very high level, our security proof takes the following steps:
\begin{enumerate}
    \item\label[step]{step:tech-overview-collapsing} \textbf{Kilian's protocol is collapsing.} We prove that Kilian's protocol is a \emph{collapsing protocol} in the sense of~\cite{LiuZ19,DonFMS19} when the underlying hash function is collapsing; we elaborate on collapsing protocols in~\cref{subsec:tech-prior-quantum-techniques}.
    \item\label[step]{step:tech-overview-extraction} \textbf{Collapsing protocols admit quantum rewinding.} We devise a general-purpose quantum extraction procedure for collapsing protocols that enables efficiently recording any desired number of malicious prover responses. This step is our main technical contribution.
\end{enumerate}

\parhead{Organization} We discuss the importance of the collapsing notion in~\cref{subsec:tech-prior-quantum-techniques}, but will otherwise defer the details of~\cref{step:tech-overview-collapsing} to the body of the paper, since proving that Kilian's protocol is collapsing is a straightforward application of techniques from~\cite{Unruh16-eurocrypt}. 

\cref{step:tech-overview-extraction} is the primary focus of this technical overview. We summarize prior work on rewinding for collapsing protocols in~\cref{subsec:tech-prior-quantum-techniques} and explain in~\cref{subsec:tech-unruh} why existing techniques are insufficient for Kilian. We then describe our extraction procedure in~\cref{subsec:tech-state-repair,subsec:tech-state-recovery,{subsec:tech-approximate-jordan}}. %
\subsection{Prior quantum techniques}
\label{subsec:tech-prior-quantum-techniques}

We discuss prior techniques for recording responses of a malicious quantum prover in a classical interactive (public-coin) protocol. While prior works did not explicitly focus on Kilian's protocol, the abstract setting is the same. A reduction runs a malicious prover $\Malicious{\ARGProver}$ up to the final round of the protocol, obtaining a fixed transcript prefix $\Transcript$ and corresponding prover state $\ProverState_{\Transcript}$. Assuming that $\Malicious{\ARGProver}(\ProverState_{\Transcript},\cdot)$ answers a random challenge $r \gets R$ with \emph{success probability} $\Advantage$, the goal is to obtain some number $k$ of accepting transcripts $(\Transcript,r_1,\LastMsg_1),\dots,(\Transcript,r_k,\LastMsg_k)$ with the same prefix $\Transcript$.

In the classical setting, this is an elementary task. By repeatedly sampling random challenges $r \gets R$ and running $z \gets \Malicious{\ARGProver}(\ProverState_{\Transcript},r)$, we can record any desired number of independent and identically distributed transcripts where an $\Advantage$-fraction of them are accepting. Put another way:
\begin{quote}
Given $\Malicious{\ARGProver}$ and $\ProverState_{\Transcript}$, one can record $k$ \emph{accepting} transcripts for any desired $k$ with probability $1$ in expected time $k/\Advantage$.
\end{quote}

In the quantum setting, it is unlikely that such a statement holds: if $\ProverState_{\Transcript}$ is a quantum state $\ket{\psi}$, it is not possible in general to run $\Malicious{\ARGProver}(\ProverState_{\Transcript},\cdot)$ multiple times independently. This is because any measurement applied by $\Malicious{\ARGProver}$ may irreversibly alter the state. Indeed, Ambainis, Rosmanis, and Unruh~\cite{AmbainisRU14} show that this statement can be false relative to a (quantum) oracle, even if $(\ARGProver,\ARGVerifier)$ is classically secure.

\parhead{Collapsing protocols}
Nevertheless, there \emph{is} a class of protocols for which the statement holds in a limited sense. A public-coin interactive argument is a \emph{collapsing protocol} \cite{Unruh16-eurocrypt,DonFMS19,LiuZ19} if, given any last-round challenge $r$, an efficient prover which produces a superposition $\ket{\phi}$ of \emph{accepting responses} cannot distinguish between $\ket{\phi}$ and the state that results after measuring the response in the computational basis.\footnote{More precisely, $\ket{\phi} = \sum_{y,z} \alpha_{y,z} \ket{y,z}$ where each $z$ in the superposition satisfies $\ARGVerifier(\Transcript,r,z) = 1$ for some fixed partial transcript $\Transcript$, and $y$ is the state on other registers. \emph{Measuring the response} means measuring the register containing $z$.}
We remark that~\cite{DonFMS19,LiuZ19} defined collapsing protocols in the context of three-round sigma protocols, but the notion easily extends to public-coin interactive arguments.

For any collapsing protocol $(\ARGProver,\ARGVerifier)$, Unruh's lemma~\cite{Unruh12,DonFMS19} gives a weaker version of the above statement. Suppose a malicious $\Malicious{\ARGProver}$ with state $\ket{\psi}$ has initial success probability $\Advantage$, i.e., $\Malicious{\ARGProver}(\ket{\psi},r)$ outputs an accepting response $z$ on a random $r \gets R$ with probability $\Advantage$. Then Unruh's lemma gives the following guarantee:

\begin{quote}
Given $\Malicious{\ARGProver}$ and $\ket{\psi}$, one can record $k$ \emph{accepting} transcripts for any desired $k$ with probability $O(\Advantage^{2k-1})$.
\end{quote}

This $O(\Advantage^{2k-1})$ probability, which does not appear in the classical statement, is over the randomness of the challenges and any quantum measurements the malicious prover performs. Notice that for constant $k$, this probability is still large enough to obtain meaningful guarantees. However, security of Kilian's protocol needs, at a minimum, $k=\Omega(\ProofLength/|\PCPQuerySet|)$ where $\ProofLength$ is the PCP length and $|\PCPQuerySet|$ is the number of queries of the PCP verifier. Thus, Unruh's lemma is insufficient since the guarantee only holds with probability $\Advantage^{\Omega(\ProofLength/|\PCPQuerySet|)}$, which is negligible for any PCP with useful parameters.

\subsection{A closer look at Unruh's lemma}
\label{subsec:tech-unruh}

Unruh's lemma is a quantum information-theoretic statement about any collection of binary-outcome projective measurements $\{\MeasA_r\}_{r \in R}$. We write binary-outcome projective measurements as $\MeasA_r = \BMeas{\Pi_r}$ where $\Pi_r$ is associated with outcome $1$, and $\Id -\Pi_r$ with outcome $0$.

Let $\MixM(\{\MeasA_r\}_{r})$ be the corresponding \emph{mixture} of the projective measurements $\{ \MeasA_r\}_{r}$, i.e., the procedure that chooses $r \gets R$ uniformly at random, applies measurement $\MeasA_r$, and outputs the outcome $b \in \Bits$. Unruh's lemma~\cite{Unruh12,DonFMS19} concerns the measurement outcomes obtained from sequential applications of $\MixM(\{ \MeasA_r\}_{r})$.
\begin{quote}
    \textbf{Unruh's lemma}: For any state $\ket{\psi}$ and any collection of binary-outcome projective measurements $\{ \MeasA_r\}_{r \in R}$, if applying $\MixM(\{ \MeasA_r\}_{r})$ to $\ket{\psi}$ returns $1$ with probability $\Advantage$, then starting from $\ket{\psi}$ and applying $\MixM(\{ \MeasA_r\}_{r})$ for $k$ times in succession returns $1$ all $k$ times with probability $\Advantage^{2k-1}$.
\end{quote}

To use this lemma in the context of an interactive protocol, for each $r$ in the challenge space $R$ one defines $\MeasA_r = \BMeas{\Pi_r}$ as follows. Let $U_r$ be the unitary describing the (purified) operation of $\Malicious{\ARGProver}$ in the last round on verifier message $r$; let $\Pi_{V,r} \eqdef \sum_{z,\ARGVerifier(\Transcript,r,z) = 1} \ketbra{z}$ be the projection onto responses $z$ that the verifier $\ARGVerifier(\Transcript,r,\cdot)$ accepts; and finally set $\Pi_r \eqdef \contra{U_r} \Pi_{\ARGVerifier,r} U_r$.

Intuitively, $\MeasA_r$ measures \emph{whether} $\Malicious{\ARGProver}$ causes $\ARGVerifier$ to accept on challenge $r$. Therefore, the probability $\Advantage$ in Unruh's lemma (the probability $\MixM(\{\MeasA_r\}_r)$ applied to $\ket{\psi}$ returns $1$) is the probability that $\Malicious{\ARGProver}(\ket{\psi},\cdot)$ successfully answers a random challenge $r \gets R$ in the interactive protocol. We sometimes refer to $\Advantage$ as the \emph{success probability} of $\ket{\psi}$.

Thus Unruh's lemma shows that it is possible to ``observe'' $k$ accepting executions with probability $\Advantage^{2k-1}$, in the following sense: whenever $\MixM$ returns $1$, one can apply $U_{r}$ for the $r$ sampled by $\MixM$, and measure the adversary's response register to obtain $z$ such that $(\Transcript,r,z)$ is an accepting transcript. Importantly, because Unruh's lemma only concerns \emph{binary-outcome} projective measurements, we require an additional \emph{collapsing} property from the underlying protocol to (undetectably) \emph{record} any accepting responses. Thus, applied to a \emph{collapsing protocol}, Unruh's lemma implies an extractor can record $k$ accepting transcripts with probability $\Advantage^{2k-1} - \negl(\secp)$, since this additional measurement of the response register is (computationally) undetectable when $\MixM$ returns $1$.

\parhead{Consecutive measurements can destroy a state}
The $\Advantage^{2k-1}$ probability comes in part from the fact that Unruh's lemma only captures the probability that $k$ \emph{consecutive} trials succeed.\footnote{Technically, $\Advantage^{2k-1}$ only applies for random uncorrelated challenges, which may not be distinct. Unruh also gives a bound that applies for distinct random challenges.} This is a strong requirement: even in the classical setting, $k$ consecutive trials succeed with probability $\Advantage^{k}$. Classically this can be resolved by performing $N = k/\epsilon$ trials to obtain roughly $k$ successful trials. One might hope that this would also work in the quantum setting: perhaps repeatedly applying $\MixM(\{ \Meas{r} \}_{r})$ some $\poly(k,1/\Advantage)$ times suffices to obtain $k$ successful trials overall.

Unfortunately, this does not work. Adapting a counterexample of Zhandry \cite[Section 5]{Zhandry20}, suppose the initial state $\ket{\psi}$ is $\ket{0}$, and for any desired success probability $\Advantage$, define each $\MeasA_r = \BMeas{\Pi_r}$ so that $\Pi_r$ is the rank-one projection onto $\sqrt{\Advantage}\ket{0} +\sqrt{1-\Advantage}\ket{r}$. Clearly, $\MixM$ applied to $\ket{\psi}$ returns $1$ with probability $\Advantage$, but one can verify that if repeated applications of $\MixM$ use \emph{distinct} challenges $r$, then the expected number of $1$ outcomes is at most $1/(2 - 2\Advantage)$ regardless of the number of trials; for small $\Advantage$ this is close to $1/2$. This counterexample is a barrier if there are a super-polynomial number of challenges, as each trial will use a distinct $r$ with overwhelming probability. Note that, in this example, the bound $1/(2-2\Advantage)$ arises because the (expected) success probability of the state after $j$ trials is exponentially small in $j$. In other words, the repeated applications of $\MixM$ ``damage'' the state.

\subsection{State recovery}
\label{subsec:tech-state-recovery}

Given the above discussion, a natural approach is to try to recover the original state after the application of $\MixM(\{\Meas{r}\}_r)$. In particular, it would suffice to build a procedure that would allow recovering a state $\ket{\psi}$ after it has been perturbed by some binary projective measurement $\mathsf{B}$. In our setting, $\ket{\psi}$ corresponds to the malicious prover's intermediate state, and $\mathsf{B}$ is the measurement $\Meas{r}$ applied by $\MixM(\{\Meas{r}\}_r)$. Applying $\Meas{r}$ to $\ket{\psi}$ disturbs the state, leaving some post-measurement state $\ket{\phi}$, and our aim is to somehow return the state back to $\ket{\psi}$. If we could do this in general (for any efficient binary projective measurement $\mathsf{B}$) this would enable ``perfect'' quantum rewinding.

Unfortunately, this is impossible in general, but to build intuition for our eventual approach, we will show how to achieve this assuming we have access to a hypothetical additional power. In particular, suppose we can perform the binary projective measurement
\begin{equation*}
\Test = \BMeas{\ketbra{\psi}}
\end{equation*}
onto the one-dimensional subspace spanned by the initial state $\ket{\psi}$. %
If $\Test$ returns the outcome $1$, then the post-measurement state is $\ket{\psi}$. In the remainder of this section, we use $\Test$ to develop a procedure that recovers the state $\ket{\psi}$ with probability close to $1$.

\parhead{The qubit case}
First we consider the case where $\ket{\psi}$ is a single qubit: $\ket{\psi}$ lies in the \emph{two-dimensional} space $\bbC^2$. If $\mathsf{B} = \BMeas{\Pi}$ is nontrivial, then $\Pi = \ketbra{\phi}$ and $\Id-\Pi = \ketbra{\phi^\perp}$ for some pair of orthogonal states $\ket{\phi},\ket{\phi^\perp} \in \bbC^2$. This is shown in \cref{fig:two-dim}.

\begin{figure}[H]
\captionsetup{width=.9\linewidth}
\centering
\begin{tikzpicture}[
every edge quotes/.append style = {anchor=south, sloped}
                        ]
\draw[thick, ->] (0,0) -- (1.5,2.59807621135) coordinate[label=above right:$\ket{\psi}$] (v);
\draw[thick, ->] (0,0) -- (-2.59807621135,1.5) coordinate[label=above right:$\ket{\psi^\perp}$] (vp);
\draw[thick, ->] (0,0) -- (0,3) coordinate[label=above right:$\ket{\phi}$] (w);
\draw[thick, ->] (0,0) -- (3,0) coordinate[label=above right:$\ket{\phi^\perp}$] (wp);
\coordinate (O);

\draw[dashed] (0,2.59807621135) node[left] {$\sqrt{p}$} -- (1.5,2.59807621135);
\draw[dashed] (1.5,0) node[below] {$\sqrt{1-p}$} -- (1.5,2.59807621135);
\end{tikzpicture}
\caption{The quantum states $\ket{\psi}$ and $\ket{\psi^\perp}$ correspond to outcomes $1$ and $0$ of $\Test = \BMeas{\ketbra{\psi}}$, respectively. The quantum states $\ket{\phi}$ and $\ket{\phi^\perp}$ correspond to outcomes $1$ and $0$ of $\mathsf{B} = \BMeas{\Pi}$, respectively.}
\label{fig:two-dim}
\end{figure}

From \cref{fig:two-dim} we see that $|\braket{\phi|\psi}|^2 = \bra{\psi}\ketbra{\phi}\ket{\psi} = \norm{\Pi \ket{\psi}}^2 = p$. By making a suitable choice of phase, we can write
\begin{align*}
\ket{\phi} &= \sqrt{p} \ket{\psi} + \sqrt{1-p} \ket{\psi^\perp} \enspace,\\
\ket{\psi} &= \sqrt{p} \ket{\phi} + \sqrt{1-p} \ket{\phi^\perp} \enspace.
\end{align*}
Suppose that we have applied $\mathsf{B}$ to the state $\ket{\psi}$ and obtained the outcome $1$. (The case of outcome $0$ is symmetric.) The post-measurement state is then $\ket{\phi}$. A natural idea to recover the original state $\ket{\psi}$ is to apply $\Test$ to $\ket{\phi}$:
\begin{itemize}[noitemsep]
    \item With probability $p$, we obtain the outcome $1$ and the state is $\ket{\psi}$.
    \item With probability $1-p$ we obtain the outcome $0$ and the state is $\ket{\psi^{\perp}}$ (which only holds because the space is two-dimensional).
\end{itemize}
In the first case we are done. But even in the second case we are not ``stuck'': if we apply $\mathsf{B}$ \emph{again}, then with probability $1-p$ we return to the state $\ket{\phi}$, and with probability $p$ we move to the state $\ket{\phi^{\perp}}$. This leads to a ``state recovery'' procedure, which follows a technique first used by Marriott and Watrous for $\QMA$ amplification~\cite{MarriottW05}.\footnote{The goal of~\cite{MarriottW05} was not to reconstruct a particular quantum state, but to estimate the probability $p$.} After potentially disturbing the state $\ket{\psi}$ by applying $\mathsf{B}$, we can recover $\ket{\psi}$ by simply alternating the measurements 
\begin{equation*}
\Test, \mathsf{B}, \Test, \mathsf{B},\dots
\end{equation*}
until $\Test$ returns $1$, at which point the state must be $\ket{\psi}$. In fact, the state of the system and the measurement outcomes throughout the procedure are remarkably easy to characterize. For instance, the effect of each $\Test$ measurement can be deduced from~\cref{fig:two-dim}:
\begin{itemize}[noitemsep]
    \item Applying $\Test$ to $\ket{\phi}$ returns $1$ with probability $p$ resulting in $\ket{\psi}$, and returns $0$ with probability $1-p$ resulting in  $\ket{\psi^\perp}$.
    \item Applying $\Test$ to $\ket{\phi^\perp}$ returns $0$ with probability $p$ resulting in $\ket{\psi^\perp}$, and returns $1$ with probability $1-p$ resulting in $\ket{\psi}$.
\end{itemize}
The effect of $\mathsf{B}$ on $\ket{\psi}$ and $\ket{\psi^\perp}$ is analogous. Letting $b_i$ denote the outcome of the $i$-th measurement, starting from $\ket{\psi}$ and applying $\mathsf{B},\Test,\ldots$ in alternating fashion (now counting the initial $\mathsf{B}$ as part of the sequence), the outcome sequence $b_1,b_2,\ldots$ follows a classical distribution $\MWDist(p)$ (for ``Marriott--Watrous''): 
\begin{enumerate}[noitemsep]
    \item Initialize $b_0 = 1$ (the initial state $\ket{\psi}$ corresponds to the $1$ outcome of $\Test$).
    \item For each $i \in \Naturals$, set $b_i := b_{i-1}$ with probability $p$, and $b_i := 1-b_{i-1}$ otherwise.
\end{enumerate}
With this characterization, we can analyze the procedure's running time. The procedure fails to terminate at the first application of $\Test$, corresponding to $b_2 = 0$, with probability $2p(1-p)$. If this occurs, the next application of $\Test$ returns $0$ with probability $1 - 2p(1-p)$. Continuing with this argument, the probability the procedure fails to terminate after $2\NumAlternations$ total measurements is 
\begin{equation*}
2p(1-p)(1-2p(1-p))^{\NumAlternations-1} < 1/\NumAlternations \enspace,
\end{equation*}
where the inequality holds for \emph{any} probability $p$.

\parhead{Extending to more qubits}
The analysis above relies on the fact that, in two dimensions, the system throughout the alternating measurement procedure is easily seen to lie in one of the four states $\{\ket{\psi},\ket{\psi^\perp},\ket{\phi},\ket{\phi^\perp}\}$. In higher dimensions, the behavior of the system is potentially more complex.\footnote{In the current setting, since $\Test$ projects onto a rank-one subspace, it turns out that even in higher dimensions the behaviour of this particular system will be two-dimensional, moving between states $\ket{\psi},(\Pi-p\Id)\ket{\psi},\Pi\ket{\psi},(\Id-\Pi)\ket{\psi}$ (appropriately normalized). Our more general treatment will be useful later on when we replace $\Test$ with a projection onto a higher-dimensional subspace.} We can nevertheless prove that the procedure terminates after $2\NumAlternations$ measurements with probability at most $1/\NumAlternations$.

To analyze the multi-qubit case, we use Jordan's lemma, a tool in quantum information theory that extends two-dimensional analyses of a pair of projectors to higher dimensions. Specifically, \emph{any} two projectors $\BProj{\MeasA},\BProj{\MeasB}$ induce a decomposition of the ambient Hilbert space into two-dimensional subspaces $\Subspace_{j}$ such both $\BProj{\MeasA}$ and $\BProj{\MeasB}$ act as rank-one projectors within each subspace.\footnote{There are also one-dimensional subspaces, which we ignore here for the purpose of exposition; in any case, these can be treated as ``degenerate'' two-dimensional subspaces.}

More precisely, for each ``Jordan subspace'' $\Subspace_{j}$, there exist orthogonal vectors $\JorKetA{j}{1},\JorKetA{j}{0}$ that span $\Subspace_{j}$, such that $\BProj{\MeasA} \JorKetA{j}{1} = \JorKetA{j}{1}$ and $\BProj{\MeasA} \JorKetA{j}{0} = 0$; similarly, there exist orthogonal vectors $\JorKetB{j}{1},\JorKetB{j}{0}$ that span $\Subspace_{j}$ such that $\BProj{\MeasB} \JorKetB{j}{1} = \JorKetB{j}{1}$ and $\BProj{\MeasB} \JorKetB{j}{0} = 0$. Defining the \emph{eigenvalue} of $\Subspace_{j}$ as $p_j \coloneqq \abs{\braket{v_j | w_j}}^2$, within each subspace $\Subspace_{j}$ we recover a two-dimensional picture, as in \cref{fig:jordan}. We refer to $p_j$ as the ``eigenvalue'' of $\Subspace_{j}$ because $\JorKetA{j}{1}$ is an eigenvector of the Hermitian matrix $\BProj{\MeasA} \BProj{\MeasB} \BProj{\MeasA}$ with eigenvalue $p_j$ (and $\JorKetB{j}{1}$ is an eigenvector of $\BProj{\MeasB} \BProj{\MeasA} \BProj{\MeasB}$ with eigenvalue $p_j$).

\begin{figure}[H]
\captionsetup{width=.9\linewidth}
\centering
\begin{tikzpicture}[
every edge quotes/.append style = {anchor=south, sloped}
                        ]
\draw[thick, ->] (0,0) -- (1.5,2.59807621135) coordinate[label=above right:$\JorKetA{j}{1}$] (v);
\draw[thick, ->] (0,0) -- (-2.59807621135,1.5) coordinate[label=above right:$\JorKetA{j}{0}$] (vp);
\draw[thick, ->] (0,0) -- (0,3) coordinate[label=above right:$\JorKetB{j}{1}$] (w);
\draw[thick, ->] (0,0) -- (3,0) coordinate[label=above right:$\JorKetB{j}{0}$] (wp);
\coordinate (O);

\draw[dashed] (0,2.59807621135) node[left] {$\sqrt{p_j}$} -- (1.5,2.59807621135);
\draw[dashed] (1.5,0) node[below] {$\sqrt{1-p_j}$} -- (1.5,2.59807621135);

\end{tikzpicture}
\caption{The states $\JorKetA{j}{1}$ and $\JorKetA{j}{0}$ correspond to $1$ and $0$ outcomes of $\BMeas{\BProj{\MeasA}}$, respectively; $\JorKetB{j}{1}$ and $\JorKetB{j}{0}$ correspond to $1$ and $0$ outcomes of $\BMeas{\BProj{\MeasB}}$, respectively.}
\label{fig:jordan}
\end{figure}

By Jordan's lemma, a quantum state $\ket{\phi}$ satisfying $\BProj{\MeasB} \ket{\phi} = \ket{\phi}$ can be written as
\begin{equation*}
\ket{\phi} = {\textstyle\sum_j} \alpha_j \JorKetB{j}{1}
\enspace,
\end{equation*}
where $\alpha_j$ is the amplitude of the state on the Jordan subspace $\Subspace_{j}$. Starting from $\ket{\phi}$, if we alternate the binary projective measurements $\BMeas{\BProj{\MeasA}}$ and $\BMeas{\BProj{\MeasB}}$, then the distribution of the resulting measurement outcomes follows $\MWDist(p_j)$ with probability $\abs{\alpha_j}^2$.

To see why this distribution arises, consider the projective measurement $\Meas{\Jor} = (\Pi^\Jor_j)_j$ that projects onto the Jordan subspaces $\{\Subspace_{j}\}_j$ and returns $j$ as the outcome, i.e., each $\Pi^\Jor_j$ is a projection onto the $\Subspace_{j}$ subspace. Since $\Meas{\Jor}$ acts as the identity within every Jordan subspace $\Subspace_{j}$, a consequence of Jordan's lemma is that $\Meas{\Jor}$ commutes with both $\BMeas{\BProj{\MeasA}}$ and $\BMeas{\BProj{\MeasB}}$. Inserting the measurement $\Meas{\Jor}$ at any point in the sequence of alternating measurements cannot change the earlier measurement outcomes, and the distribution above arises from commuting $\Meas{\Jor}$ to the beginning of the procedure.

With Jordan's lemma in hand, our analysis of the ``state recovery'' procedure in the two-dimensional setting extends to higher dimensions by associating $\BMeas{\BProj{\MeasA}}$ with $\Test$ and $\BMeas{\BProj{\MeasB}}$ with $\mathsf{B}$. Since the procedure's running time is determined solely by the measurement outcomes, we recover the original state $\ket{\psi}$ after $2\NumAlternations$ alternating measurements except with probability
\begin{equation*}
\textstyle{\sum_{j}} |\alpha_j|^2 \cdot 2p_j(1-p_j)(1-2p_j(1-p_j))^{\NumAlternations-1} \leq \textstyle{\frac{1}{\NumAlternations} \sum_{j}} |\alpha_j|^2 = 1/\NumAlternations \enspace.
\end{equation*}

Summarizing, we obtain the following general lemma for binary projective measurements $\MeasA,\MeasB$:
\begin{quote}
    \textbf{Setup:} Fix measurements $\MeasA = \BMeas{\BProj{\MeasA}}$ and $\MeasB = \BMeas{\BProj{\MeasB}}$ and a state $\ket{\psi}$ in the span of $\BProj{\MeasA}$. Apply $\MeasB$ to $\ket{\psi}$ and let $\ket{\phi}$ be the post-measurement state.
    
    \textbf{Alternate:} Starting from $\ket{\phi}$, apply $\MeasA,\MeasB,\MeasA,\MeasB,\ldots$ until $\MeasA$ returns $1$. The procedure requires $O(1)$ measurements in expectation.
\end{quote}
In particular, if $\MeasA$ is our hypothetical $\Test$ measurement, then after the procedure terminates, we recover the state $\ket{\psi}$.

\subsection{State repair}
\label{subsec:tech-state-repair}

Perhaps unsurprisingly, we cannot efficiently implement the measurement $\Test$, and in general we cannot \emph{recover} the original state $\ket{\psi}$\enspace\footnote{One may notice that, for the setting of interactive arguments, $\ket{\psi}$ was generated by an efficient procedure. Nevertheless, there is no efficient procedure to re-generate the particular $\ket{\psi}$ that corresponds to the partial transcript seen so far. This is because $\ket{\psi}$ is the collapsed state leftover after measuring the prover's commitment message, and this may yield different outcomes every time.}. However, our goal is to efficiently extract successful attacker responses, which ``only'' requires that the probability  $\MeasA_r$ for a random $r \gets R$ returns $1$ (the ``success probability'') does not significantly decay with repeated applications. One of our key observations is that we can satisfy this requirement \emph{without} having to recover the original state.

\begin{quote}
    \textbf{Observation:} Restoring the state's \emph{success probability} suffices for extraction.
\end{quote}
We refer to the process of restoring the success probability as \emph{state repair}. Jumping ahead, the repaired state in our state repair procedure may be far in trace distance from the original state $\ket{\psi}$.

Below we explain how to modify the ``state recovery'' procedure from the previous subsection into a ``state repair'' procedure. Informally, we replace $\Test$ with a measurement $\TestRep$ having a relaxed guarantee on post-measurement states: when $\TestRep$ returns $1$, the post-measurement state has the same \emph{success probability} as $\ket{\psi}$.

\parhead{Defining $\TestRep$}
To define a projective measurement $\TestRep$ suitable for performing ``state repair'', it suffices to identify a linear space for which every $\ket{\psi}$ in the space has success probability at least $\Advantage$. We achieve this by identifying a particular operator $E$ with an extremely useful property: any eigenstate of $E$ with eigenvalue $p$ corresponds to a state $\ket{\psi}$ with success probability $p$. We then define $\TestRep$ to be the projection onto the direct sum of eigenspaces of $E$ with eigenvalue $p \geq \Advantage$. 

Our choice of $E$ must somehow capture the probability that a random $\MeasA_r$ for $r \gets R$ returns $1$ when applied to a state $\ket{\psi}$. Thus, a natural place to start is to consider the \emph{purification} of $\MixM(\{\MeasA_r\}_{r \in R})$, i.e., the procedure that applies $\Meas{r}$ for random $r \gets R$. For this, in addition to the original Hilbert space $\RegH$, we need an ancilla register $\RegR$. We initialize this register to a uniform superposition $\UniState{R}$ over the indices $r \in R$. We then define a binary projective measurement $\CProj$ (for ``controlled projection'') that applies $\{\Meas{r} = \BMeas{\SProj{r}}\}_{r}$ controlled on $\RegR$:

\begin{equation*}
\CProj \eqdef \BMeas{\BProj{\CProj}}
\text{ where }
\Pi^\CProj \eqdef {\textstyle\sum_{r \in R}} \ketbra{r}^{\RegR}  \otimes \SProj{r}
\enspace.
\end{equation*}
Letting $\MixM(\{\Meas{r}\}_{r};\ket{\psi})$ denote the application of $\MixM(\{\Meas{r}\}_{r})$ to $\ket{\psi}$, observe that applying $\CProj$ to $\UniState{R}^{\RegR} \otimes \ket{\psi}$ and tracing out $\RegR$ is equivalent to $\MixM(\{\Meas{r}\}_{r};\ket{\psi})$.

We remark that the measurement $\CProj$ represents a ``superposition query'' to the adversary $\Malicious{\ARGProver}(\ket{\psi}, \cdot)$. This is a qualitative departure from the techniques of \cite{Unruh12,DonFMS19}, which only make classical queries to the adversary. Superposition queries have been used in \cite{VidickZ21} in the context of proofs of \emph{quantum} knowledge. We find it interesting that superposition queries also arise in an essential way when extracting only classical knowledge.

We are now ready to define the operator $E$:
\begin{equation*}
E \eqdef \UniProj{R}^{\RegR} \cdot \BProj{\CProj} \cdot \UniProj{R}^{\RegR}
\text{ where $\UniProj{R}^{\RegR}$ denotes $\UniProj{R}^{\RegR}\otimes \Id^{\RegH}$}
\enspace.
\end{equation*}

As desired, any eigenstate of $E$ with positive eigenvalue $p$ is of the form $\UniState{R} \ket{\chi}$ where $\ket{\chi} \in \RegH$ has success probability $p$:
\begin{equation*}
  \Pr\Big[ \MixM(\{ \Meas{r} \};\ket{\chi}) = 1 \Big]
= \norm{\BProj{\CProj} \UniState{R}\ket{\chi}}^2
= (\bra{+_{R}} \otimes \bra{\chi}) E (\UniState{R} \otimes \ket{\chi})
= p
\enspace.
\end{equation*}
We stress that this implication only goes in one direction, as it is \emph{not true} that every state $\ket{\psi}$ with success probability $p$ corresponds to an eigenstate $\UniState{R}\ket{\psi}$ of $E$ with eigenvalue $p$. The precise relationship is summarized in the following observation:
\begin{quote} 
\textbf{Key fact:} For every state $\ket{\psi}$ with success probability $p$, $\UniState{R}\ket{\psi}$ can be written as a \emph{linear combination of eigenstates of $E$}
\begin{equation*}
\UniState{R}\ket{\psi} = {\textstyle\sum_{j} \alpha_{j} \UniState{R}\ket{\chi_j}}
\end{equation*}
where each $\UniState{R}\ket{\chi_j}$ has eigenvalue/success probability $p_j$, and $p = \sum_j \abs{\alpha_j}^2 p_j$.
\end{quote}

We now define $\Pi_{\Advantage}$ as the projector onto the span of eigenstates of $E$ with eigenvalue at least $\Advantage$. Let the corresponding binary-outcome measurement be $\TestRep \eqdef \BMeas{\Pi_{\Advantage}}$. Importantly, $\TestRep$ satisfies the following properties.
\begin{itemize}
    \item \textbf{Property 1: applied to any $2\Advantage$-successful state, $\TestRep$ returns $1$ with probability $\Advantage$.} 
    By the ``key fact'' above, any state $\UniState{R}\ket{\psi}$ where $\ket{\psi}$ has success probability $2\Advantage$ is a linear combination of eigenstates $\sum_j \alpha_j \UniState{R}\ket{\chi_j}$ where $2\Advantage = \sum_j \abs{\alpha_j}^2 p_j$. By Markov's inequality, there must be at least probability mass $\Advantage$ on eigenstates with eigenvalue/success probability at least $\Advantage$.

    \item \textbf{Property 2: when $\TestRep$ returns $1$, the post-measurement state is $\Advantage$-successful.} This follows from the definition of $\Pi_\Advantage$, since any state in the image of $\Pi_\Advantage$ is a linear combination of eigenstates $\UniState{R}\ket{\chi_j}$ where every $\ket{\chi_j}$ has success probability at least $\Advantage$. 

\end{itemize}

\parhead{A state repair procedure}
We now present a state prepare procedure using $\TestRep$. We stress that the following procedure is not yet sufficient to implement an efficient extraction procedure, since we have not specified how to implement $\TestRep$.

\begin{itemize}[noitemsep]
    \item[] Start with state $\UniState{R} \ket{\psi} \in (\RegR,\RegH)$ where $\ket{\psi}$ has success probability $2\Advantage$.
\begin{enumerate}[noitemsep]
    \item \textbf{Initialization.}
    Apply the measurement $\TestRep$ and abort if the outcome is $0$.
    \item \textbf{Measure-and-repair.}
    Repeat the following loop as many times as desired.
    \begin{enumerate}[nolistsep]
    \item (Measure step) Sample a random $r \gets R$ and apply $\MeasA_r$ to $\RegH$ to obtain an outcome $b$. Call this step ``successful'' if $b=1$.
    \item\label[step]{step:ideal-meas} (Repair step) 
   Repair the state by applying $\TestRep,\MeasA_{r},\TestRep,\MeasA_{r},\ldots$ until $\TestRep$ outputs $1$.
   \end{enumerate}
\end{enumerate}
\end{itemize}

Since the state $\ket{\psi}$ at the beginning of the procedure has success probability at least $2\Advantage$, the initialization step aborts with probability at most $1-\Advantage$.

We now analyze the execution of this procedure conditioned on the event that the initialization step \emph{does not abort}. We argue that the procedure can repeatedly iterate the measure-and-repair loop. By construction, the state after any (non-aborting) Initialization step or Repair step is in the span of $\Pi_\Advantage$. Thus, the state at the beginning of the Measure step is always in the span of $\Pi_\Advantage$. Since any state in the span of $\Pi_\Advantage$ is of the form $\UniState{R} \ket{\chi}$ where $\ket{\chi}$ has success probability $\Advantage$, the Measure step is equivalent to an application of $\MixM(\{\MeasA_r\}_r)$ that succeeds with at least $\Advantage$ probability.

\parhead{Recap}
We summarize what our state repair procedure implies for extraction. Suppose we are given a malicious prover $\Malicious{\ARGProver}(\ket{\psi},\cdot)$ for a collapsing interactive protocol who successfully answers a random challenge $r \gets R$ with success probability $\Advantage$. Moreover, assume that we can implement $\TestRep$. Then for any desired $c \in \mathbb{N}$, if the initialization step does not abort, then we can repeat the measure-and-repair iteration $c$ times and achieve the following:
\begin{itemize}[noitemsep]
    \item in each iteration we ask $\Malicious{\ARGProver}$ a random challenge $r \gets R$, and record an accepting transcript $(\Transcript,r,z)$ with probability at least $\Advantage$; and
    \item in expectation, the total number of measurements performed is $O(c)$.
\end{itemize}

While this is promising, we are far from done, because we do not know of a way to efficiently implement $\TestRep$. Hence, in~\cref{subsec:tech-approximate-jordan}, we show how to replace $\TestRep$ with an efficient measurement $\mathsf{ApproxTest}_{\Advantage}$ that \emph{approximates} the behavior of $\TestRep$. While the idea behind $\mathsf{ApproxTest}_{\Advantage}$ is natural, proving that $\mathsf{ApproxTest}_{\Advantage}$ suffices for extraction is the most technically challenging part of this work. 

\subsection{Approximate state repair}
\label{subsec:tech-approximate-jordan}

\parhead{Approximating $\TestRep$}
While we do not know how implement $\TestRep$, we have \emph{already} developed a way to \emph{approximate} $\TestRep$: the alternating measurements technique we used for state repair doubles as a way to \emph{estimate} the success probability! Note that estimating success probability (not repairing the state) was the motivation for alternating measurements in~\cite{MarriottW05,Zhandry20}.

Let $\UniState{R} \ket{\chi_j}$ be an eigenstate of $E = \UniProj{R}^{\RegR}\BProj{\CProj}\UniProj{R}^{\RegR}$ with eigenvalue $p_j$; recall from~\cref{subsec:tech-state-repair} that $\ket{\chi_j}$ has success probability $p_j$.

An important observation is that the eigenspectrum of $E$ corresponds to the decomposition of $(\RegR,\RegH)$ induced by Jordan's lemma for $\BProj{\CProj}$ and $\UniProj{R}^{\RegR}$: any state in the span of $\UniProj{R}^{\RegR}$ that is in the Jordan subspace $\Subspace_{j}$ must be an eigenstate $\UniState{R} \ket{\chi_j}$ of $E$ with eigenvalue $p_j$.

Then, by the analysis in~\cref{subsec:tech-state-recovery}, if we start from $\UniState{R} \ket{\chi_j}$ and apply the binary projective measurements $\CProj = \BMeas{\BProj{\CProj}}$ and $\Meas{\UniState{R}} = \BMeas{\UniProj{R}}$ in an alternating fashion:
\begin{equation*}
\CProj,\Meas{\UniState{R}},\CProj,\Meas{\UniState{R}},\ldots,
\end{equation*}
then the corresponding measurement outcomes $b_1,b_2,b_3,\ldots$ are distributed so that $\mathbf{1}_{b_i = b_{i+1}}$ (the indicator for the event $b_i = b_{i+1}$, where we define $b_0 \eqdef 1$) is an independent Bernoulli random variable with expectation $p_j$ for all $i \geq 0$.

Following~\cite{MarriottW05,Zhandry20}, this yields a simple, \emph{non-projective} procedure $\ApproxTest$: 
\begin{enumerate}[noitemsep]
    \item[] Initial state: $\UniState{R}\ket{\psi}$ for state $\ket{\psi}$ with success probability at least $2\Advantage$. \smallskip
    \item Apply $2t$ measurements $\CProj,\Meas{\UniState{R}},\ldots,\CProj,\Meas{\UniState{R}}$. Denote the binary outcome of the $i$-th measurement by $b_i$ and additionally set $b_0 \eqdef 1$.
    \item Compute $p \eqdef \frac{1}{2t} \cdot |\{ i \in \{1,\ldots,2t\} : b_{i-1} = b_{i} \}|$ and output $1$ if $p \geq \Advantage$.
\end{enumerate}

To analyze the distribution of outcomes from applying $\ApproxTest$ to an arbitrary state of the form $\UniState{R} \ket{\psi}$, we employ the method from~\cref{subsec:tech-state-recovery} of projecting onto the Jordan subspaces $\{\Subspace_{j}\}_j$ for the projectors $\BProj{\CProj}$ and $\UniProj{R}$. Since any state $\UniState{R} \ket{\psi}$ can be written as a linear combination $\sum_j \alpha_j \UniState{R} \ket{\chi_j}$ of eigenstates of $E$, the result of applying $\ApproxTest$ to $\UniState{R} \ket{\psi}$ can be described as follows, where $\TestRep$ is included for comparison:
\begin{itemize}
    \item $\TestRep$: Sample $j$ with probability $\abs{\alpha_j}^2$, and then return $1$ if $p_j \geq \Advantage$ and $0$ otherwise.
    \item $\ApproxTest$: Sample $j$ with probability $\abs{\alpha_j}^2$; flip $2t$ independent Bernoulli random variables with parameter $p_j$; let $p$ be the fraction of flips that return $1$; output $1$ if $p \geq \Advantage$ and $0$ otherwise.
\end{itemize}
Thus, we have from~\cref{subsec:tech-state-repair} a working extraction procedure based on $\TestRep$, and now a way to efficiently approximate $\TestRep$ to any desired precision using $\ApproxTest$. However, turning this intuition into a working extraction procedure requires overcoming a number of technical challenges, stemming from the fact that $\ApproxTest$ as defined above is \emph{not} a projective measurement.

\parhead{Challenge: $\ApproxTest$ is not projective}
In~\cref{subsec:tech-state-recovery} we claimed that if a state $\ket{\psi}$ initially in the span of some projector $\BProj{\MeasA}$ is disturbed by an binary-outcome measurement $\MeasB$, then by performing alternating measurements, we can return our state to the span of $\BProj{\MeasA}$ in $2T$ measurements except with probability $1/T$. It is not clear that such a statement holds if $\MeasA = \BMeas{\BProj{\MeasA}}$ is replaced by a \emph{non-projective} measurement.

Concretely, we need to analyze the behavior of the alternating measurement procedure
\begin{equation*}
\ApproxTest,\MeasA_{r},\ApproxTest,\MeasA_{r},\ldots
\end{equation*}
\emph{where $\ApproxTest$ itself is an alternating measurements procedure}, i.e., $\ApproxTest$ runs
\begin{equation*}
\CProj,\Meas{\UniState{R}},\CProj,\Meas{\UniState{R}},\ldots
\enspace.
\end{equation*}
The core technical challenge is to prove that the guarantees of alternating measurements used in~\cref{subsec:tech-state-repair} extend to ``nested'' alternating measurements.

\parhead{Can we appeal to trace distance?}
One might hope to show that for large $t$, the post-measurement states of $\ApproxTest$ and $\TestRep$ are close. If $\ApproxTest \ket{\psi}$ were sufficiently close in trace distance to $\TestRep \ket{\psi}$ for all $\ket{\psi}$, then we could show that any property of the procedure $\TestRep,\MeasA_{r},\TestRep,\MeasA_{r},\ldots$ still applies if we swap out $\TestRep$ for $\ApproxTest$, up to a small loss.

Unfortunately, a simple example illustrates why such a claim about the trace distance is false. Suppose we have an eigenstate $\UniState{R}\ket{\chi_j}$ of the operator $E$ with eigenvalue $p_j = \Advantage$. Then since $\TestRep$ projects onto eigenspaces of $E$ with eigenvalue $\geq \Advantage$, applying $\TestRep$ to this state returns $1$ with probability $1$. However, applying $\ApproxTest$ returns $1$ with essentially $1/2$ probability, since it performs $\Advantage$-weighted coin flips and only accepts if the fraction of $1$'s is at least $\Advantage$.

\parhead{Expanding the Hilbert space}
Since a trace distance argument is unlikely to work, the next idea is to simply force $\ApproxTest$ to be projective by expanding the Hilbert space. The hope is that by making the measurement projective, we regain our ability to apply Jordan's lemma. Specifically, we introduce $2t$-qubit ancilla registers $\RegL$ to store the $2t$ outcomes of $\CProj$ and $\Meas{\UniState{R}}$, which we perform \emph{coherently}, meaning that instead of actually performing the measurements, we apply corresponding unitaries to CNOT the measurement results onto the ancilla registers $\RegL$. To ensure the measurement is projective, we must also uncompute all the (coherent applications of) $\CProj$ and $\Meas{\UniState{R}}$ once we obtain the probability estimate $p$.

\parhead{Technical challenge: $\ApproxTest$ is only meaningful if $\RegL$ is $\ket{0^{2t}}$} Unfortunately, expanding the Hilbert space introduces a new problem. If $\ApproxTest$ computes its estimate of $p$ using a $2t$-qubit ancilla register $\RegL$, then we have to ensure the register $\RegL$ is set to $\ket{0^{2t}}$, or else the estimate of $p$, computed based on the contents of the $\RegL$ register, may be meaningless. A natural idea would be to ensure that, before any application of $\ApproxTest$, we trace out the potentially non-zero registers $\RegL$ and manually reset them to $\ket{0^{2t}}$. However, doing this is equivalent to performing the original non-projective version of $\ApproxTest$, and we would be back where we started.

\parhead{Resolution: project $\RegL$ onto $\ket{0^{2t}}$}
Instead we modify the measurement $\MeasA_{r}$ (which originally acts as identity on the $\RegL$ registers) to additionally project $\RegL$ onto $\ket{0^{2t}}$. This modified measurement $\MeasA_{r,b}$ returns $1$ if and only if $\MeasA_{r}$ returns $b$ \emph{and} the binary projective measurement of $\RegL$ onto $\ket{0^{2t}}$ returns $1$; in particular, $\MeasA_{r,b}$ is still a binary projective measurement. Proving that the state is repaired after the \emph{projective} version of $\ApproxTest$ returns $1$ requires a very careful analysis of the properties of the Jordan decomposition induced by (projective) $\ApproxTest$ and $\MeasA_{r,b}$. The analysis of this procedure is the most technical component of the paper; see \cref{sec:state-repair} for details.

\subsection{Quantum strategies for repeated games}

Our quantum rewinding techniques can be cast in the language of \emph{single-player games}, i.e., a referee asks a player a random question $r \gets R$, the player responds with some $z$, and wins if $f(r,z) = 1$ for some predicate $f$. Mapped onto this setting, the quantum rewinding task is to transform any efficient quantum strategy for winning the game once into an efficient strategy that can win in many rounds in an $n$-fold sequential repetition of this game, where in each repetition the referee \emph{only measures whether the player has won}. Importantly, we are only given one copy of the quantum state used by the one-time strategy.

In the context of rewinding, we set $f(r,z) \eqdef V(\Transcript,r,z)$ to be the verifier predicate with partial transcript $\Transcript$. The strategy of the prover in the last round of the protocol is then an efficient strategy for the one-time game. To obtain multiple accepting transcripts, a rewinding extractor plays the sequential repetition of the game. Note that by measuring $\RegZ$ in the computational basis if the player has won, the extractor obtains an accepting response $z$; collapsing ensures that this additional measurement is not detectable by an efficient strategy.

This gives a conceptually simple characterization of the quantum rewinding task, which may be of independent interest. In the body of the paper, we develop general techniques that apply to any single-player game (see \cref{sec:alt-proj-new}).

\subsection{Discussion: is collapsing necessary for Kilian's protocol?}

Since collision-resistant hash functions (CRHFs) suffice in the classical setting, a natural question is whether Kilian's protocol (in its original formulation using Merkle trees) is post-quantum secure when instantiated with any post-quantum CRHF. We do not know the answer, but believe that the existing evidence points to collision resistance being \emph{insufficient} for Kilian's protocol.

Ambainis et al.\ \cite{AmbainisRU14} give a counter-example showing that, in general, collision resistance alone is likely not enough for rewinding in interactive protocols. The counter-example works by giving a construction of an \emph{equivocal} hash function.\footnote{The terminology ``equivocal'' is due to~\cite{amos2020one}.} This is a hash function that is collision resistant, but where it is possible to break the security of the hash function as a commitment scheme. For example, it is possible to send a hash image $y$, and then upon receiving an arbitrary prefix $z$, ``open'' that image to a pre-image $x$ of $y$ with prefix $z$. Such equivocal hash functions do not exist classically, due to a rewinding argument, but Ambainis et al.~\cite{AmbainisRU14} show how to construct them relative to a quantum oracle. Amos et al.~\cite{amos2020one} later give a construction relative to a classical oracle.

While Ambainis et al.\ use equivocal hash functions to give unsound interactive proofs, the results do not immediately apply to the case of Kilian's protocol. This is because Merkle trees do not necessarily preserve equivocality of the component hash function. In particular, equivocating Merkle trees would seem to require equivocating the underlying hash function on either the left half or the right half of the input. On the other hand, only a very short prefix can be equivocated by the existing works.\footnote{Ambanis et al. allow for a richer class of equivocations than just prefixes, but they must still be short relative to the input length.}

Nevertheless, we observe that a slight variant of Merkle trees \emph{does} preserve the equivocality of the underlying hash function. Namely, if each node is obtained by hashing the children together with an arbitrarily long auxiliary string. By setting the length of the auxiliary strings sufficiently long, one can equivocate on a prefix long enough to arbitrarily choose the child nodes. This allows for full equivocality of Merkle trees, while still preserving collision resistance. More generally, it yields a \emph{vector commitment} that is collision resistant, but equivocal and therefore insufficient for the post-quantum security of Kilian's protocol.

We leave as an interesting open question whether Kilian's protocol instantiated with vanilla Merkle trees using a post-quantum CRHF is sufficient for post-quantum security. We note, however, that if Kilian's protocol instantiated with a CRHF is \emph{not} post-quantum secure, then it means the CRHF is not collapsing. As shown by Zhandry~\cite{Zhandry19-lightning}, such a CRHF would yield strong cryptographic objects, namely ``quantum lightning'', which have no known instantiations under well-studied assumptions.\footnote{More precisely, Zhandry~\cite{Zhandry19-lightning} shows that non-collapsing CRHFs imply \emph{infinitely-often secure} quantum lightning, a slightly weaker notion.}

\iffocs
\clearpage
\setcounter{tocdepth}{2}
\tableofcontents
\clearpage
\fi

\doclearpage
\section{Preliminaries}
\label{sec:preliminaries}

The security parameter is denoted by $\secp$. A function $f \colon \Naturals \rightarrow [0,1]$ is \emph{negligible}, denoted $f(\secp) = \negl(\secp)$, if it decreases faster than the inverse of any polynomial. A probability is \emph{overwhelming} if is at least $1- \negl(\secp)$ for a negligible function $\negl(\secp)$. For any positive integer $n$, let $[n] \coloneqq \{1,2,\dots,n\}$. For a set $R$, we write $r \gets R$ to denote a uniformly random sample $r$ drawn from $R$.

\subsection{Concentration inequalities}
\label{sec:concentration-inequalities}

We denote by $\Bin(n,p)$ the binomial distribution with $n$ trials and success probability $p$ (sum of $n$ independent Bernoullis with parameter $p$). We use the following Chernoff bounds.

\begin{proposition}[additive Chernoff bound]
\label{prop:chernoff}
For $\delta,\epsilon > 0$, define $\ChernoffN_{\epsilon,\delta} \eqdef \frac{\log(1/2\delta)}{2\epsilon^2}$. If $n \geq \ChernoffN_{\epsilon,\delta}$ then
\begin{equation*}
\Pr_{X \gets \Bin(n,p)}\left[\, p-\epsilon \leq \frac{X}{n} \leq p+\epsilon\, \right]
\geq 1-\delta
\enspace.
\end{equation*}
\end{proposition}

\begin{proposition}[multiplicative Chernoff bound]
\label{prop:chernoff-2}
Let $x_1,\ldots,x_N \in \Bits$ and define $\mu \eqdef \frac{K}{N} \sum_{i=1}^{N} x_{i}$. Let $X_1,\ldots,X_K$ be independent uniformly random samples from $x_1,\ldots,x_N$. Then
\begin{equation*}
\Pr\left[\, \sum_{i=1}^{K} X_K \geq (1+\delta)\mu \,\right]
\leq e^{-\delta^2 \mu/3}
\enspace.
\end{equation*}
\end{proposition}

\subsection{Quantum preliminaries and notation}
\label{sec:quantum-prelims}

\parhead{Quantum information}
A (pure) \emph{quantum state} is a vector $\ket{\psi}$ in a complex Hilbert space $\RegH$ with $\norm{\ket{\psi}} = 1$; in this work, $\RegH$ is always finite-dimensional. We denote by $\Hermitians{\RegH}$ the space of Hermitian operators on $\RegH$. A \emph{density matrix} is a Hermitian operator $\DMatrix \in \Hermitians{\RegH}$ with $\Tr(\DMatrix) = 1$. A density matrix represents a probabilistic mixture of pure states (a mixed state); the density matrix corresponding to the pure state $\ket{\psi}$ is $\ketbra{\psi}$. Typically we divide a Hilbert space into \emph{registers}, e.g. $\RegH = \RegH_1 \otimes \RegH_2$. We sometimes write, e.g., $\DMatrix^{\RegH_1}$ to specify that $\DMatrix \in \Hermitians{\RegH_1}$.

A unitary operation is represented by a complex matrix $U$ such that $U \contra{U} = \Id$. The operation $U$ transforms the pure state $\ket{\psi}$ to the pure state $U \ket{\psi}$, and the density matrix $\DMatrix$ to the density matrix $U \DMatrix \contra{U}$.

A \emph{projector} $\Projector$ is a Hermitian operator ($\contra{\Projector} = \Projector$) such that $\Projector^2 = \Projector$. A \emph{projective measurement} is a collection of projectors $\ProjMeasurement = (\Projector_i)_{i \in S}$ such that $\sum_{i \in S} \Projector_i = \Id$. This implies that $\Projector_i \Projector_j = 0$ for distinct $i$ and $j$ in $S$. The application of a projective measurement to a pure state $\ket{\psi}$ yields outcome $i \in S$ with probability $p_i = \norm{\Projector_i \ket{\psi}}^2$; in this case the post-measurement state is $\ket{\psi_i} = \Projector_i \ket{\psi}/\sqrt{p_i}$. We will sometimes refer to the post-measurement state $\Projector_i \ket{\psi}/\sqrt{p_i}$ as the result of applying $\ProjMeasurement = (\Projector_i)_{i \in S}$ to $\ket{\psi}$ and \emph{post-selecting} (i.e., conditioning) on outcome $i$. A state $\ket{\psi}$ is an \emph{eigenstate} of $\ProjMeasurement$ if it is an eigenstate of every $\Projector_i$.

A two-outcome projective measurement is called a \emph{binary projective measurement}, and is written as $\ProjMeasurement = \BMeas{\Projector}$, where $\Projector$ is associated with the outcome $1$, and $\Id - \Projector$ with the outcome $0$.

General (non-unitary) evolution of a quantum state can be represented via a \emph{completely-positive trace-preserving (CPTP)} map $T \colon \Hermitians{\RegH} \to \Hermitians{\RegH'}$. We omit the precise definition of these maps in this work; we will only use the facts that they are trace-preserving (for every $\DMatrix \in \Hermitians{\RegH}$ it holds that $\Tr(T(\DMatrix)) = \Tr(\DMatrix)$) and linear.

For every CPTP map $T \colon \Hermitians{\RegH} \to \Hermitians{\RegH}$ there exists a \emph{unitary dilation} $U$ that operates on an expanded Hilbert space $\RegH \otimes \RegK$, so that $T(\DMatrix) = \Tr_{\RegK}(U (\rho \otimes \ketbra{0}^{\RegK}) U^{\dagger})$. This is not necessarily unique; however, if $T$ is described as a circuit then there is a dilation $\CPTPUnitary$ represented by a circuit of size $O(|T|)$.

For Hilbert spaces $\RegA,\RegB$ the \emph{partial trace} over $\RegB$ is the unique CPTP map $\Tr_{\RegB} \colon \Hermitians{\RegA \otimes \RegB} \to \Hermitians{\RegA}$ such that $\Tr_{\RegB}(\DMatrix_A \otimes \DMatrix_B) = \Tr(\DMatrix_B) \DMatrix_A$ for every $\DMatrix_A \in \Hermitians{\RegA}$ and $\DMatrix_B \in \Hermitians{\RegB}$.

A \emph{general measurement} is a CPTP map $\Measurement \colon \Hermitians{\RegH} \to \Hermitians{\RegH \otimes \RegO}$, where $\RegO$ is an ancilla register holding a classical outcome. Specifically, given measurement operators $\{ M_{i} \}_{i=1}^{N}$ such that $\sum_{i=1}^{N} M_{i} M_{i}^{\dagger} = \Id$ and a basis $\{ \ket{i} \}_{i=1}^{N}$ for $\RegO$, $\Measurement(\DMatrix) \eqdef \sum_{i=1}^{N} (M_{i} \DMatrix M_{i}^{\dagger} \otimes \ketbra{i}^{\RegO})$. We will sometimes implicitly discard the outcome register. A projective measurement is simply a general measurement where the $M_{i}$ are projectors. A measurement induces a probability distribution over its outcomes given by $\Pr[i] = \Tr(\ketbra{i}^{\RegO} \Measurement(\DMatrix))$; we denote sampling from this distribution by $i \gets \Measurement(\DMatrix)$.

The \emph{trace distance} between states $\DMatrix,\DMatrixB$, denoted $d(\DMatrix,\DMatrixB)$, is defined as $\frac{1}{2}\Tr( \sqrt{(\DMatrix - \DMatrixB)^2})$. The trace distance is contractive under CPTP maps, i.e. for any CPTP map $T$, $d(T(\DMatrix),T(\DMatrixB)) \leq d(\DMatrix,\DMatrixB)$. It follows that for any measurement $\Measurement$, the statistical distance between the distributions $\Measurement(\DMatrix)$ and $\Measurement(\DMatrixB)$ is bounded by $d(\DMatrix,\DMatrixB)$. We have the following \emph{gentle measurement lemma}, which bounds how much a state is disturbed by applying a measurement whose outcome is almost certain.

\begin{lemma}[Gentle Measurement~\cite{Winter99}]
\label{lemma:gentle-measurement}
    Let $\DMatrix \in \Hermitians{\RegH}$ and $\ProjMeasurement = \BMeas{\Projector}$ be a binary projective measurement on $\RegH$ such that $\Tr(\Projector \DMatrix) \geq 1-\delta$. Let $\DMatrix'$ be the state after applying $\ProjMeasurement$ to $\DMatrix$ and post-selecting on obtaining outcome $1$. Then
    \[d(\DMatrix,\DMatrix') < 2\sqrt{\delta}.\]
\end{lemma}

\parhead{Quantum algorithms}
In this work, a \emph{quantum adversary} is a family of quantum circuits $\{ \Adversary_{\secp} \}_{\secp \in \Naturals}$ represented classically using some standard universal gate set. A quantum adversary is \emph{polynomial-size} if there exists a polynomial $p$ and $\secp_0 \in \Naturals$ such that for all $\secp > \secp_0$ it holds that $|\Adversary_{\secp}| \leq p(\secp)$ (i.e., quantum adversaries have classical non-uniform advice).

In this work we refer to the \emph{expected running time} of quantum algorithms. This means that there is a classical control algorithm that applies quantum circuits of a fixed size and decides whether to terminate based on the classical outputs of those circuits. The expected running time is then the expected number of unit operations, classical or quantum, applied during this execution.

\parhead{Black-box access}
A circuit $C$ with black-box access to a unitary $U$, denoted $C^{U}$, is a standard quantum circuit with special gates that act as $U$ and $U^{\dagger}$. We also use $C^{T}$ to denote black-box access to a map $T$, which we interpret as $C^{\CPTPUnitary}$ for a unitary dilation $\CPTPUnitary$ of $T$; all of our results will be independent of the choice of dilation. This allows, for example, the ``partial application'' of a projective measurement, and the implementation of a general measurement via a projective measurement on a larger space.

\subsection{Jordan's lemma}
\label{subsec:jordan}

We state Jordan's lemma and, for completeness, provide a proof that roughly follows \cite{Regev06-XXX}.

\begin{lemma}[\cite{Jordan75}]
\label{lemma:jordan}
For any two Hermitian projectors $\BProj{\MeasA}$ and $\BProj{\MeasB}$ on a Hilbert space $\RegH$, there exists an orthogonal decomposition of $\RegH = \bigoplus_j \Subspace_{j}$ into one-dimensional and two-dimensional subspaces $\{\Subspace_{j}\}_{j}$ (the \emph{Jordan subspaces}), where each $\Subspace_{j}$ is invariant under both $\BProj{\MeasA}$ and $\BProj{\MeasB}$. Moreover:
\begin{itemize}[noitemsep]
\item in each one-dimensional space, $\BProj{\MeasA}$ and $\BProj{\MeasB}$ act as identity or rank-zero projectors; and
\item in each two-dimensional subspace $\Subspace_{j}$, $\BProj{\MeasA}$ and $\BProj{\MeasB}$ are rank-one projectors. In particular, there exist distinct orthogonal bases $\{\JorKetA{j}{1},\JorKetA{j}{0}\}$ and $\{\JorKetB{j}{1},\JorKetB{j}{0}\}$ for $\Subspace_{j}$ such that $\BProj{\MeasA}$ projects onto $\JorKetA{j}{1}$ and $\BProj{\MeasB}$ projects onto $\JorKetB{j}{1}$.
\end{itemize}
\end{lemma}

\begin{proof}
Since $\BProj{\MeasA}$ and $\BProj{\MeasB}$ are both Hermitian, their sum $\BProj{\MeasA} + \BProj{\MeasB}$ is also Hermitian. By the spectral theorem for Hermitian matrices, it follows that the eigenvectors of $\BProj{\MeasA} + \BProj{\MeasB}$ span $\RegH$. Let $\ket{\psi}$ be an eigenvector with eigenvalue $p$ (i.e., $\BProj{\MeasA}\ket{\psi} + \BProj{\MeasB}\ket{\psi} = p\ket{\psi}$). There are two cases to consider.

If $\BProj{\MeasA}\ket{\psi}$ lies in $\spanset(\ket{\psi})$, then $\BProj{\MeasB}\ket{\psi}$ must also be in $\spanset(\ket{\psi})$, so $\spanset(\ket{\psi})$ is a one-dimensional subspace invariant under both $\BProj{\MeasA}$ and $\BProj{\MeasB}$. Since $\BProj{\MeasA}$ and $\BProj{\MeasB}$ are projectors, their eigenvalues are $0$ or $1$, so in $\spanset(\ket{\psi})$ they act as identity or rank-zero projectors.

If $\BProj{\MeasA}\ket{\psi}$ does not lie in $\spanset(\ket{\psi})$, then $\spanset(\ket{\psi},\BProj{\MeasA}\ket{\psi})$ is a two-dimensional subspace. This subspace is invariant under $\BProj{\MeasA}$, which acts as a projector onto $\JorKetA{j}{1} \coloneqq \BProj{\MeasA}\ket{\psi}$. Moreover, this subspace can be written as $\spanset(\ket{\psi},\BProj{\MeasB}\ket{\psi})$, and by an identical argument, $\BProj{\MeasB}$ projects this subspace onto $\JorKetB{j}{1} \coloneqq \BProj{\MeasB}\ket{\psi}$.

By setting $\JorKetA{j}{0} \coloneqq \Subspace_{j} \cap \ker(\BProj{\MeasA})$ (i.e. the state in $\Subspace_{j}$ orthogonal to $\JorKetA{j}{1})$) and $\JorKetB{j}{0} \coloneqq \Subspace_{j} \cap \ker(\BProj{\MeasB})$, we obtain two different orthogonal bases $\{\JorKetA{j}{1},\JorKetA{j}{0}\}$ and $\{\JorKetB{j}{1},\JorKetB{j}{0}\}$ for $\Subspace_{j}$ where $\BProj{\MeasA}$ projects onto $\JorKetA{j}{1}$ and $\BProj{\MeasB}$ projects onto $\JorKetB{j}{1}$.
\end{proof}

\subsection{Interactive arguments}
\label{sec:interactive-arguments}

For interactive classical algorithm $\ARGVerifier$ and interactive (potentially) quantum circuit $A$, we denote by $\Interact{A(\ket{\psi})}{\ARGVerifier}$ the random variable corresponding to the output of $\ARGVerifier$ when interacting with $A(\ket{\psi})$; note that since $\ARGVerifier$ is classical, the communication in this interaction is also classical. For a general formal treatment of interactive quantum circuits, see \cite{VidickW16}.

\begin{definition}
A (post-quantum) \emph{interactive argument} for a relation $\Relation$ with soundness $\SoundnessError$ is a pair of interactive classical polynomial-time algorithms $\ARGSystem = (\ARGProver,\ARGVerifier)$ such that the following holds.
\begin{itemize}

  \item \textbf{Completeness.}
  For every $\secp \in \Naturals$ and $(\Instance,\Witness) \in \Relation$, $\Pr[\Interact{\ARGProver(1^{\secp},\Instance,\Witness)}{\ARGVerifier(1^{\secp},\Instance)} = 1] = 1$.
	
  \item \textbf{Soundness.}
  For every $\secp \in \Naturals$, $\Instance \notin \Language(\Relation)$, and polynomial-size interactive quantum circuit $\Malicious{\ARGProver}$,
\begin{equation*}  
\Pr[\Interact{\Malicious{\ARGProver}}{\ARGVerifier(1^{\secp},\Instance)} = 1]
\leq \SoundnessError(\secp)
\enspace.
\end{equation*}
	
\end{itemize}
We say that $\ARGSystem$ is \textbf{succinct} if the total amount of communication between $\ARGProver$ and $\ARGVerifier$ is at most $c(\secp,\log |\Instance|)$ for some fixed polynomial $c$.

In this work a \emph{round} is a back-and-forth interaction consisting of a verifier message followed by a prover message.
\end{definition}

We also consider interactive arguments that satisfy the stronger property of \emph{knowledge soundness}. Below we write $\ARGExtractor^{\Malicious{\ARGProver}}$ for an extractor with ``black-box'' access to $\Malicious{\ARGProver}$; we will give this a precise meaning shortly. Our definition loosely follows that of~\cite{Unruh12}.

\begin{definition}
\label{def:post-quantum-knowledge}
$\ARGSystem = (\ARGProver,\ARGVerifier)$ has \emph{knowledge soundness} with knowledge error $\KnowledgeError$ if there exists an expected polynomial time quantum extractor $\ARGExtractor$ such that for every polynomial-size interactive quantum circuit $\Malicious{\ARGProver}$, quantum state $\ket{\psi}$, $\secp \in \Naturals$, instance $\Instance$, and parameter $\Advantage(\secp) \leq \Pr[\Interact{\Malicious{\ARGProver}(x,\ket{\psi})}{\ARGVerifier(1^{\secp},\Instance)} = 1]$ the following holds: 
\[\Pr\big[(\Instance,\Witness) \in \Relation \mid \Witness \gets \ARGExtractor^{\Malicious{\ARGProver}(x;\ket{\psi})}(1^{\secp},\Instance,1^{1/\Advantage})\big]
= \Omega(\Advantage(\secp) - \KnowledgeError)~.\]
\end{definition}

We describe the differences between our definition and the definition of quantum proofs of knowledge given in~\cite{Unruh12}.
\begin{itemize}
    \item Our definition asks that the extractor succeed with probability \emph{linear} in $(\Advantage(\secp)-\KnowledgeError)$, whereas Unruh's definition only requires the extractor's success probability be $(\Advantage(\secp)-\KnowledgeError)^d/p(\secp)$ for a constant $d \in \Naturals$ and polynomial $p$.
    \item Our definition is incomparable to Unruh's definition when $\ket{\psi}$ is a general quantum state, since we require that the extractor be given as input a lower bound $\Advantage$ on the success probability of the adversary. This arises due to a technical requirement in our security proof. 
    \item When $\ket{\psi}$ is a computational basis state (or any other efficiently-constructible state), our definition is stronger than Unruh's definition since in this case the extractor can compute for itself a lower bound on the success probability of the adversary by simply running the adversary many times (independently, from the \emph{beginning} of the protocol).
\end{itemize}

To define black-box access to $\Malicious{\ARGProver}$, we will need to consider in more detail how an interactive quantum circuit is specified.

\begin{definition}[Interactive quantum circuits]
	A $\NumRounds$-round interactive quantum circuit $A$ is a sequence of unitary quantum circuits $(\UPrvRound{1},\dots,\UPrvRound{\NumRounds})$ where $\UPrvRound{i}$ operates on registers $(\RegInt,\RegChal_{i},\RegResp_{i})$.
	
	The \emph{size} of an interactive quantum circuit is the sum of the sizes of the circuits implementing $\UPrvRound{1},\ldots,\UPrvRound{\NumRounds}$.
\end{definition}

Let $\Malicious{\ARGProver} \eqdef (\UPrvRound{1},\ldots,\UPrvRound{\NumRounds})$; then $\ARGExtractor^{\Malicious{\ARGProver}}$ is a quantum circuit with special gates corresponding to $\UPrvRound{i}(r)$ and $\contra{(\UPrvRound{i}(r))}$ for $i \in [\NumRounds]$.

The requirement that the $\UPrvRound{i}$ be unitary is without loss of generality, in the sense that any quantum circuit not of this form can be ``purified'' into a circuit of this form which is only a constant factor larger with the same observable behavior. Using this formulation, we can sample the random variable $\Interact{\Malicious{\ARGProver}}{\ARGVerifier}$ equivalently as:
\begin{enumerate}[noitemsep]
	\item Initialize the register $\RegInt$ to $\ket{\psi}$, and $\Transcript \eqdef ()$.
	\item For $i = 1,\ldots,\NumRounds$,
	\begin{enumerate}[noitemsep]
		\item Sample $r_i \gets R_i$. Initialize the $\RegChal_{i}$ register to $\ket{r_i}$.
		\item Apply unitary $\UPrvRound{i}$ to $(\RegInt,\RegChal_{i},\RegResp_{i})$.
		\item Measure $\RegResp_{i}$ in the computational basis to obtain response $\Response_i$. Append $(r_i,\Response_i)$ to $\Transcript$.
	\end{enumerate}
	\item Return the output of $\ARGVerifier(\Transcript)$.
\end{enumerate}
In particular, the interaction is \emph{public coin}. Note again that we restrict the operation of $A$ in each round to be unitary except for the measurement of $\RegResp_{i}$ in the computational basis.  

\subsection{Collapsing hash functions}
\label{sec:hash-functions}

Let $\HashFamily = \{\HashDistribution_{\secp}\}_{\secp \in \Naturals}$ be such that each $\HashDistribution_{\secp}$ is a distribution over functions $\HashFunction \colon \Bits^{\hinlen(\secp)} \to \Bits^{\houtlen(\secp)}$.

\begin{definition}
\label{def:crh}
$\HashFamily$ is \emph{post-quantum collision resistant} if for every polynomial-size quantum adversary $\Adversary$,
\begin{equation*}
\Pr
\left[
\begin{array}{c}
x \neq x' \; \wedge \\
\HashFunction(x) = \HashFunction(x')  
\end{array}
\middle\vert
\begin{array}{r}
\HashFunction \gets \HashDistribution_{\secp} \\
(x,x') \gets \Adversary(\HashFunction)
\end{array}
\right]
= \negl(\secp)
\enspace.
\end{equation*}
\end{definition}

\begin{definition}
\label{def:collapsing}
$\HashFamily$ is \emph{collapsing} \cite{Unruh16-eurocrypt} if for every security parameter $\secp$ and polynomial-size quantum adversary $\Adversary$,
\begin{equation*}
\Big|
\Pr[\HCollapsingExp{0}{\secp}{\Adversary} = 1]
-
\Pr[\HCollapsingExp{1}{\secp}{\Adversary} = 1]
\Big|
\leq \negl(\secp)
\enspace.
\end{equation*}
For $b \in \Bits$ the experiment $\HCollapsingExp{b}{\secp}{\Adversary}$ is defined as follows:
\begin{enumerate}[noitemsep]
  \item The challenger samples $\HashFunction \gets \HashDistribution_{\secp}$ and sends $\HashFunction$ to $\Adversary$.
  \item $\Adversary$ replies with a (classical) binary string $y \in \Bits^{\houtlen(\secp)}$ and a $\hinlen(\secp)$-qubit quantum state on registers $\RegX$. (The requirement that $y$ is classical can be enforced by having the challenger immediately measure these registers upon receiving them.)
  \item The challenger computes $\HashFunction$ in superposition on the $\hinlen(\secp)$-qubit quantum state, and measures the bit indicating whether the output of $\HashFunction$ equals $y$. If $\HashFunction$ does not equal $y$, the challenger aborts and outputs $\bot$.
  \item If $b = 0$, the challenger does nothing. If $b = 1$, the challenger measures the $\hinlen(\secp)$-qubit state in the standard basis.
  \item The challenger returns contents of the registers $\RegX$ to $\Adversary$.
  \item $\Adversary$ outputs a bit $b'$, which is the output of the experiment.
\end{enumerate}
\end{definition}

\begin{claim}[\cite{Unruh16-eurocrypt}]
\label{claim:collapsing-crhf}
If $\HashFamily$ is collapsing then $\HashFamily$ is collision resistant.
\end{claim}

\begin{proof}
A proof can be found in \cite[Lemma 25]{Unruh16-eurocrypt}, but for convenience we include a proof here.

Let $\Adversary$ be an adversary that breaks collision resistance of $\HashFamily$ with probability at least $\Advantage(\secp)$. We construct an adversary $\Adversary'$ that breaks collapsing of $\HashFamily$ with probability at least $\Advantage(\secp)/2$.

The adversary $\Adversary'$ works as follows. First, given as input $\HashFunction \gets \HashDistribution_{\secp}$, $\Adversary'$ computes $(x,x') \gets \Adversary(\HashFunction)$. If $(x,x')$ is not a valid collision (they are equal or they map to different outputs under $\HashFunction$) then $\Adversary'$ sends to the challenger an arbitrary classical bitstring $y$ and an arbitrary quantum state on register $\RegX$, and then outputs $0$ at the conclusion of the experiment. If $(x,x')$ is a valid collision (they are distinct and they map to the same ouput under $\HashFunction$), then $\Adversary'$ sends $y \eqdef \HashFunction(x)$ and the quantum state $\ket{\psi} \coloneqq \frac{1}{\sqrt{2}}(\ket{x} + \ket{x'})$ on register $\RegX$; when the challenger returns the contents of $\RegX$, $\Adversary'$ applies the binary projective measurement $\ProjMeasurement = \BMeas{\ketbra{\psi}}$, and outputs the measurement outcome $b$. 

In $\HCollapsingExp{0}{\secp}{\Adversary'}$, the adversary $\Adversary'$ outputs $1$ with probability at least $\Advantage(\secp)$, since as long as $\Adversary$ outputs a valid collision $(x,x')$, the measurement $\ProjMeasurement$ is applied to $\frac{1}{\sqrt{2}}(\ket{x} + \ket{x'})$ and must return $1$. In $\HCollapsingExp{1}{\secp}{\Adversary'}$, the adversary $\Adversary'$ outputs $1$ with probability at most $\Advantage(\secp)/2$, since as long as $\Adversary$ outputs a valid collision $(x,x')$, the measurement $\ProjMeasurement$ is applied to either $\ket{x}$ or $\ket{x'}$, and thus returns $1$ with probability at most $1/2$. The overall difference in the two probabilities is $\Advantage(\secp)/2$. 
\end{proof}

\subsection{Collapsing protocols}

\begin{definition}[\cite{Unruh16-eurocrypt,LiuZ19,DonFMS19}]
\label{def:collapsing-protocol}
We say that a protocol is collapsing if for every polynomial-size interactive quantum adversary $\Malicious{\ARGProver}$ and polynomial-size quantum distinguisher $\Adversary$,
\begin{equation*}
\Big| \Pr[\ProtocolCollapseExp(0,\Malicious{\ARGProver},\Adversary) = 1]
- \Pr[\ProtocolCollapseExp(1,\Malicious{\ARGProver},\Adversary) = 1] \Big|
\leq \negl(\secp) \enspace.
\end{equation*}
For $b \in \Bits$, the experiment $\ProtocolCollapseExp(b,\Malicious{\ARGProver},\Adversary)$ is defined as follows:
\begin{enumerate}[noitemsep]
	\item The challenger simulates $\Interact{\Malicious{\ARGProver}}{\ARGVerifier}$, stopping just before the measurement of $\RegResp_{\NumRounds}$. Let $\Transcript' = (r_1,\LastMsg_1,\ldots,r_{\NumRounds-1},\LastMsg_{\NumRounds-1},r_{\NumRounds})$ be the transcript up to this point (i.e., excluding the final prover message).
	\item The challenger applies a unitary $U$ that computes the bit $\ARGVerifier(\Transcript',\RegResp_{\NumRounds})$ into a fresh ancilla, measures the ancilla, and applies $\contra{U}$. If the measurement outcome is $0$, the experiment aborts.
	\item If $b = 0$, the challenger does nothing. If $b = 1$, the challenger measures the $\RegResp_{\NumRounds}$ register in the computational basis and discards the result.
	\item The challenger sends all registers to $\Adversary$. $\Adversary$ outputs a bit $b'$, which is the output of the experiment.
\end{enumerate}
\end{definition}

\doclearpage
\section{Efficient quantum strategies for repeated games}
\label{sec:alt-proj-new}

We consider a classical single-player game $\Game$ played with quantum strategies. This section makes use of the notion of quantum interaction and interactive quantum algorithms; for details on how to model this formally, see \cite{VidickW16}.

\begin{definition}
A game $\Game = (\RSet,\ZSet,\Pred)$ consists of a question set $\RSet$, answer set $\ZSet$, and win predicate $\Pred \colon \RSet \times \ZSet \to \Bits$. An (efficient) \emph{quantum} strategy for $\Game$ is an interactive quantum algorithm $\StrategyMap$ with initial state $\DMatrix$.

The \emph{value} of a strategy $(\StrategyMap,\DMatrix)$, denoted $\Value{\Game}(\StrategyMap,\DMatrix)$, is the probability that a player using strategy $(\StrategyMap,\DMatrix)$ in the following game causes the referee to output $1$: the referee sends the player a question $r \gets R$, and the player answers with (classical) $z \in Z$; the referee outputs $f(r,z)$.
\end{definition}

We now define a \emph{quantum experiment} in which the player's answer can be an arbitrary quantum state on $\RegZ$, and the referee determines whether the player wins by computing $f(r,\RegZ)$ in superposition and measuring the output; it then uncomputes $f$ and returns $\RegZ$ to the player. The key difference between the classical and quantum experiments is that the only measurement performed in the quantum experiment is on the output of $f$, whereas a quantum player in a classical interaction must measure to send a classical $z$. While this does not affect the value of a game when played once, it is crucial when the game is repeated \emph{sequentially}.

In more detail, our quantum experiment consists of the following quantum interaction:

\begin{enumerate}[noitemsep]
    
    \item The referee samples a question $r \gets R$ and sends it to the  player.
    
    \item The player responds with a quantum state on register $\RegZ$.
    
    \item The referee computes $f(r,\RegZ)$ in superposition, measures the result to obtain an outcome $b \in \Bits$, and uncomputes $f$. The referee then returns $\RegZ$ to the player.
\end{enumerate}
It is easily verified that the probability a player following strategy $(\StrategyMap,\DMatrix)$ wins in the above experiment is $\Value{\Game}(S,\DMatrix)$, as in the classical experiment. Without loss of generality, we can assume that the strategy $S$ is implemented by a unitary $U_S$.

We now consider the $n$-fold \emph{sequential repetition} of the quantum experiment. Formally, the interaction consists of $n$ sequential rounds, where in the $i$th round:
\begin{enumerate}[noitemsep]
    \item The referee samples a question $r_i \gets R$ and sends it to the player.
    
    \item The player responds with a quantum state on $\RegZ$.
    
    \item The referee computes $f(r_i,\RegZ)$ in superposition, measures the result to obtain an outcome $b_i \in \Bits$, and uncomputes $f$. The referee then returns $\RegZ$ to the player, along with $b_i$.
\end{enumerate}

\begin{definition}[Value of a strategy in a repeated game]
    The \emph{value} of a strategy $(S,\DMatrix)$ in the above experiment is denoted $\Value{\Game}^{n}(\StrategyMap,\DMatrix)$, and is equal to $\Expectation[\sum_{i=1}^{n} b_i]$, the \emph{expected number of wins} across all trials. Note that $\Value{\Game}^{1}(S,\DMatrix) = \Value{\Game}(S,\DMatrix)$.
\end{definition}

When $\DMatrix$ is a classical state, the $n$-fold repetition $\StrategyMap^{n}$ of any strategy $\StrategyMap$ trivially achieves $\Value{\Game}^{n}(\StrategyMap^{n},\DMatrix) = n \cdot \Value{\Game}(\StrategyMap,\DMatrix)$. For quantum $\DMatrix$, this may not be true, since the state is in general disturbed by the referee's measurement. In this section we show that, given any \emph{quantum} strategy $(\StrategyMap,\DMatrix)$ for the one-round experiment, there is an efficient quantum algorithm $\StrategyMap'$ that makes black-box use of $U_{S}$ (and $U_{S}^\dagger$) such that $\Value{\Game}^{n}(\StrategyMap',\DMatrix) \approx n \cdot \Value{\Game}(\StrategyMap,\DMatrix)$.

\begin{theorem}
    \label{theorem:repeated-games}
    For any single-player quantum game $\Game = (\RSet,\ZSet,\Pred)$ with classically efficient predicate $\Pred$, $n \in \Naturals,\RWLoss \in [0,1]$, there is a quantum oracle algorithm $A_{\Game,n,\RWLoss}$ such that for all $(\StrategyMap,\DMatrix)$,
    \begin{equation*}
        \Value{\Game}^{n}(A_{\Game,n,\RWLoss}^{\StrategyMap},\DMatrix) \geq n \cdot (\Value{\Game}(\StrategyMap,\DMatrix) - \RWLoss)
    \end{equation*}
    and $A$ runs in expected time $\tilde{O}(|\Pred| \cdot n/\RWLoss)$ and makes an expected $\tilde{O}(n/\RWLoss)$ queries to $U_{\StrategyMap},U_{\StrategyMap}^{\dagger}$.
\end{theorem}

We prove the theorem using two key subroutines, $\ApproxEig$ and $\RepairProb$, which do the following:
\begin{itemize}
    \item $\ApproxEig^{\StrategyMap}$ applied to $\DMatrix$ is an \emph{approximate measurement} of $\Value{\Game}(\StrategyMap,\DMatrix)$. That is, it produces an outcome $p$ where $\Expectation[p] = \Value{\Game}(\StrategyMap,\DMatrix)$, and conditioned on obtaining outcome $p$ the post-measurement state $\DMatrix'$ satisfies $\Value{\Game}(\StrategyMap,\DMatrix') \approx p$.
    
    \item $\RepairProb^{\StrategyMap}_p$ is a procedure that \emph{repairs} a state that has been perturbed by the referee's measurement. In more detail, if $\DMatrix$ is the state of the system after applying $\ApproxEig^{\StrategyMap}$ and obtaining outcome $p$, and playing a one-round experiment with strategy $(\StrategyMap,\DMatrix)$ results in leftover state $\DMatrix'$, then applying $\RepairProb^{\StrategyMap}_p$ to $\DMatrix'$ outputs a \emph{repaired} state $\DMatrix^*$ in the sense that $\Value{\Game}(\StrategyMap,\DMatrix^*) \approx p$.
\end{itemize}
We remark that our implementations of $\ApproxEig^{\StrategyMap}$ and $\RepairProb^{\StrategyMap}_{p}$ make black-box use of $U_{\StrategyMap},U_{\StrategyMap}^{\dagger}$.

Given a strategy $(\StrategyMap,\DMatrix)$ for the one-round experiment, our $n$-time strategy is as follows.
\begin{enumerate}[noitemsep]
    \item[] Repeat for $i \in [n]$:
    \begin{enumerate}[noitemsep]
        \item Apply $p_i \gets \ApproxEig^{\StrategyMap}$.
        \item Receive $r_i \in \RSet$; run $S(r_i)$ coherently to compute $\RegZ$ and send it to the referee.
        \item Receive $\RegZ$ and measurement result $b_i \in \Bits$ from the referee.
        \item Apply $\RepairProb^{\StrategyMap}_{p_i}$.
    \end{enumerate}
\end{enumerate}
The guarantee of $\ApproxEig$ implies that $\Expectation[p_1] = \Value{\Game}(\StrategyMap,\DMatrix)$, and that $\Pr[b_i = 1] \approx \Expectation[p_i]$ for all $i$. The guarantee of $\RepairProb$ implies that $p_1 \approx p_2 \approx \cdots \approx p_n$ with high probability. Together these imply \cref{theorem:repeated-games}, by linearity of expectation.

\parhead{Organization} In \cref{sec:alternating-projectors} we present general technical lemmas that are useful for analysing algorithms which consist of alternating applications of two binary projective measurements; both $\ApproxEig$ and $\RepairProb$ are of this type. In \cref{sec:prob-estimation} we describe and analyze our $\ApproxEig$ procedure, which is a variant of procedures from \cite{MarriottW05,Zhandry20}. In \cref{sec:state-repair} we describe and analyze $\RepairProb$. Finally, in \cref{sec:proof-of-main-theorem} we prove~\cref{theorem:repeated-games}.

\subsection{Jordan subspaces and alternating measurements}
\label{sec:alternating-projectors}

We provide general tools for analysing \emph{alternating projection algorithms}, which were introduced by Marriott and Watrous~\cite{MarriottW05} for witness-preserving amplification of $\mathsf{QMA}$. In more detail, given two binary-outcome projective measurements $\MeasA = \BMeas{\BProj{\MeasA}}$ and $\MeasB = \BMeas{\BProj{\MeasB}}$ on a Hilbert space $\RegH$, an alternating projection algorithm applies the measurements in alternating fashion ($\MeasA,\MeasB,\MeasA,\MeasB,\ldots$) until a stopping condition is met (e.g., a certain number of measurements have been performed or some outcome has been observed). We can describe the distribution of measurement outcomes using Jordan's lemma (\cref{lemma:jordan}).

\parhead{Jordan decomposition}
Applying Jordan's lemma (\cref{lemma:jordan}) to $(\BProj{\MeasA},\BProj{\MeasB})$ induces an orthogonal decomposition $\RegH = \bigoplus_{j} \Subspace_{j}$ into one- and two-dimensional \emph{Jordan subspaces} $\Subspace_{j}$. 

Within each two-dimensional Jordan subspace $\Subspace_{j}$, we define four states $\JorKetA{j}{1},\JorKetA{j}{0},\JorKetB{j}{1},\JorKetB{j}{0}$:
\begin{itemize}[noitemsep]
    \item $\JorKetA{j}{1}$ is a state in $\Subspace_{j} \cap \image(\BProj{\MeasA})$.
    \item $\JorKetB{j}{1}$ is a state in $\Subspace_{j} \cap \image(\BProj{\MeasB})$.
    \item $\JorKetA{j}{0}$ is a state in $\Subspace_{j} \cap \ker(\BProj{\MeasA})$ (orthogonal to $\JorKetA{j}{1}$).
    \item $\JorKetB{j}{0}$ is a state in $\Subspace_{j} \cap \ker(\BProj{\MeasB})$ (orthogonal to $\JorKetB{j}{1}$).
\end{itemize}
These states are unique up to phase. Let
\begin{equation*}
p_j \eqdef \norm{\JorBraKetAB{j}{1}}^2 = \norm{\JorBraKetAB{j}{0}}^2
\enspace.
\end{equation*}
We adopt the convention that the phases of these states are chosen to satisfy
\begin{equation}
\label{eqn:jordan-vector-relations}
\JorKetA{j}{1} = \sqrt{p_j} \JorKetB{j}{1} + \sqrt{1-p_j} \JorKetB{j}{0} \ \text{ and } \
\JorKetB{j}{1} = \sqrt{p_j} \JorKetA{j}{1} + \sqrt{1-p_j} \JorKetA{j}{0}
\enspace.
\end{equation}

Notice that if $\ket{\psi}$ is the post-measurement state after $\MeasA$ has returned $1$, then $\ket{\psi} = \sum_j \alpha_j \JorKetA{j}{1}$ for some choice of amplitudes $\{\alpha_j\}_j$. Likewise, if $\ket{\psi}$ is the post-measurement state after $\MeasB$ has returned $1$, then $\ket{\psi} = \sum_j \alpha_j \JorKetB{j}{1}$ for some choice of amplitudes $\{\alpha_j\}_j$.

We can view each one-dimensional subspace $\Subspace_{j}$ as a \emph{degenerate} two-dimensional subspace. If $\BProj{\MeasA}$ acts as the identity on $\Subspace_{j}$ then we label the vector spanning the subspace $\JorKetA{j}{1}$; if $\BProj{\MeasA}$ is the zero projection on $\Subspace_{j}$ then we label the vector $\JorKetA{j}{0}$. We use a similar convention for $\BProj{\MeasB}$ (so the vector spanning a one-dimensional subspace has two labels). We set $p_j \eqdef 1$ if both $\BProj{\MeasA}$ and $\BProj{\MeasB}$ act as the identity or both act as zero, and $p_j \eqdef 0$ otherwise. One can verify that the discussion above for two-dimensional subspaces holds for one-dimensional subspaces under this convention.

\parhead{Distribution of measurement outcomes}
Consider the following (classical) probability distribution $\MWDistFull{p}{T}$ (for ``Marriott--Watrous distribution''), parameterized by a probability $p \in [0,1]$ and positive integer $T$.
\begin{center}
\begin{minipage}{0.9\textwidth}
\begin{itemize}[noitemsep]
  \item[] $\MWDistFull{p}{T}$:
\begin{enumerate}[noitemsep]
    \item For each $i \in [T]$, set $a_i := 1$ with probability $p$ and $a_i := 0$ otherwise. 
    \item Let $b_0 \eqdef 1$. For $i \in [T]$, define $b_i := b_{i-1} \oplus a_i$. 
    \item Output $b_1,b_2,\ldots,b_T$.
\end{enumerate}
\end{itemize}
\end{minipage}
\end{center}
The following two lemmas characterize the distribution of measurement outcomes of an alternating measurement procedure. The analysis closely follows that of \cite{MarriottW05,Regev06-XXX}.
\begin{lemma}
\label{lemma:alt-meas-eigenstate}
The measurement outcomes that result from applying $T$ alternating measurements $\MeasA,\MeasB,\MeasA,\MeasB\ldots$ to $\JorKetB{j}{1}$ are distributed according to $\MWDistFull{p_j}{T}$.
\end{lemma}

\begin{proof}
This is a consequence of two symmetric claims that follow directly from~\cref{eqn:jordan-vector-relations}.
\begin{itemize}
    \item If $\MeasA$ is applied to $\JorKetB{j}{b}$, then with probability $p_j$ the outcome is $b$ and the post-measurement state is $\JorKetA{j}{b}$, and with probability $1-p_j$ the outcome is $1-b$ and the post measurement state is $\JorKetA{j}{1-b}$.
    \item If $\MeasB$ is applied to $\JorKetA{j}{b}$, then with probability $p_j$ the outcome is $b$ and the post-measurement state is $\JorKetB{j}{b}$, and with probability $1-p_j$ the outcome is $1-b$ and the post measurement state is $\JorKetB{j}{1-b}$. 
\end{itemize}
It is convenient to think of the initial state $\JorKetB{j}{1}$ as the post-measurement state after $\MeasB$ returns $1$. Letting $b_0 \eqdef 1$, for any $i \in [T]$ the $i$-th measurement outcome $b_i$ is equal to $b_{i-1}$ with probability $p_j$ and equal to $1-b_{i-1}$ with probability $1-p_j$, giving the distribution $\MWDistFull{p_j}{T}$.
\end{proof}

We can generalize~\cref{lemma:alt-meas-eigenstate} to characterize the measurement outcomes when we begin with any state in $\image(\BProj{\MeasB})$, which must be of the form $\sum_j \alpha_j \JorKetB{j}{1}$.

\begin{lemma}
\label{lemma:alt-meas-general-state}
The measurement outcomes that result from applying $T$ alternating measurements $\MeasA,\MeasB,\MeasA,\MeasB,\ldots$ to the state $\sum_j \alpha_j \JorKetB{j}{1}$ have the following distribution:
\begin{enumerate}[nolistsep]
    \item sample $p_j$ with probability  $\abs{\alpha_j}^2$;
    \item output $\MWDistFull{p_j}{T}$.
\end{enumerate}
\end{lemma}

\begin{proof}
Consider the \emph{Jordan subspace measurement} $\Meas{\Jor}[\BProj{\MeasA},\BProj{\MeasB}] \eqdef (\SProj[\Jor]{j})_j$ on $\RegH$, where
\begin{equation*}
\SProj[\Jor]{j} \eqdef \JorProjA{j}{1} + \JorProjA{j}{0} = \JorProjB{j}{1} + \JorProjB{j}{0}
\enspace.
\end{equation*}
In words, $\Meas{\Jor}[\BProj{\MeasA},\BProj{\MeasB}]$ is the projective measurement onto the Jordan subspaces $\{S_j\}_{j}$ that outputs a Jordan subspace label $j$. 

Suppose that we perform the measurement $\Meas{\Jor}[\BProj{\MeasA},\BProj{\MeasB}]$ on $\sum_j \alpha_j \JorKetB{j}{1}$, and subsequently perform $T$ alternating measurements $\MeasA,\MeasB,\MeasA,\MeasB,\ldots$. The outcome of $\Meas{\Jor}[\BProj{\MeasA},\BProj{\MeasB}]$ is $j$ with probability $\abs{\alpha_j}^2$, and the subsequent alternating measurement outcomes are distributed according to $\MWDistFull{p_j}{T}$ by~\cref{lemma:alt-meas-eigenstate}. It remains to prove that the distribution of measurement outcomes is unchanged even if we skip the $\Meas{\Jor}[\BProj{\MeasA},\BProj{\MeasB}]$ measurement.

This is because $\Meas{\Jor}[\BProj{\MeasA},\BProj{\MeasB}]$ \emph{commutes} with both $\MeasA$ and $\MeasB$. To see that $\Meas{\Jor}[\BProj{\MeasA},\BProj{\MeasB}]$ commutes with $\MeasA$, observe that the corresponding measurement operators are diagonal in the basis $\{\JorKetA{j}{b}\}_{j,b}$, since $\BProj{\MeasA} = \sum_j \JorProjA{j}{1}$ by Jordan's lemma and $\SProj[\Jor]{j} = \JorProjA{j}{0} + \JorProjA{j}{1}$ for all $j$ by definition. $\Meas{\Jor}[\BProj{\MeasA},\BProj{\MeasB}]$ commutes with $\MeasB$ by an identical argument for the basis $\{\JorKetB{j}{b}\}_{j,b}$.

As a consequence, we can commute $\Meas{\Jor}[\BProj{\MeasA},\BProj{\MeasB}]$ to occur \emph{after} the $T$ alternating measurements $\MeasA,\MeasB,\MeasA,\MeasB,\ldots$, at which point $\Meas{\Jor}[\BProj{\MeasA},\BProj{\MeasB}]$ has no effect on the measurement outcomes.
\end{proof}

\newcommand{\NumReps}{\mathsf{NReps}}

\parhead{Almost projective measurements}
We state a property of general measurements due to~\cite{Zhandry20} that captures when a measurement is ``close'' to being projective, in the sense that sequential applications of the measurement yield similar outcomes.

\begin{definition}
\label{def:almost-proj}
A real-valued measurement $\Measurement$ on $\RegH$ is \defemph{$(\varepsilon,\delta)$-almost-projective} if applying $\Measurement$ twice in a row to any state $\DMatrix \in \Hermitians{\RegH}$ produces measurement outcomes $p,p'$ where 
\begin{equation*}
\Pr[\abs{p-p'} \leq \varepsilon] \geq 1-\delta
\enspace.
\end{equation*}
\end{definition}

We briefly discuss how alternating measurements $\MeasA,\MeasB$ constitutes a $(\varepsilon,\delta)$-almost projective approximation of $\MeasJor[\BProj{\MeasA},\BProj{\MeasB}]$. While we will not make use of this fact directly (we prove a variant of it in \cref{lemma:approx-eig}), we will introduce some concepts and notation that are useful later. For $\vec{b} \in \Bits^{n+1}$, and letting $\Fractions{n} \eqdef \{ 0, 1/n, 2/n, \ldots, 1\}$, define
\begin{equation*}
\NumReps(\vec{b}) \eqdef \frac{|\{ j \in \{1, \ldots, n\} : b_{j-1} = b_{j} \}|}{n} \in \Fractions{n}
\enspace.
\end{equation*}
That is, $\NumReps(\vec{b})$ is the number of pairs of consecutive repeated bits in $\vec{b}$, divided by $n$; for example $p(0,0,1,1,1,0) = 3/5$. The following proposition is immediate from the definition of $\MWDist$:

\begin{proposition}
\label{prop:pest-mwdist}
If $\vec{b} \sim \MWDistFull{p}{T}$ then $\NumReps(1,\vec{b}) \sim \Bin(T,p)/T$.
\end{proposition}

Consider the measurement procedure $\Measurement$ that applies $T$ measurements $\MeasA,\MeasB,\MeasA,\MeasB,\ldots$ in an alternating fashion, and outputs $\NumReps(b_1,\ldots,b_{T})$, where the $b_i$ are the measurement outcomes. Then by \cref{lemma:alt-meas-general-state}, for $\ket{\psi} = \sum_{j} \alpha_j \JorKetB{j}{1}$, $\Expectation_{p \gets \Measurement(\ket{\psi})}[p] = |\alpha_j|^2 p_j$. Moreover, if $T \approx \frac{1}{\varepsilon} \log \frac{1}{\delta}$ then $\Measurement$ is $(\varepsilon,\delta)$-almost projective. This is because sequential applications of $\Measurement$ are equivalent to a single application of $\Measurement$ of length $2T$; $(\varepsilon,\delta)$-almost projectivity follows by a Chernoff bound.

\subsection{Probability estimation}
\label{sec:prob-estimation}

We describe a measurement procedure $\ApproxEig$ that \emph{estimates} $\Value{\Game}(\StrategyMap,\DMatrix)$, following techniques of \cite{MarriottW05,Zhandry20}. The procedure is a variation on the ``approximate projective implementation'' procedure of \cite{Zhandry20}, and we show that it is $(\varepsilon,\delta)$-almost projective. We also show that if $\ApproxEig(\DMatrix)$ produces an outcome $\geq p$ with high probability, then $\Value{\Game}(\StrategyMap,\DMatrix)$ cannot be much smaller than $p$.

A player with unitary strategy $U_S$ and initial state $\DMatrix$ in the game $\Game = (\RSet,\ZSet,\Pred)$ receives a random challenge $r \gets R$, applies $U_S$ to $\ketbra{r}^{\RegR} \otimes \DMatrix^{\RegZ,\RegI}$, and sends $\RegZ$ to the referee; here $\RegR$ is supported on $\{\ket{r}\}_{r \in R}$, $\RegZ$ is supported on $\{\ket{z}\}_{z \in Z}$, and $\RegI$ denotes the player's internal registers.

The procedure $\ApproxEig$ is parameterized by $\varepsilon,\delta \in [0,1]$ and a game $\Game$, and has black-box access to the player's unitary $U_{\StrategyMap}$ and its inverse $U_{\StrategyMap}^\dagger$, and operates on registers $(\RegZ,\RegI)$.

We set
\begin{equation*}
\AEigReps \eqdef \AEigReps(\varepsilon,\delta) \eqdef \max\{\lceil \ChernoffN_{\varepsilon/2,\delta/4}/2 \rceil, \log_{5/8}(\delta/2)\} = O\left(\frac{1}{\varepsilon} \log \frac{1}{\delta}\right) \enspace,    
\end{equation*}
where $\ChernoffN_{\varepsilon,\delta}$ is a parameter defined in \cref{prop:chernoff} for the Chernoff bound. Let $\RegR'$ be a register with basis $\{ \ket{r} \}_{r \in \RSet} \cup \{ \ket{\top}, \ket{\bot} \}$. We define the state $\UniState{R}$ on $\RegR',\RegR$ as 
\begin{align*}
\UniState{R} \eqdef \frac{1}{2} \ket{\top,0} + \frac{1}{2} \ket{\bot,0} + \frac{1}{\sqrt{2|\RSet|}} \sum_{r \in \RSet} \ket{r,r}
\enspace,
\end{align*}
where $\top$ and $\bot$ are arbitrary symbols distinct from the elements of $R$, and $0 \in \RSet$.
\begin{remark}
    We introduce the auxiliary (control) register $\RegR'$ for two reasons:
    \begin{enumerate}[nolistsep,label=(\alph*)]
        \item $\RegR'$ has two additional basis elements $\ket{\top},\ket{\bot}$. These are special symbols which correspond to ``automatically'' winning or losing the game, respectively. This forces our probability estimates to be scaled within the range $[1/4,3/4]$, which can easily be rescaled to $[0,1]$ before outputting a final value. This modification ensures that the procedure terminates within a polynomial number of steps except with negligible probability.
        \item Tracing out the $\RegR'$ register leaves the classical mixed state $\frac{1}{2}\ketbra{0} + \frac{1}{2|\RSet|} \sum_r \ketbra{r}$ on $\RegR$; this ensures that $U_{S}$ behaves as if it were invoked on random $\ket{r}$ (or $0$, with probability $1/2$).
    \end{enumerate}
\end{remark}

We are now ready to define the procedure $\ApproxEig$.

\begin{center}
\begin{minipage}{0.9\textwidth}
\noindent $\ApproxEig_{\Game,\varepsilon,\delta}^{U}$:
\begin{enumerate}[itemsep=0pt]
\item Initialize registers $(\RegR',\RegR)$ to $\UniState{R}$;
\item Define $\Meas{\ConProj} := \BMeas{\BProj{\ConProj}}$ where $\BProj{\ConProj} \eqdef U_{\StrategyMap}^{\dagger} \BProj{\Pred} U_{\StrategyMap}$
for
\begin{align*}
    \BProj{\Pred} &\eqdef \sum_{r,z,\Pred(r,z) = 1} \ketbra{r,z}^{\RegR',\RegZ} + \ketbra{\top}^{\RegR'} \otimes I^{\RegZ} \enspace.
\end{align*}
\item \label[step]{step:aeig-estimate} For $i = 1, \ldots, \AEigReps$:
\begin{enumerate}[nolistsep]
\item Apply $\Meas{\ConProj}$, obtaining outcome $L_{2i-1} \in \Bits$.
\item Apply $\Meas{\UniState{\RSet}} \eqdef \BMeas{\UniProj{\RSet}^{\RegR',\RegR}}$, obtaining outcome $L_{2i} \in \Bits$.
\end{enumerate}
\item \label[step]{step:aeig-disentangle} If $L_{2\AEigReps} = 1$, skip to  \cref{step:aeig-output}. Otherwise, apply $\Meas{\ConProj},\Meas{\UniState{\RSet}}$ to $\RegA$ in an alternating fashion until $\Meas{\UniState{R}} \to 1$, or a further $2\AEigReps$ measurements have been applied.
\item \label[step]{step:aeig-output} Discard $\RegR$ and $\RegR'$; output $\PEst \eqdef 2\cdot \NumReps(1,L_1,\dots,L_{2\AEigReps})-1/2$.
\end{enumerate}
\end{minipage}
\end{center}

\begin{lemma}
    \label{lemma:approx-eig}
    The measurement $\ApproxEig \eqdef \ApproxEig_{\Game,\varepsilon,\delta}^{\StrategyMap}$ has the following properties:
    \begin{enumerate}[label=(\roman*)]
    \item \label{item:aeig-time} $\ApproxEig$ is an oracle circuit of size $O(|\Pred| \cdot \frac{1}{\varepsilon} \log \frac{1}{\delta})$ that applies $U_{\StrategyMap}$ and $U_{\StrategyMap}^{\dagger}$ $O(\frac{1}{\varepsilon} \log \frac{1}{\delta})$ times;
    \item \label{item:aeig-expectation} for every $\DMatrix \in \Hermitians{\RegZ,\RegI}$, $\Expectation_{\PEst \gets \ApproxEig(\DMatrix)}[\PEst] = \Value{\Game}(\StrategyMap,\DMatrix)$;
    \item \label{item:aeig-approx-proj} $\ApproxEig$ is $(\varepsilon,\delta)$-almost projective;
    \item \label{item:aeig-good-states} for every $p \in \bbR$, if $\Pr_{p' \gets \ApproxEig(\DMatrix)}[p' \geq p] \geq 1-\gamma$ then $\Value{\Game}(\StrategyMap,\DMatrix) \geq p - \gamma - \varepsilon - \delta$;
    \item \label{item:aeig-pm-state} for every $\DMatrix \in \Hermitians{\RegZ,\RegI}$, $\Value{\Game}(\StrategyMap,\ApproxEig(\DMatrix)) \geq \Value{\Game}(\StrategyMap,\DMatrix) - \delta$.
    \end{enumerate}
\end{lemma}

\begin{proof}
\cref{item:aeig-time} follows directly from the description; we proceed to prove \cref{item:aeig-expectation,item:aeig-approx-proj,item:aeig-good-states,item:aeig-pm-state}. It suffices to prove each property for pure states $\ket{\psi} \in \RegZ \otimes \RegI$, as the statement for mixed states follows by convexity.

Consider a decomposition of $\RegR' \otimes \RegR \otimes \RegZ \otimes \RegI$ into the Jordan subspaces for projectors $\BProj{\ConProj}$ and $\UniProj{\RSet}^{\RegR',\RegR} \otimes \Id^{\RegZ,\RegI}$ (henceforth we will write the projector $\UniProj{\RSet}^{\RegR',\RegR} \otimes \Id^{\RegZ,\RegI}$ as $\UniProj{\RSet}^{\RegR',\RegR}$). Following our notation for Jordan subspaces in \cref{sec:alternating-projectors}, we will associate $\Meas{\ConProj}$ with $\MeasA$ and $\Meas{\UniState{R}}$ with $\MeasB$, so that in the $j$-th Jordan subspace:
\begin{itemize}[noitemsep]
    \item $\BProj{\ConProj}$ is a projection onto $\JorKetA{j}{1}$,
    \item $\UniProj{\RSet}^{\RegR',\RegR}$ is a projection onto $\JorKetB{j}{1}$, and
    \item $p_j = \norm{\JorBraKetAB{j}{1}}^2$.
\end{itemize} 
    
    Write $\UniState{\RSet}^{\RegR,\RegR'} \otimes \ket{\psi}^{\RegZ,\RegI} = \sum_{j} \alpha_j \JorKetB{j}{1}$. Note that
    \begin{equation*}
        \sum_{j} |\alpha_j|^2 p_j = \norm{\BProj{\ConProj} \UniState{\RSet}\ket{\psi}}^2 = \frac{\Value{\Game}(\StrategyMap,\ket{\psi})}{2} + \frac{1}{4} \enspace.
    \end{equation*}
    
    By \cref{lemma:alt-meas-general-state}, $\PEst \gets \ApproxEig(\ket{\psi})$ is distributed as:
    \begin{enumerate}[noitemsep]
        \item Choose $j$ with probability $|\alpha_{j}|^2$.
        \item Sample $L_{1},\ldots,L_{2\AEigReps} \gets \MWDist(p_j, 2\AEigReps)$.
        \item Output $\PEst \eqdef 2p(1,L_{1},\ldots,L_{2\AEigReps})-1/2$.
    \end{enumerate}
    Hence in particular we have that \[\Expectation[\PEst] = 2\sum_{j} |\alpha_j|^2 \Expectation[\NumReps(1,L_1,\ldots,L_{2t})] - 1/2 = 2\sum_{j} |\alpha_j|^2 p_j - 1/2 = \Value{\Game}(\StrategyMap,\ket{\psi}) \enspace,\]
    which establishes \ref{item:aeig-expectation}.
    
    We now prove \ref{item:aeig-good-states}. Suppose that $\Pr_{p' \gets \ApproxEig(\ket{\psi})}[p' \geq p] \geq 1-\gamma$. Then
    \begin{equation*}
        \gamma \geq \Pr_{\PEst \gets \ApproxEig(\ket{\psi})}[\PEst < p] = \sum_{j} |\alpha_j|^2 \Pr_{\vec{L} \gets \MWDist(p_j,2\AEigReps)}[p(1,\vec{L}) < p/2+1/4] \geq \sum_{j, p_j < p/2 + 1/4 - \varepsilon} |\alpha_j|^2(1 - \delta) \enspace,
    \end{equation*}
    by \cref{prop:pest-mwdist}. Rearranging,
    \begin{equation*}
        \sum_{j, p_j < p/2 + 1/4 - \varepsilon} |\alpha_j|^2 \leq \gamma + \delta \enspace.
    \end{equation*}
    Hence
    \begin{equation*}
        \Value{\Game}(\StrategyMap,\ket{\psi}) = 2\sum_{j} |\alpha_j|^2 p_j - 1/2 \geq p - \gamma - \varepsilon - \delta \enspace.
    \end{equation*}
    
    Next we prove \ref{item:aeig-approx-proj}. Let $D$ be the distribution on $\Fractions{2\AEigReps} \times \Fractions{2\AEigReps}$ arising from two sequential applications of $\ApproxEig$ with initial state $\ket{\psi}$ (recall that $Q_{2t} = \{0,\frac{1}{2t},\frac{2}{2t},\ldots,1\}$). Let $D'$ be sampled as follows.
    \begin{enumerate}[noitemsep]
        \item Choose $j$ with probability $|\alpha_{j}|^2$.
        \item Sample $L_{1},\ldots,L_{4\AEigReps} \gets \MWDist(p_j, 4\AEigReps)$.
        \item Sample $L'_{1},\ldots,L'_{2\AEigReps} \gets \MWDist(p_j, 2\AEigReps)$.
        \item Compute $\PEst \eqdef p(1,L_{1},\ldots,L_{2\AEigReps})$ and $\PEst' \eqdef p(1,L'_{1},\ldots,L'_{2\AEigReps})$.
        \item Output $(\PEst,\PEst')$.
    \end{enumerate}
    The statistical distance between $D$ and $D'$ is bounded by $\Pr[ \forall i \in [\AEigReps,2\AEigReps],\, L_{2i} = 0]$. This can be shown by coupling the outcomes of the first $4\AEigReps$ measurements with $L_{1},\ldots,L_{4t}$ drawn by $D'$. If this bad event does not occur, the first application of $\ApproxEig$ terminates in some state $\UniState{\RSet}^{\RegR',\RegR}\ket{\phi}^{\RegZ,\RegI}$, and so tracing out $(\RegR',\RegR)$ and then reinitializing it to $\UniState{\RSet}$ at the beginning of the second application of $\ApproxEig$ has no overall effect on the state. In this case, therefore, we can view the two applications of $\ApproxEig$ as a single alternating measurement procedure of length $2i+2\AEigReps$ conditioned on the outcome of the $2i$-th measurement being $1$. Then by \cref{lemma:alt-meas-general-state}, in this case $D$ and $D'$ are identically distributed.
    
    We now bound $\Pr[ \forall i \in [\AEigReps,2\AEigReps],\, L_{2i} = 0]$. Suppose that $j$ is sampled in the first step. For each $i \in [\AEigReps+1,\ldots,2\AEigReps]$, the probability that $L_{2i} = 1$ given that $L_{2i-2} = 0$ is $2p_j(1-p_j)$. Note that for every subspace $j$ where $\JorKetB{j}{1}$ is nonzero, $p_j = \norm{\BProj{\ConProj} \JorKetB{j}{1}}^2 \in [1/4,3/4]$; in particular, this holds for all subspaces $j$ such that $\alpha_j \neq 0$. Hence for any $j$ sampled with positive probability, $2p_j(1-p_j) \geq 3/8$. It follows that the probability that $L_{2i} = 0$ for all $i \in [\AEigReps,2\AEigReps]$ is at most $(5/8)^\AEigReps \leq \delta/2$.
    
    Finally we show that
    \begin{equation*}
        \Pr_{(\PEst,\PEst') \gets D'}[|\PEst-\PEst'| > \varepsilon] < \delta/2 \enspace,
    \end{equation*}
    which will complete the proof. Observe that for $j$ sampled in the first step, $\PEst,\PEst' \sim \Bin(2\AEigReps,p_j)/2\AEigReps$. Hence by \cref{prop:chernoff} (Chernoff bound), $\Pr[|\PEst - p_j| > \varepsilon/2] < \delta/4$, and similarly for $\PEst'$. The equation follows by a union bound.
    
    It remains to prove \cref{item:aeig-pm-state}. Recall that for any state $\DMatrix \in \Hermitians{\RegResp,\RegInt}$, we have that \[\Value{\Game}(\StrategyMap,\DMatrix) = 2\sum_{j} p_j \Tr(\SProj[\Jor]{j} (\UniProj{\RSet} \otimes \DMatrix)) - 1/2.\] 
    
    Let $\ApproxEig'$ be defined identically to $\ApproxEig$ except that it does not discard $\RegR,\RegR'$. Since for all $j$, $\SProj[\Jor]{j}$ commutes with $\Meas{\ConProj}$,$\Meas{\UniState{\RSet}}$, we have \[\Tr(\SProj[\Jor]{j} \ApproxEig'(\DMatrix)) = \Tr(\SProj[\Jor]{j} (\UniProj{\RSet}^{\RegR',\RegR} \otimes \DMatrix))~.\]
    Then we have
    \begin{align*}
    \Value{\Game}(\StrategyMap,\ApproxEig(\DMatrix)) &= 2\sum_{j} p_j \Tr(\SProj[\Jor]{j} (\UniProj{\RSet}^{\RegR',\RegR} \otimes \ApproxEig(\DMatrix))) - 1/2 \\
    &= 2\sum_{j} p_j \Tr(\SProj[\Jor]{j} (\UniProj{\RSet}^{\RegR',\RegR} \otimes \Tr_{\RegR,\RegR'}(\ApproxEig'(\DMatrix)))) - 1/2 \\
    &\geq 2\sum_{j} p_j \Tr(\SProj[\Jor]{j} (\UniProj{\RSet}^{\RegR',\RegR} \otimes \Tr_{\RegR,\RegR'}(\UniProj{\RSet}^{\RegR',\RegR} \cdot \ApproxEig'(\DMatrix)))) - 1/2 \\
    &= 2\sum_{j} p_j \Tr(\SProj[\Jor]{j} \UniProj{\RSet}^{\RegR',\RegR} \cdot \ApproxEig'(\DMatrix)) - 1/2 \\
    &\geq 2\sum_{j} p_j \Tr(\SProj[\Jor]{j} \ApproxEig'(\DMatrix)) - 1/2 - \delta \\
    &= 2\sum_{j} p_j \Tr(\SProj[\Jor]{j} \DMatrix) - 1/2 - \delta = \Value{\Game}(\StrategyMap,\DMatrix) - \delta ~,
    \end{align*}
    where the final inequality follows because $\Tr(\UniProj{\RSet} \ApproxEig'(\DMatrix))$ is at least the probability that $\ApproxEig$ terminates with $\Meas{\UniState{R}} \to 1$, which is at least $1-\delta$.
\end{proof}

\subsection{A state repair procedure}
\label{sec:state-repair}

We construct a procedure $\Repair^{\Measurement}(p)$ parameterized by an almost-projective measurement $\Measurement$ and with input $p \in \bbR$ that (under certain conditions) outputs a state $\DMatrix$ satisfying the guarantee: ``applying $\Measurement$ to $\DMatrix$ produces an outcome $\approx p$ with high probability''. We then obtain $\RepairProb$ by plugging in the almost-projective measurement $\ApproxEig$ for $\Measurement$.

\paragraph{The procedure.} Formally, our state repair procedure $\Repair^{\Measurement,\MeasC}_T$ is a CPTP map on a register $\RegH$, parameterized by:
\begin{itemize}[noitemsep]
    \item a positive integer $T$,
    \item an oracle for an $(\varepsilon,\delta)$-almost-projective measurement $\Measurement$ on $\RegH$, and
    \item an oracle for an $N$-outcome projective measurement $\MeasC = (\BProj{k})_{k=1}^{N}$ on $\RegH$,
\end{itemize}
and taking classical inputs $(k,p)$ where $k \in [N]$ and $p \in \bbR$. 

Recall that the measurement $\Measurement = (M_q)_{q \in I}$, where $I \subseteq \bbR$ is the set of outcomes of $\Measurement$, can be implemented as a unitary $\MeasUnitary$ on $(\RegH,\RegW)$ for some ancilla register $\RegW$, followed by some projective measurement $(\BProj{\Measurement,q})_{q \in I}$ on $\RegW$. Formally, for each $q \in I$, the unitary $\MeasUnitary$ and projector $\BProj{\Measurement,q}$ satisfy $M_{q} \DMatrix M_{q}^{\dagger} = \Tr_{\RegW}(\BProj{\Measurement,q} \MeasUnitary (\DMatrix \otimes \ketbra{0}^{\RegW}) \MeasUnitary^{\dagger})$ for all $\DMatrix \in \Hermitians{\RegH}$. We are now ready to give the state repair procedure.
\begin{center}
\begin{minipage}{0.9\textwidth}
$\Repair_{T}^{\Measurement,\MeasC}(k,p)$:
\begin{enumerate}[itemsep=0pt]
\item Define measurements
\begin{align*}
\MeasA_p \coloneqq \BMeas{\BProj{\MeasA,p}} \ &\mathrm{where} \ \BProj{\MeasA,p} \coloneqq \sum_{q \in [p \pm \varepsilon]} \MeasUnitary^{\dagger} \BProj{\Measurement,q} \MeasUnitary \enspace, \\
\MeasB_k \coloneqq \BMeas{\BProj{\MeasB,k}} \ &\mathrm{where} \ \BProj{\MeasB,k} \coloneqq \BProj{k}\otimes \ketbra{0}^{\RegW} \enspace.
\end{align*}
\item \label[step]{step:repair-init-w} Initialize $\RegW$ to $\ket{0}$.
\item Apply the measurement $\MeasA_p$. If the outcome is $1$, skip to \cref{step:trace-out}.
\item Apply the measurements $\MeasB_k,\MeasA_p,\MeasB_k,\MeasA_p,\ldots$ in alternating fashion until either (1) $\MeasA_p \rightarrow 1$ occurs or (2) $T$ applications of $(\MeasB_k,\MeasA_p)$ have been applied (whichever comes first).

\item \label[step]{step:trace-out} Apply $\MeasUnitary$ to $(\RegH,\RegW)$, and discard the $\RegW$ registers.
\end{enumerate}
\end{minipage}
\end{center}

The following lemma describes the effect of the repair procedure.

\newcommand{\RepairExpt}{\mathsf{RepairExpt}}
\newcommand{\Expt}{\mathsf{Expt}}
\begin{lemma}[State repair]
    \label{lemma:state-repair}
    Let $\Measurement$ be an $(\varepsilon,\delta)$-almost projective measurement on $\RegH$, $\MeasC = (\BProj{k})_{k=1}^{N}$ be a projective measurement on $\RegH$ with $N$ outcomes, $T$ be a positive integer. Consider the following quantum measurement procedure $\RepairExpt$ on $\RegH$:
    \begin{enumerate}[noitemsep]
        \item Measure the initial state: apply $\Measurement$, obtaining outcome $p$;
        \item Damage the state: apply $\MeasC$, obtaining outcome $k$;
        \item Repair the state: run $\Repair_{T}^{\Measurement,\MeasC}(k,p)$ and let $R$ denote the total number of calls to $\Measurement$ and $\MeasC$.
        \item Output $p$.
    \end{enumerate}
    Then $\RepairExpt$ is $(2\varepsilon, N(\delta + 1/T) + 4\sqrt{\delta})$-almost projective,
    and $\Expectation[R] \leq N+ 4T\sqrt{\delta} + 1$.
\end{lemma}

\begin{proof}[Proof of \cref{lemma:state-repair}]

We write out in full the steps applied in $\RepairExpt$:
\begin{enumerate}[noitemsep]
    \item[] $\RepairExpt$:
    \item Apply $\Measurement$, obtaining outcome $p$;
    \item \label[step]{step:perturb-1} Apply $\MeasC$, obtaining outcome $k \in [N]$.
    \item \label[step]{step:repair-init-w-1} Initialize $\RegW$ to $\ket{0}$.
    \item \label[step]{step:initial-meas-1} Apply the measurement $\MeasA_p$. If the outcome is $1$, skip to \cref{step:trace-out-1}.
    \item \label[step]{step:alt-meas-1} Apply the measurements $\MeasB_k,\MeasA_p,\MeasB_k,\MeasA_p,\ldots$ in alternating fashion until either (1) $\MeasA_p \rightarrow 1$ occurs or (2) $T$ applications of $(\MeasB_k,\MeasA_p)$ have been applied (whichever comes first).
    \item \label[step]{step:trace-out-1} Apply $\MeasUnitary$ to $(\RegH,\RegW)$, and discard the $\RegW$ registers.
\end{enumerate}

From this point on, we refer to \cref{step:initial-meas-1,step:alt-meas-1,step:trace-out-1} as $\Repair'(k,p)$, which maps $\RegH \otimes \RegW \to \RegH$.

\noindent Define the $(N+1)$-outcome projective measurement $\MeasC' \eqdef ((\BProj{\MeasB,k})_{k=1}^N, \BProj{\bot})$ on $\RegH \otimes \RegW$ where $\BProj{\bot} \eqdef \Id^{\RegH} \otimes (\Id - \ketbra{0})^{\RegW}$. Next consider the following experiment $\Expt_1$ (differences highlighted in \textcolor{red}{red}).
\begin{enumerate}[noitemsep]
    \item[] $\Expt_1$:
    \item Apply $\Measurement$, obtaining outcome $p$;
    \textcolor{red}{
    \item \label[step]{step:repair-init-w-2} Initialize $\RegW$ to $\ket{0}$.
    \item \label[step]{step:perturb-2} Apply $\MeasC'$, obtaining outcome $k \in [N] \cup \{\bot\}$.
    }
    \item Apply $\Repair'(k,p)$.
\end{enumerate}

$\Expt_1$ and $\RepairExpt$ are equivalent, since $\Expt_1$ can be obtained by performing the following changes to $\RepairExpt$:
\begin{itemize}[noitemsep]
    \item Swap the order of \cref{step:perturb-1} and \cref{step:repair-init-w-1} in $\RepairExpt$. This does not change the resulting experiment since $\MeasC$ acts trivially on $\RegW$.
    \item Then, apply $\MeasC'$ instead of $\MeasC$ to obtain $k$. This causes no change since $\MeasC'(\DMatrixB^{\RegH} \otimes \ketbra{0}^{\RegW}) = \MeasC(\DMatrixB^{\RegH}) \otimes \ketbra{0}^{\RegW}$ for all $\DMatrixB \in \Hermitians{\RegH}$.
\end{itemize}

We now define another experiment $\Expt_2$ as follows (differences from $\Expt_1$ highlighted in \textcolor{red}{red}).

\begin{enumerate}[noitemsep]
    \item[] $\Expt_2$:
    \item \label[step]{step:initial-measurement-3} Apply $\Measurement$, obtaining outcome $p$;
    \item \label[step]{step:repair-init-w-3} Initialize $\RegW$ to $\ket{0}$.
    \textcolor{red}{\item \label[step]{step:repair-postselect} Apply $\MeasA_{p}$ to $(\RegH,\RegW)$ and postselect on obtaining outcome $1$.}
    \item \label[step]{step:perturb-3} Apply $\MeasC'$, obtaining outcome $k \in [N] \cup \{\bot\}$.
    \item \label[step]{step:repair-3} Apply $\Repair'(k,p)$.
\end{enumerate}

It will be convenient hereafter to treat $\Measurement$ and $\MeasC'$ as CPTP maps that write their output to a new output register, i.e., $\Measurement \colon \Hermitians{\RegH} \to \Hermitians{\RegH \otimes \RegO_1}$ and $\MeasC' \colon \Hermitians{\RegH} \to \Hermitians{\RegH \otimes \RegO_2}$. For the remainder of the proof, fix an initial state $\DMatrix \in \Hermitians{\RegH}$. Let $\DMatrix_1$ denote the state on $(\RegH,\RegW,\RegO_1)$ directly before \cref{step:perturb-2} in $\Expt_1$ applied to $\DMatrix$. Let $\DMatrix_2$ denote the state on the same registers directly before \cref{step:perturb-3} in $\Expt_2$ applied to $\DMatrix$. We show that these states are close in trace distance.

\begin{claim}
\label{claim:stat-distance-hybrids}
The trace distance between $\DMatrix_1$ and $\DMatrix_2$ is at most $2\sqrt{\delta}$.
\end{claim}

\begin{proof}
Let $\BProj{\MeasA}' \eqdef \sum_{p \in I} \ketbra{p}^{\RegO_1} \otimes \BProj{\MeasA,p}^{\RegH,\RegW}$. We have that
	\[\DMatrix_2 = \frac{\BProj{\MeasA}' \DMatrix_1 \BProj{\MeasA}'}{\Tr(\BProj{\MeasA}' \DMatrix_1)} \enspace.\]
Note that $\Tr(\BProj{\MeasA}' \DMatrix_1) = \Tr(\BProj{\MeasA}' (\Measurement(\DMatrix) \otimes \ketbra{0}^\RegW))$ is equal to the probability that applying $\Measurement$ twice in succession to $\DMatrix$ yields outcomes $p,p'$ such that $|p-p'| \leq \varepsilon$, and hence is at least $1-\delta$. The claim follows by the gentle measurement lemma (\cref{lemma:gentle-measurement}).
\end{proof}

To complete the proof of the lemma, we make use of the following key claim about $\Expt_2$. Roughly speaking, we show that in $\Expt_2$, if we obtain outcome $p \in I$ in \cref{step:initial-measurement-3} and an outcome $k \neq \bot$ in \cref{step:perturb-3} (which occurs with probability at least $1-2\sqrt{\delta}$ due to~\cref{claim:stat-distance-hybrids}) where the probability of obtaining $k$ was $\beta$, then the final state $\DMatrix^*$ after \cref{step:repair-3} has the following guarantee: applying $\Measurement$ to $\DMatrix^*$ produces an outcome $p'$ within $2\varepsilon$ of $p$ except with probability inversely proportional to $\beta$.

\begin{claim}
\label{claim:state-repair-main}
Fix $p \in I, k \in [N]$; let $\ket{\phi_{\MeasA}}$ be an arbitrary state in $\image(\BProj{\MeasA,p}) \subseteq \RegH \otimes \RegW$, and define $\ket{\phi_{\MeasB}} \coloneqq \BProj{\MeasB,k}\ket{\phi_{\MeasA}}/\sqrt{\beta}$ where $\beta \coloneqq \norm{\BProj{\MeasB,k}\ket{\phi_{\MeasA}}}^2$. Applying $\Repair'(k,p)$ to $\ket{\phi_{\MeasB}} \in \RegH \otimes \RegW$ yields the state $\DMatrix^* \in \Hermitians{\RegH}$ where
\begin{equation*}
    \Pr_{p' \gets \Measurement(\DMatrix^*)}[|p' - p| > 2\varepsilon] \leq (\delta + 1/T)/\beta,
\end{equation*}
and $\Repair'(k,p)$ applies $1+1/\beta$ measurements in expectation.
\end{claim}

We show how \cref{lemma:state-repair} follows from \cref{claim:state-repair-main}, and subsequently prove \cref{claim:state-repair-main}.
	
	Write $\DMatrix_2 = \sum_{p \in I} \ketbra{p}^{\RegO_1} \otimes \DMatrix_p^{\RegH,\RegW}$; note that $\Tr(\BProj{\MeasA,p} \DMatrix_p) = \Tr(\DMatrix_p)$ due to the post-selection in \cref{step:repair-postselect}. By the definition of $\MeasC' \colon \Hermitians{\RegH} \to \Hermitians{\RegH \otimes \RegO_2}$,
	\[ \MeasC'(\DMatrix_2) = \sum_{p \in I} \ketbra{p}^{\RegO_1} \otimes \left( \BProj{\bot} \DMatrix_p^{\RegH,\RegW} \BProj{\bot} \otimes \ketbra{\bot}^{\RegO_2} + \sum_{k=1}^{N} \BProj{\MeasB,k} \DMatrix_p^{\RegH,\RegW} \BProj{\MeasB,k} \otimes \ketbra{k}^{\RegO_2} \right) \enspace. \]
	
	By \cref{claim:stat-distance-hybrids}, $\Tr(\BProj{\bot} \DMatrix_2) \leq \Tr(\BProj{\bot} \DMatrix_1) + 2\sqrt{\delta} = 2\sqrt{\delta}$. For $p \in I$, write $\DMatrix_{p} = \sum_{i} q_i \ketbra{\psi_{i}}$ for unit states $\ket{\psi_{i}} \in \RegH \otimes \RegW$ ; note that $\ket{\psi_{i}} \in \image(\BProj{\MeasA,p})$. For all $i$ and any $k \in [N]$, we can define $\ket{\psi_{i,k}} \eqdef \BProj{\MeasB,k} \ket{\psi_{i}}/\norm{\BProj{\MeasB,k} \ket{\psi_{i}}}$ and apply \cref{claim:state-repair-main} with $\ket{\phi_{\MeasA}}$ set to $\ket{\psi_{i}}$ to obtain
	\begin{equation*}
		\Pr_{p' \gets \Measurement(\DMatrix_{i,p,k}^*)}[|p' - p| > 2\varepsilon] \leq (\delta + 1/T)/\norm{\BProj{\MeasB,k} \ket{\psi_{i}}}^2 \enspace,
	\end{equation*}
	where $\DMatrix_{i,p,k}^* \in \Hermitians{\RegH}$ is the state after applying $\Repair'(k,p)$ to $\ket{\psi_{i,k}}$.
	
	To conclude, we show that $\Expt_2$ is $(2\varepsilon,N(\delta + 1/T) + 2\sqrt{\delta})$-almost projective; the statement for $\RepairExpt$ will then follow by \cref{claim:stat-distance-hybrids}. Let $\RegO_3$ be a new ancilla register that will store the outcome of the \emph{second} application of $\Expt_2$, and consider the projector $\BProj{\mathsf{bad}} \eqdef \sum_{p} \sum_{p' \notin [p \pm 2\varepsilon]} \ketbra{p,p'}^{\RegO_1,\RegO_3}$ corresponding to the event that applying $\Expt_2$ twice yields outcomes $(p,p')$ more than $2\varepsilon$ apart. Since the outcome of $\Expt_2$ is determined by the outcome of $\Measurement$, we have by convexity
	\begin{equation*}
		\Tr(\BProj{\mathsf{bad}} \cdot \Measurement(\Expt_2(\DMatrix))) \leq N(\delta + 1/T) + 2\sqrt{\delta} \enspace .
	\end{equation*}
	Hence by \cref{claim:stat-distance-hybrids},
	\begin{equation*}
		\Tr(\BProj{\mathsf{bad}} \cdot \Measurement(\RepairExpt(\DMatrix))) \leq N(\delta + 1/T) + 4\sqrt{\delta} \enspace,
	\end{equation*}
	which completes the proof that $\RepairExpt$ is $(2\varepsilon, N(\delta + 1/T) + 4\sqrt{\delta})$-almost projective.
	
	By \cref{claim:stat-distance-hybrids} and \cref{claim:state-repair-main}, and law of total expectation, it holds that
	\begin{align*}
	    \Expectation[R] &\leq d(\DMatrix_1',\DMatrix_2') \cdot T + \Tr(\BProj{\bot} \DMatrix_2') \cdot T + \sum_{p \in I,k \in [N]} \Tr(\BProj{\MeasB,k} \DMatrix_{p}) (1 + \Tr(\DMatrix_{p})/\Tr(\BProj{\MeasB,k}  \DMatrix_{p})) \\
	    & \leq 2T\sqrt{\delta} + 2T\sqrt{\delta} + \sum_{p \in I,k \in [N]} (\Tr(\BProj{\MeasB,k} \DMatrix_{p}) + \Tr(\DMatrix_{p}))\\
	    & \leq N + 4T\sqrt{\delta}+ 1 \enspace,
	\end{align*}
	which concludes the proof, given \cref{claim:state-repair-main}.
\end{proof}

\begin{proof}[Proof of \cref{claim:state-repair-main}]

For this proof, we write $\MeasA,\MeasB$ for $\MeasA_{p},\MeasB_{k}$ and $\BProj{\MeasA},\BProj{\MeasB}$ for $\BProj{\MeasA,p}, \BProj{\MeasB,k}$ respectively.

Consider a decomposition of $\RegH \otimes \RegW$ into the Jordan subspaces $\{ \Subspace_{j} \}_j$ for projectors $\BProj{\MeasA}$ and $\BProj{\MeasB}$. Following our standard notation for Jordan subspaces, in the $j$-th Jordan subspace $\Subspace_{j}$, $\BProj{\MeasA}$ is a projection onto $\JorKetA{j}{1}$ and $\BProj{\MeasB}$ is a projection onto $\JorKetB{j}{1}$, and $p_j = |\JorBraKetAB{j}{1}|^2$. Recall that we write $\SProj[\Jor]{j}$ for the projection onto $\Subspace_{j}$.

Since $\ket{\phi_{\MeasA}} = \sum_j \alpha_j \JorKetA{j}{1}$ for some choice of $\{\alpha_j\}_j$, we can write $\ket{\phi_{\MeasB}}$ as
\begin{align*}
    \ket{\phi_{\MeasB}} = \frac{1}{\sqrt{\beta}}\sum_j\alpha_j \BProj{\MeasB} \JorKetA{j}{1} = \frac{1}{\sqrt{\beta}}\sum_j\alpha_j\sqrt{p_j}\JorKetB{j}{1}
\end{align*}

Let $\DMatrix' \in \Hermitians{\RegH \otimes \RegW}$ be the state immediately before ``Apply $\MeasUnitary$ to $(\RegH,\RegW)$, and discard the $\RegW$ registers.'' in $\Repair'(k,p)$, so that $\DMatrix^* = \Tr_{\RegW}(\MeasUnitary \DMatrix' \MeasUnitary^{\dagger})$. We first bound $\Tr(\Pi_{\MeasA} \DMatrix')$, i.e., the probability that $\Repair'(k,p)$ stops because $\MeasA \to 1$, by analyzing the distribution of measurement outcomes that result from applying a total of $2T+1$ alternating measurements $\MeasA,\MeasB,\MeasA,\MeasB,\ldots,\MeasA$. Note that the real $\Repair$ procedure terminates after obtaining a $1$ outcome for $\MeasA$; we consider the distribution of a fixed number of measurements for the purpose of analysis.

Let $I(b_1,b_2,\ldots,b_{2T+1})$ be the smallest $i$ such that $b_{2i+1} = 1$, or $T+1$ if there is no such $i$. Let $D$ be denote the following distribution:
\begin{enumerate}[noitemsep]
    \item \label[step]{step:sample-j} Sample $j$ with probability $\abs{\alpha_j}^2 p_j/\beta$
    \item Sample $(b_1,b_2,\ldots,b_{2T+1}) \gets \MWDistFull{p_j}{2T+1}$. 
    \item Output $I(b_1,b_2,\ldots,b_{2T+1})$.
\end{enumerate}
By~\cref{lemma:alt-meas-general-state}, the expected number of measurements applied by $\Repair$ is $2\Expectation[D]+1$, and
\[
\Tr(\Pi_{\MeasA} \DMatrix') = 1 - \Pr_{i \gets D}[i = T+1] \enspace.
\]

We now analyse the distribution $D$. Suppose that $j$ is sampled in \cref{step:sample-j}. The probability that $b_1 = 1$ occurs is then $p_j$. Then for each $i \in [T]$, the probability that $b_{2i+1} = 1$ given that $b_{2i-1} = 0$ is $2p_j(1-p_j)$. Hence conditioned on $j$ being sampled, $D$ is dominated by the random variable $D'$ which takes value $0$ with probability $p_j$ and is distributed as $\mathsf{Geo}(2p_j(1-p_j))$ with probability $1-p_j$, where $\mathsf{Geo}(q)$ is the geometric distribution with parameter $q$.

It follows that $\Expectation[D] \leq \frac{1}{\beta} \sum_{j} \abs{\alpha_j}^2 p_j(1-p_j) \Expectation[\mathsf{Geo}(2p_j(1-p_j))] = 1/(2\beta)$, and
\[
\Pr_{i \gets D}[i = T+1] \leq \frac{1}{\beta}\sum_{j} \abs{\alpha_j}^2 p_j (1-p_j) (1-2p_j(1-p_j))^{T} \leq \frac{1}{\beta T} \enspace,
\]
since $x(1-2x)^{T} \leq 1/T$ for all $x \in [0,1/4]$. This establishes that $\Tr(\BProj{\MeasA} \DMatrix') \geq 1-\frac{1}{\beta T}$.

To complete the proof of \cref{claim:state-repair-main}, we prove that applying $\Measurement$ to $\DMatrix^* = \Tr_{\RegW}(\MeasUnitary \DMatrix' \MeasUnitary^{\dagger})$ produces $p'$ within $2\varepsilon$ of $p$ with probability at least $1 - (\delta + 1/T)/\beta$.

Since $\MeasA$ and $\MeasB$ commute with $\SProj[\Jor]{j}$, $\Tr(\SProj[\Jor]{j} \DMatrix') = \norm{\SProj[\Jor]{j} \ket{\phi_{\MeasB}}}^2 = |\alpha_{j}|^2 p_j/\beta$. In particular, $\eta$ as defined in \cref{claim:good-states} is equal to $\beta$, and $\Tr(\SProj[\Jor]{j} \DMatrix') = 0$ for all $j$ with $p_j = 0$. By definition, the last measurement applied during $\Repair$ is $\MeasA$, and so since $\MeasA$ is projective, $\DMatrix' = \MeasA(\DMatrix') = \BProj{\MeasA} \DMatrix' \BProj{\MeasA} + (I-\BProj{\MeasA}) \DMatrix' (I-\BProj{\MeasA})$, which commutes with $\BProj{\MeasA}$. The statement then follows by \cref{claim:good-states}.
\end{proof}

\begin{claim}
    \label{claim:good-states}
    Suppose $\DMatrix' \in \Hermitians{\RegH \otimes \RegW}$ satisfies each of the following:
    \begin{itemize}[nolistsep]
        \item $\Tr(\Pi_{\MeasA} \DMatrix') = 1-\gamma$,
        \item $\DMatrix'$ commutes with $\Pi_{\MeasA}$, and
        \item $\Tr(\SProj[\Jor]{j} \DMatrix') = 0$ for all $j$ where $p_j = 0$.
    \end{itemize}
    Let 
    \[\eta \eqdef \frac{1}{\sum_{j, p_j > 0} \Tr(\SProj[\Jor]{j} \DMatrix')/p_j} \enspace,\]
    and $\DMatrix^* \eqdef \Tr_{\RegW}(\MeasUnitary \DMatrix' \MeasUnitary^{\dagger})$. Then
    \begin{equation*}
        \Pr_{p' \gets \Measurement(\DMatrix^*)}[|p' - p| > 2\varepsilon] \leq \delta/\eta + \gamma \enspace.
    \end{equation*}
\end{claim}
\begin{proof}
    Since $\DMatrix'$ commutes with $\Pi_{\MeasA}$, we can write $\DMatrix' = \sum_{i} q_{i} \ketbra{\phi_i}$, where the $\ket{\phi_i}$ are eigenstates of $\Pi_{\MeasA}$. Consider the unitary $U$ on $\RegH \otimes \RegW$ that maps $\JorKetA{j}{b}$ to $\JorKetA{j}{1-b}$ for $b \in \Bits$ for each $2$-dimensional Jordan subspace $\Subspace_j$, and acts as identity on each $1$-dimensional subspace. Formally,
    \[ U \eqdef \sum_{j,p_j \notin \{0,1\}} (\JorKetA{j}{1}\JorBraA{j}{0} + \JorKetA{j}{0}\JorBraA{j}{1}) + \sum_{j, p_j = 1} \JorKetBraA{j}{1} + \sum_{j, p_j = 0} \JorKetBraA{j}{0} \enspace.\]
    In particular, if $\ket{\phi_i} = \sum_{j, p_j > 0} \zeta_j \JorKetA{j}{0}$, then $U \ket{\phi_i} = \sum_{j, p_j > 0} \zeta_j \JorKetA{j}{1} \in \image(\Pi_{\MeasA})$. Moreover, $\SProj[\Jor]{j}$ commutes with $U$ for all $j$.
    
    Let $\DMatrixB \eqdef \Pi_{\MeasA} \DMatrix' + U (\Id-\Pi_{\MeasA})\DMatrix' U^{\dagger}$. $\DMatrixB$ does not appear during the procedure; it is defined for the purpose of analysis. Intuitively, $\DMatrixB$ is the result of rotating, within each Jordan subspace, the part of $\DMatrix'$ in $\image(\Id-\BProj{\MeasA})$ into $\image(\BProj{\MeasA})$. By unitary invariance of the trace,
    \begin{equation*}
        \Tr(\DMatrixB) = \Tr(\Pi_{\MeasA} \DMatrix') + \Tr(U (\Id-\Pi_{\MeasA})\DMatrix' U^{\dagger}) = \Tr(\DMatrix') = 1 \enspace.
    \end{equation*}
    For all $j$, we have $\Tr(\SProj[\Jor]{j} \DMatrixB) = \Tr(\SProj[\Jor]{j} \DMatrix')$ since $\SProj[\Jor]{j}$ commutes with both $U$ and $\BProj{\MeasA}$. The trace distance between $\DMatrixB$ and $\DMatrix'$ is at most $\Tr((\Id - \BProj{\MeasA})\DMatrix') = \gamma$. Finally, by definition of $U$, $\Tr(\BProj{\MeasA} \DMatrixB) = 1$. 
    
    We will now show that the outcome of $\Measurement(\Tr_{\RegW}(\MeasUnitary \DMatrixB \MeasUnitary^{\dagger}))$ is in the range $p \pm 2\varepsilon$ with probability $\delta/\eta$, which will complete the proof by contractivity of the trace distance. Define the linear operator $C \eqdef \sum_{j,p_j > 0} \frac{1}{\sqrt{p_j}} \JorKetB{j}{1}\JorBraA{j}{1}$. Notice that $\BProj{\MeasA} C$ is the projection onto $\image(\BProj{\MeasA}) \cap (\bigoplus_{j, p_j > 0} \Subspace_{j})$ since
    \[
        \BProj{\MeasA} C = \sum_{j, p_j > 0} \frac{1}{\sqrt{p_j}} \JorProjA{j}{1} \JorKetB{j}{1}\JorBraA{j}{1} = \sum_{j, p_j > 0} \JorProjA{j}{1} \enspace.
    \]
    
    Let $\DMatrixB' \eqdef C \DMatrixB C^{\dagger}/\Tr(C \DMatrixB C^{\dagger})$. We have that
    \[ \Tr(C \DMatrixB C^{\dagger}) = \Tr(C^{\dagger} C \DMatrixB) = \sum_{j,p_j > 0} \frac{1}{p_j} \Tr(\JorKetBraA{j}{1} \DMatrixB) = \sum_{j, p_j > 0} \frac{1}{p_j} \Tr(\SProj[\Jor]{j} \DMatrix') = 1/\eta,\]
    and $\Tr(\Pi_{\MeasB} \DMatrixB') = 1$. By the definition of $\Pi_{\MeasB}$, this implies that $\DMatrixB' = \DMatrixB'' \otimes \ketbra{0}^{\RegW}$ for some $\DMatrixB'' \in \Hermitians{\RegH}$. We also have that $\Pi_{\MeasA} \DMatrixB' \Pi_{\MeasA} = \eta \BProj{\MeasA} C \DMatrixB C^{\dagger} \BProj{\MeasA} = \eta \DMatrixB$, where the second equality follows from the fact that $\DMatrixB \in \Hermitians{\image(\BProj{\MeasA}) \cap (\bigoplus_{j, p_j > 0} \Subspace_{j})}$ by construction. 
    
    Recall that (1) applying $\MeasUnitary$ to a state of the form $\DMatrix'' \otimes \ketbra{0}^{\RegW}$, then applying the projective measurement $(\BProj{\Measurement,q})_{q \in I}$ on $\RegW$ and tracing out $\RegW$ is equivalent to applying the $(\varepsilon,\delta)$-almost-projective measurement $\Measurement = (M_q)_{q \in I}$ to $\DMatrix''$ and (2) $\BProj{\MeasA} = \sum_{q \in [p \pm \varepsilon]} \MeasUnitary^\dagger \BProj{\Measurement,q} \MeasUnitary$. So we have:
    \begin{align*}
        \DMatrix^* &= \Tr_{\RegW}(\MeasUnitary \DMatrixB \MeasUnitary^{\dagger}) = \frac{1}{\eta} \Tr_{\RegW}(\MeasUnitary \BProj{\MeasA} \DMatrixB' \BProj{\MeasA} \MeasUnitary^{\dagger}) \\
        &= \frac{1}{\eta} \Tr_{\RegW}\left(\sum_{q,q' \in [p \pm \varepsilon]} \BProj{\Measurement,q} \MeasUnitary (\DMatrixB'' \otimes \ketbra{0}^{\RegW}) \MeasUnitary^{\dagger} \BProj{\Measurement,q'}\right) \\
        &= \frac{1}{\eta} \sum_{q \in [p \pm \varepsilon]} \Tr_{\RegW}\left(\BProj{\Measurement,q} \MeasUnitary (\DMatrixB'' \otimes \ketbra{0}^{\RegW}) \MeasUnitary^{\dagger}\right) \\
        &= \frac{1}{\eta} \sum_{q \in [p \pm \varepsilon]} M_{q} \DMatrixB'' M_{q}^{\dagger} = \frac{\sum_{q \in [p \pm \varepsilon]} M_{q} \DMatrixB'' M_{q}^{\dagger}}{\Tr(\sum_{q \in [p \pm \varepsilon]} M_{q} \DMatrixB'' M_{q}^{\dagger})} \enspace.
    \end{align*} 
    That is, $\DMatrix^*$ is the state after applying $\Measurement$ to $\DMatrixB''$ conditioned on obtaining an outcome in the range $p \pm \varepsilon$, which occurs with probability $\eta$. But then by $(\varepsilon,\delta)$-almost projectivity, the outcome of $\Measurement(\DMatrix^*)$ is in the range $p \pm 2\varepsilon$ with probability $1-\delta/\eta$.
\end{proof}

\subsection{Proof of~\cref{theorem:repeated-games}}
\label{sec:proof-of-main-theorem}

We now prove~\cref{theorem:repeated-games}. For $r \in \RSet$, define $\Meas{\Pred,r} \eqdef \BMeas{\BProj{\Pred,r}^{\RegZ,\RegI}}$ where \[\BProj{\Pred,r}^{\RegZ,\RegI} \eqdef U_{S,r}^{\dagger} \sum_{z,\Pred(r,z)=1} \ketbra{z}^{\RegZ} U_{S,r} ~,\]
where $U_{S,r}$ is a unitary implementation of the action of $S$ on message $r$.

We set $\ApproxEig$ and $\RepairProb$ as follows:
\begin{itemize}[noitemsep]
    \item Let $\ApproxEig_{\Game,\varepsilon,\delta}^{\StrategyMap}$ be a CPTP map from $(\RegResp,\RegInt)$ to $(\RegResp,\RegInt)$ as in \cref{lemma:approx-eig}.
    \item Let $\RepairProb_{\Game,\varepsilon,\delta,T,r}^{\StrategyMap} \eqdef \Repair_{T}^{\ApproxEig_{\Game,\varepsilon,\delta}^{\StrategyMap},\Measurement_{\Pred,r}}$ be a CPTP map from $(\RegResp,\RegInt)$ to $(\RegResp,\RegInt)$ as in \cref{lemma:state-repair} (that is, with $\RegH = (\RegResp,\RegInt)$).
\end{itemize}

The algorithm $A$ operates on registers $(\RegResp,\RegInt)$ and works as follows.
\begin{center}
\begin{minipage}{0.9\textwidth}
    $A^{\StrategyMap}_{\Game,n,\RWLoss}$:
    \begin{enumerate}[nolistsep]
        \item Let $\varepsilon \eqdef \RWLoss/(2n+2)$, $\delta \eqdef \RWLoss^2/cn^2$ for some universal constant $c$.
        \item \label[step]{step:repeated-main-loop} (Main loop.) For $i = 1,\ldots,n$,
        \begin{enumerate}[nolistsep,ref=\arabic{enumi}(\alph*)]
            \item Measure $p_i \gets \ApproxEig_{\Game,\varepsilon,\delta}^{\StrategyMap}$ on registers $(\RegResp,\RegInt)$.
            \item \label[step]{step:strat-apply-u} Receive $r_i \in \RSet$ from the referee and apply $U_{\StrategyMap,r_i}$ to $(\RegResp,\RegInt)$.
            \item \label[step]{step:strat-send-state} Send the register $\RegZ$ to the referee.
            \item \label[step]{step:strat-recv-state} Receive the (partially measured) register $\RegZ$ from the referee, along with the outcome $b_i \in \Bits$.
            \item \label[step]{step:ext-reverse-u} Apply $U_{\StrategyMap,r_i}^{\dagger}$ to $(\RegResp,\RegInt)$.
            \item Apply $\RepairProb_{\Game,\varepsilon,\delta,T,r_i}^{\StrategyMap}(p,b)$ to $(\RegResp,\RegInt)$ with $T \coloneqq \lceil 1/\sqrt{\delta} \rceil$.
        \end{enumerate}
    \end{enumerate}
\end{minipage}
\end{center}

\begin{claim}
    \label{claim:p_i}
    For each $i \in [n]$, $p_{i+1} \geq p_{i} - 2\varepsilon$ with probability $1-O(\sqrt{\delta})$.
\end{claim}
\begin{proof}
    \cref{step:strat-apply-u,step:strat-send-state,step:strat-recv-state,step:ext-reverse-u} are equivalent to applying $\Meas{\Pred,r_i}$ to $(\RegZ,\RegI)$. Since $\ApproxEig_{\Game,\varepsilon,\delta}^{\StrategyMap}$ is $(\varepsilon,\delta)$-almost projective (\cref{lemma:approx-eig}, \cref{item:aeig-approx-proj}), the claim follows from applying \cref{lemma:state-repair} with $\Measurement = \ApproxEig_{\Game,\varepsilon,\delta}^{\StrategyMap}$, $\RegH = (\RegResp,\RegInt)$, $\MeasC = \Meas{\Pred,r_i}$, $N = 2$, $T = \lceil 1/\sqrt{\delta} \rceil$ and observing that the entire ``Main loop'' amounts to a single invocation of $\RepairExpt$, and is therefore a $(2\varepsilon,O(\sqrt{\delta}))$-almost-projective measurement.\footnote{On the $(i+1)$-th invocation of the main loop the challenge $r_{i+1}$ will generally be different than the challenge $r_i$ used in the $i$-th invocation; however, almost projectivity still applies since $p_{i+1}$ is clearly independent of $r_{i+1}$.}
\end{proof}
Let $\DMatrix_{i}$ be the state on $(\RegResp,\RegInt)$ at the beginning of the $i$-th iteration.
\begin{claim}
    For all $i \in [n]$, $\Value{\Game}(\StrategyMap,\DMatrix_{i}) \geq \Value{\Game}(\StrategyMap,\DMatrix) - 2i \cdot \varepsilon - O(i \cdot \sqrt{\delta})$.
\end{claim}
\begin{proof}
    By \cref{claim:p_i}, with probability $1 - O(i \cdot \sqrt{\delta})$ it holds that \[ p_{i} \geq p_{i-1} - 2\varepsilon \geq p_{i-2} - 4\varepsilon \geq \cdots \geq p_{1} - 2(i-1)\varepsilon~. \] Then by \cref{lemma:approx-eig}, \cref{item:aeig-good-states}, \[\Value{\Game}(\StrategyMap,\DMatrix_{i}) \geq \Expectation[p_1] - 2i \cdot \varepsilon - O(i \cdot \sqrt{\delta})~.\]
    Finally, by \cref{lemma:approx-eig}, \cref{item:aeig-expectation}, $\Expectation[p_1] = \Value{\Game}(\StrategyMap,\DMatrix)$.
\end{proof}
Now since $\Value{\Game}(\StrategyMap,\ApproxEig(\DMatrixB)) \geq \Value{\Game}(\StrategyMap,\DMatrixB) - \delta$ for all states $\DMatrixB$ by \cref{lemma:approx-eig}, \cref{item:aeig-pm-state}, we have that $\Pr[b_{i} = 1] \geq \Value{\Game}(\StrategyMap,\DMatrix) - 2i \cdot \varepsilon - O(i \cdot \sqrt{\delta})$. Hence 
\begin{align*}
    \Value{\Game}^{n}(A^{\StrategyMap}_{\Game,n,\RWLoss},\DMatrix) &= \Expectation[\sum_{i \in [n]} b_i] = \sum_{i \in [n]} \Pr[b_{i} = 1]\\
    &\geq n \cdot (\Value{\Game}(\StrategyMap,\DMatrix) - (n+1)\varepsilon - O(n \cdot \sqrt{\delta}))\\
    &\geq n \cdot (\Value{\Game}(\StrategyMap,\DMatrix) - \RWLoss)~, 
\end{align*} 
which completes the proof. The expected running time of this procedure is $\tilde{O}(|\Pred| \cdot n/\RWLoss)$.

\doclearpage
\section{A quantum rewinding lemma}
\label{sec:rewinding}

We use \cref{theorem:repeated-games} to prove a ``quantum forking lemma'' for collapsing protocols. We denote by $(\Transcript,\DMatrix) \gets \Interact{\Malicious{\ARGProver}}{\ARGVerifier}_{\NumRounds-1}$ the partial transcript $\Transcript$ and intermediate state $\DMatrix$ of the malicious prover $\Malicious{\ARGProver}$ after running $\NumRounds-1$ rounds of the interaction between $\Malicious{\ARGProver}$ and $\ARGVerifier$. Recall that $\RSet_m$ denotes the set of random coins for round $m$ of the protocol.

\newcommand{\Extract}{\mathsf{Fork}}
\begin{theorem}
    \label{theorem:quantum-forking}
    Let $(\ARGProver,\ARGVerifier)$ be an $m$-round collapsing protocol. There exists an algorithm $\Extract$ running in expected polynomial time with black-box access to an adversary such that the following holds. Let $\Malicious{\ARGProver}$ be an efficient quantum adversary such that $\Pr[\Interact{\Malicious{\ARGProver}}{\ARGVerifier} \to 1] \geq \SuccProb$. Then for any $n \in \Naturals$, $\RWLoss \in [0,1]$,
    
    \begin{equation*}
        \Expectation\left[
            ~|W|~
        \middle\vert
            \begin{array}{r}
                (\Transcript,\DMatrix) \gets \Interact{\Malicious{\ARGProver}}{\ARGVerifier}_{\NumRounds-1} \\
                \vec{r} = (r_1,\ldots,r_n) \gets (\RSet_{\NumRounds})^n \\
                W \gets \Extract^{\Malicious{\ARGProver}}(1^{\secp},1^{1/\RWLoss},\Transcript,\vec{r},\DMatrix)
            \end{array}
        \right] \geq n(\SuccProb - \RWLoss) - n^2/|\RSet_m| - \negl(\secp) \enspace.
    \end{equation*}
    
    Moreover, with probability $1$, we have $\{(s_i,z_i)\}_{i} \gets \Extract^{\Malicious{\ARGProver}}(1^{\secp},1^{1/\RWLoss},\Transcript,\vec{r},\DMatrix)$ where:
    \begin{itemize}[noitemsep]
        \item $\ARGVerifier(\Transcript,s_i,z_i) = 1$ holds for all $i \in [k]$,
        \item all $s_i$ are distinct, and
        \item for each $i$ there exists $j \in [n]$ such that $s_i = r_k$.
    \end{itemize}
    $\Extract$ runs in expected time $\poly(\secp) \cdot \tilde{O}(n/\RWLoss)$.
\end{theorem}
\begin{proof}
    We define an interactive quantum algorithm $C$ that acts as the referee in an $n$-round single-player quantum game as in \cref{sec:alt-proj-new}. For $r \in \RSet$, define
    \begin{equation*}
        \BProj{\ARGVerifier,r} \eqdef \sum_{z, V(\Transcript,r,z) = 1} \ketbra{z} \enspace.
    \end{equation*}
    \begin{center}
    \begin{minipage}{0.9\textwidth}
        $C(\Transcript,\vec{r})$:
        \begin{enumerate}[nolistsep]
            \item \label[step]{step:ext-main-loop} Set $W \eqdef \emptyset$. For $j = 1,\ldots,n$,
            \begin{enumerate}[nolistsep,ref=\arabic{enumi}(\alph*)]
                \item Send $r_j \in \RSet$ to the player.
                \item Receive register $\RegResp$ from the player.
                \item \label[step]{step:ext-measure} Apply the binary measurement $\Meas{\ARGVerifier,r_j} \eqdef \BMeas{\BProj{\ARGVerifier,r_j}}$ to register $\RegResp$, obtaining outcome $b$.
                \item \label[step]{step:ext-measure-resp} If $b = 1$, measure $\RegResp_{\NumRounds}$ in the computational basis to obtain response $z$. If there is no $z'$ such that $(r_j,z') \in W$, set $W \gets W \cup \{ (r_j,z) \}$.
                \item Return register $\RegResp$ to the player.
            \end{enumerate}
            \item Output $W$.
        \end{enumerate}
    \end{minipage}
    \end{center}
    The extractor $\Extract$ is obtained by simulating $\Interact{A_{\Game,n,\RWLoss}^{\UPrvRound{\NumRounds}}}{C(\Transcript,\vec{r})}$, where $A_{\Game,n,\RWLoss}$ is the algorithm guaranteed by \cref{theorem:repeated-games} with $\Game \eqdef (\RSet_{\NumRounds},\ZSet_{\NumRounds},\ARGVerifier(\Transcript,\cdot,\cdot))$, and $\UPrvRound{\NumRounds}$ is the unitary that the prover applies in the final round. The properties of the output of $\Extract$ (aside from the expected size of $W$) follow immediately from the definition.
    
    By the collapsing property of $(\ARGProver,\ARGVerifier)$, the measurement in \cref{step:ext-measure-resp} is undetectable to any efficient distinguisher; in particular, it is undetectable to $A$. We can therefore apply \cref{theorem:repeated-games} to show that the expected number of successful iterations is at least $n(\SuccProb - \RWLoss) - \negl(\secp)$. The expected number of repeated $r_j$ is at most $n^2/|R|$, which yields the bound.
\end{proof}

\begin{remark}
    If the quantum prover $\Malicious{\ARGProver}$ has (non-uniform) \emph{quantum} advice, then in general we can only run the extractor once.
    
    However, if the malicious quantum prover $\Malicious{\ARGProver}$ has (non-uniform) \emph{classical} advice, we can generate $(\Transcript,\DMatrix) \gets \Interact{\Malicious{\ARGProver}}{\ARGVerifier}_{\NumRounds-1}$ as many times as we would like (obtaining a different $(\Transcript,\DMatrix)$ each time). By running $\Extract^{\Malicious{\ARGProver}}$ on each $(\Transcript,\DMatrix)$, we eventually obtain a set $W$ of accepting transcripts with a shared prefix $\tau$ where $|W| \geq n(\SuccProb - \RWLoss) - n^2/|\RSet_m|$ with probability arbitrarily close to $1$.
\end{remark}

\subsection{Special sound protocols}

\cref{theorem:quantum-forking} immediately implies that any collapsing $k$-special sound protocol is an argument of knowledge. We first define $k$-special soundness, and then briefly explain how to apply \cref{theorem:quantum-forking} to obtain this result. Recall that a \emph{sigma protocol} is a three-message protocol where the prover moves first.

\begin{definition}[Special soundness]
    A sigma protocol $(\ARGProver,\ARGVerifier)$ is $k$-\emph{special sound} if there exists an extractor $\SSExtractor$ such that, given $k$ accepting transcripts $(a,r_1,z_1),\ldots,(a,r_k,z_k)$ with all $r_i \in \RSet$ distinct, $\SSExtractor(x,a,(r_i,z_i)_{i=1}^{k})$ outputs $w$ such that $(x,w) \in \Relation$.
\end{definition}

\begin{theorem}
\label{theorem:special-soundness}
    Any collapsing $k$-special sound protocol is a post-quantum argument of knowledge with knowledge error $O(k/|\RSet|)$.
\end{theorem}
\begin{proof}[Proof sketch]
    Let $\Malicious{\ARGProver}$ be an adversary that convinces $\ARGVerifier$ with probability $\Advantage > 4k/|\RSet|$. The extractor $\ARGExtractor$ for $(\ARGProver,\ARGVerifier)$ operates as follows, where $\Extract$ is as guaranteed by \cref{theorem:quantum-forking}.
    \begin{enumerate}[noitemsep]
        \item Obtain first message $a$ from $\Malicious{\ARGProver}$; let $\DMatrix$ be the prover's state after sending $a$.
        \item Sample $\vec{r} = (r_1,\ldots,r_n) \gets \RSet^{n}$ uniformly at random, where $n = 8k/\Advantage$.
        \item Run $W \gets \Extract^{\Malicious{\ARGProver}}(1^{\secp},1^{1/\RWLoss},a,\vec{r},\DMatrix)$, for $\RWLoss = \Theta(\Advantage)$ to be chosen.
        \item If $|W| \geq k$, output $w \gets \ARGExtractor_{\mathsf{ss}}(x,a,W)$.
    \end{enumerate}
    $\RWLoss$ can be chosen such that $\Expectation[|W|] \geq n(\Advantage - 2k/|\RSet|)/2 \geq n\Advantage/4$, and so the probability that $|W| \geq n\Advantage/8 = k$ is at least $\Advantage/8$ by Markov's inequality. The theorem follows by the definition of special soundness.
\end{proof}

For constant $k$, \cref{theorem:special-soundness} states that the post-quantum knowledge error of any $k$-special sound collapsing sigma protocol is $O(1/|\RSet|)$, which asymptotically matches the classical knowledge error. Previously, the post-quantum knowledge error of such protocols was only shown to be $O(1/\sqrt{|\RSet|})$ via Unruh's rewinding lemma~\cite{Unruh12}. 

\begin{remark}
    \cref{theorem:special-soundness} alone is insufficient to imply post-quantum security of Kilian's protocol (when instantiated with a PCP of knowledge), since Kilian's protocol is \emph{not} $k$-special sound for any $k = \poly(\secp)$. In particular, $k$-special soundness requires successful extraction from \emph{any} set of $k$ accepting transcripts with distinct challenges $r_i$; the extractor for Kilian's protocol requires that the $r_i$ are also ``sufficiently random''. We therefore prove post-quantum security of Kilian's protocol in \cref{sec:kilian} by directly applying \cref{theorem:quantum-forking} to obtain accepting transcripts for randomly sampled $r_i$. 
\end{remark}

\doclearpage
\section{Collapsing vector commitments}
\label{sec:collapsing-vc}

We define collapsing vector commitments (\Cref{sec:vector-commitments}), and then prove that Merkle trees are collapsing vector commitments when the underlying hash function is collapsing (\Cref{sec:merkle-trees}). Later on, in \cref{sec:kilian}, we will formulate Kilian's protocol in terms of vector commitments, and establish its post-quantum security when the vector commitment is collapsing.

\subsection{Definition}
\label{sec:vector-commitments}

A (static) vector commitment scheme $\VCScheme$ \cite{CatalanoF13} consists of the following algorithms.
\begin{itemize}
\item $\VCGen(1^\secp,\Alphabet,\veclen)$ is a probabilistic algorithm that takes as input the security parameter $1^\secp$, an alphabet $\Alphabet$, and a vector length $\veclen \in \bbN$, and outputs a commitment key $\ck$.
\item $\VCCommit(\ck,\Message)$ is a (possibly probabilistic) algorithm that takes as input a commitment key $\ck$ and a vector $\Message \in \Alphabet^{\veclen}$, and outputs a commitment string $\VCcm$ and auxiliary information $\VCaux$.
\item $\VCOpen(\ck,\VCaux,\VCQuerySet)$ is a deterministic algorithm that takes as input a commitment key $\ck$, auxiliary information $\VCaux$, and a subset $\VCQuerySet \subseteq [\veclen]$, and outputs an opening proof $\VCauth$.
\item $\VCVerify(\ck,\VCcm,\VCQuerySet,\VCval,\VCauth)$ is a deterministic algorithm that takes as input a commitment key $\ck$, a commitment $\VCcm$, a subset $\VCQuerySet \subseteq [\veclen]$, alphabet symbols $\VCval \in \Alphabet^{\VCQuerySet}$, and an opening proof $\VCauth$, and outputs a bit $b \in \Bits$.
\end{itemize}
The vector commitment scheme $\VCScheme$ is \emph{complete} if for every security parameter $\secp$, alphabet $\Alphabet$, vector length $\veclen \in \bbN$, and adversary $\Adversary$,
\begin{equation*}
\Pr\left[
\VCVerify(\ck,\VCcm,\VCQuerySet,\Message[\VCQuerySet],\VCauth)=1
\;\middle\vert\;
\begin{array}{r}
\ck \gets \VCGen(1^\secp,\Alphabet,\veclen) \\
(\Message \in \Alphabet^{\veclen},\VCQuerySet \subseteq [\veclen]) \gets \Adversary(\ck) \\
(\VCcm,\VCaux) \gets \VCCommit(\ck,\Message) \\
\VCauth \gets \VCOpen(\ck,\VCaux,\VCQuerySet)
\end{array}
\right]
=1 \enspace.
\end{equation*}
The traditional definition of security for a vector commitment scheme is \emph{position binding}, which states that no efficient attacker can open any location to two different values. In more detail, for every security parameter $\secp$, alphabet $\Alphabet$, vector length $\veclen \in \bbN$, and polynomial-size (classical or quantum) adversary $\Adversary$,
\begin{equation*}
\Pr\left[
\begin{array}{c}
\exists\,i \in \VCQuerySet_1 \cap \VCQuerySet_2 \text{ s.t. } \VCval_1[i] \neq \VCval_2[i] \\
\wedge\;\VCVerify(\ck,\VCcm,\VCQuerySet_1,\VCval_1,\VCauth_1)=1 \\
\wedge\;\VCVerify(\ck,\VCcm,\VCQuerySet_2,\VCval_2,\VCauth_2)=1
\end{array}
\;\middle\vert\;
\begin{array}{r}
\ck \gets \VCGen(1^\secp,\Alphabet,\veclen) \\
\left(\VCcm,
\begin{array}{c}
\VCQuerySet_1 \subseteq [\veclen], \VCval_1 \in \Alphabet^{\VCQuerySet_1}, \VCauth_1 \\
\VCQuerySet_2 \subseteq [\veclen], \VCval_2 \in \Alphabet^{\VCQuerySet_2}, \VCauth_2
\end{array}
\right) 
\gets \Adversary(\ck) \\
\end{array}
\right]
= \negl(\secp) \enspace.
\end{equation*}
While position binding against classical adversaries suffices to prove security of Kilian's protocol against classical adversaries, it is not known whether position binding against quantum adversaries suffices to prove security of Kilian's protocol against quantum adversaries. (And, as discussed in \cref{sec:intro}, it is unlikely to.) Hence we rely on an additional \emph{collapsing} property that we introduce.

\begin{definition}
\label{def:vc-collapsing}
$\VCScheme$ is \emph{collapsing} if for every security parameter $\secp$, alphabet $\Alphabet$, vector length $\veclen \in \bbN$, and polynomial-size quantum adversary $\Adversary$,
\begin{equation*}
\Big|
\Pr[\VCCollapsingExp{0}{\secp}{\Alphabet}{\veclen}{\Adversary} = 1]
-
\Pr[\VCCollapsingExp{1}{\secp}{\Alphabet}{\veclen}{\Adversary} = 1]
\Big|
\leq \negl(\secp)
\enspace.
\end{equation*}
For $b \in \Bits$ the experiment $\VCCollapsingExp{b}{\secp}{\Alphabet}{\veclen}{\Adversary}$ is defined as follows:
\begin{enumerate}[noitemsep]
  \item The challenger samples $\ck \gets \VCGen(1^{\secp},\Alphabet,\veclen)$ and sends $\ck$ to $\Adversary$.
  \item $\Adversary$ replies with a classical message $(\VCcm,\VCQuerySet \subseteq [\veclen])$, and a quantum state on registers $(\RegV,\RegO)$, where the $\RegV$ registers contain strings $\VCval \in \Alphabet^{\VCQuerySet}$ and the $\RegO$ registers contain opening proofs $\VCauth$.
  \item The challenger computes into an ancilla register the bit $\VCVerify(\ck,\VCcm,\VCQuerySet,\RegV,\RegO)$ via some unitary $U$, measures the ancilla, and then applies $\contra{U}$ to uncompute. If the measured bit is $0$ (verification fails), the challenger aborts and outputs $\bot$.
  \item If $b = 0$, the challenger does nothing. If $b = 1$, the challenger measures the registers $(\RegV,\RegO)$ in the standard basis to obtain a string $\VCval$ and opening proof $\VCauth$, which it discards.
  \item The challenger returns the contents of the (potentially measured) registers $(\RegV,\RegO)$ to $\Adversary$.
  \item $\Adversary$ outputs a bit $b$, which is the output of the experiment.
\end{enumerate}
\end{definition}

\begin{remark}
The definition of collapse binding for standard commitments implies (classical-style) binding \cite{Unruh16-eurocrypt}. However, we do not know whether our definition of collapsing for vector commitments implies position binding in general, without imposing additional structure on the vector commitment.
\end{remark}

\subsection{Merkle trees are collapsing}
\label{sec:merkle-trees}

We describe Merkle trees as an instance of vector commitments (\Cref{sec:vector-commitments}), and then prove that they are collapsing when the underlying hash function is collapsing.

\begin{construction}
\label{construction:merkle-tree}
Let $\HashFamily = \{\HashDistribution_{\secp}\}_{\secp \in \Naturals}$ be a function family with input size $\hinlen(\secp)$ and output size $\houtlen(\secp) = \hinlen(\secp)/2$. Let $\VCScheme \eqdef \GetMerkle{\HashFamily}$ be the vector commitment for messages over alphabet $\Alphabet \eqdef\Bits^{\hinlen(\secp)}$ that is constructed as follows.
\begin{itemize}
\item $\VCGen(1^\secp,\Alphabet,\veclen)$: sample a hash function $\HashFunction \gets \HashDistribution_{\secp}$ and output the commitment key $\ck \eqdef (\veclen,\HashFunction)$.
\item $\VCCommit(\ck,\Message)$: use $\HashFunction \colon \Bits^{\hinlen(\secp)} \to \Bits^{\hinlen(\secp)/2}$ to pairwise hash the message $\Message$ to obtain a corresponding Merkle tree $\MTree$ with root $\MRoot \in \Bits^{\hinlen(\secp)/2}$, and then output $\VCcm \eqdef \MRoot$ as a commitment and $\VCaux \eqdef (\Message,\MTree)$ as auxiliary information.
\item $\VCOpen(\ck,\VCaux,\VCQuerySet)$: for each index $i \in \VCQuerySet$, deduce the authentication path $\MPath_{i}$ for index $i$ in the Merkle tree $\MTree$, and then output the opening proof $\VCauth \eqdef (\MPath_{i})_{i \in \VCQuerySet}$. (Some of the paths may have overlaps, in which case the opening proof $\VCauth$ can be compressed accordingly.)
\item $\VCVerify(\ck,\VCcm,\VCQuerySet,\VCval,\VCauth)$: for each index $i \in \VCQuerySet$, check that the authentication path $\MPath_{i}$ in $\VCauth$ is for messages of length $\veclen$, and that it authenticates the value $\VCval_{i}$ for location $i$ in a Merkle tree with root $\VCcm$.
\end{itemize}
\end{construction}

It is well-known that Merkle trees satisfy the position binding property.

\begin{claim}
\label{claim:merkle-trees-are-position-binding}
If $\HashFamily$ is a collision-resistant hash function with input size $\hinlen(\secp)$ and output size $\houtlen(\secp) = \frac{\hinlen(\secp)}{2}$ against classical (resp., quantum) adversaries then $\VCScheme \eqdef \GetMerkle{\HashFamily}$ is a position-binding vector commitment scheme over alphabet $\Alphabet \eqdef \Bits^{\hinlen(\secp)}$ against classical (resp., quantum) adversaries.
\end{claim}

We now show that if $\HashFamily$ is a collapsing hash function then $\VCScheme \eqdef \GetMerkle{\HashFamily}$ is a collapsing vector commitment.

\begin{claim}
\label{claim:merkle-trees-are-collapsing}
If $\HashFamily$ is a collapsing hash function with input size $\hinlen(\secp)$ and output size $\houtlen(\secp) = \frac{\hinlen(\secp)}{2}$ then $\VCScheme \eqdef \GetMerkle{\HashFamily}$ is a collapsing vector commitment over alphabet $\Alphabet \eqdef \Bits^{\hinlen(\secp)}$.
\end{claim}

\begin{proof}
The proof is a standard application of the collapsing hash function security property. We write the proof for the case of a singleton query set $Q = \{ i \}$; extending to the general case is straightforward.

Fix a message length $\veclen$, and let $\MHeight \eqdef \lceil \log_2 \veclen \rceil$ be the height of a Merkle tree for messages of length $\veclen$. For $j \in \{0,1,\dots,\MHeight\}$, we define a hybrid experiment $\Hybrid_j$ as follows:
\begin{enumerate}[noitemsep]
  \item The challenger samples $\HashFunction \gets \HashDistribution_{\secp}$ and sends $\HashFunction$ to $\Adversary$.
  \item $\Adversary$ replies with a classical message $(\MRoot,i \in [\veclen])$ (a Merkle root and a location) and a quantum state on registers $(\RegV,\RegO_1,\dots,\RegO_{\MHeight})$, where the register $\RegV$ corresponds to strings in $\Alphabet^{\VCQuerySet}$ and each register $\RegO_j$ corresponds to the $j$-th node in the Merkle opening proof ($j = 1$ is a leaf node).  For convenience we set $\RegY_1 \eqdef \RegV$.
  \item The challenger coherently applies $\VCVerify$ using $\MHeight$ ancilla registers $\RegY_2,\dots,\RegY_{\MHeight+1}$:
  \begin{enumerate}
    \item Let $\MTUnitary_{i}$ be a unitary on the registers $(\RegO_1,\dots,\RegO_{\MHeight},\RegY_1,\dots,\RegY_{\MHeight+1})$ that works as follows: for $k = 1,\dots,\MHeight$, apply $\HashFunction$ to $(\RegY_{k},\RegO_k)$ or $(\RegO_k,\RegY_{k})$ (depending on the $k$-th bit of $i$) and XOR the result onto $\RegY_k$.
    \item Apply $\MTUnitary_{i}$ and then measure the bit indicating whether $\RegY_{\MHeight+1}$ equals $\MRoot$ (by applying the binary projective measurement $\BMeas{\ketbra{\MRoot}^{\RegY_{\MHeight+1}}}$). If the measured bit is $0$ (verification fails), then the challenger aborts and outputs $\bot$. 
  \end{enumerate}
  \item The challenger measures registers $(\RegO_{\MHeight-j+1},\dots,\RegO_{\MHeight})$ and $(\RegY_{\MHeight-j+1},\dots,\RegY_{\MHeight+1})$. (If $j = 0$ then the challenger does not measure any of the $\RegO$ registers.)
  \item The challenger applies $\contra{\MTUnitary_{i}}$ to uncompute the $\RegY_2,\dots,\RegY_{\MHeight+1}$ registers, and returns the registers $(\RegV,\RegO_1,\dots,\RegO_{\MHeight})$ to the adversary $\Adversary$.
\end{enumerate}
Hybrid $\Hybrid_0$ corresponds to the experiment $\VCCollapsingExp{0}{\secp}{\Alphabet}{\veclen}{\Adversary}$ and hybrid $\Hybrid_{\MHeight}$ corresponds to the experiment $\VCCollapsingExp{1}{\secp}{\Alphabet}{\veclen}{\Adversary}$, for the vector commitment scheme $\VCScheme \eqdef \GetMerkle{\HashFamily}$. (See \Cref{def:vc-collapsing} for the definition of the collapsing experiment for $\VCScheme$.)

We are left to argue that, for each $j \in \{0,1,\dots,\MHeight-1\}$, $\Hybrid_j$ and $\Hybrid_{j+1}$ are indistinguishable. Suppose by way of contradiction that for some $j \in \{0,1,\dots,\MHeight-1\}$ the attacker $\Adversary$ can distinguish $\Hybrid_j$ and $\Hybrid_{j+1}$ with advantage at least $\epsilon$. We construct an adversary $\Adversary_{j}$ that has distinguishing advantage at least $\epsilon$ for $\HashFamily$'s collapsing experiment $\HCollapsingExp{b}{\secp}{\Adversary_{j}}$ (see \Cref{def:collapsing}). The adversary $\Adversary_{j}$ works as follows.
\begin{enumerate}[noitemsep]
  \item Receive a hash function $\HashFunction$ from the challenger.
  \item Send $\HashFunction$ to $\Adversary$, and obtain the message $(\MRoot,i)$ and a quantum state on registers $(\RegV,\RegO_1,\dots,\RegO_{\MHeight})$.
  \item Similarly to the challenger in the hybrids, set $\RegV \eqdef \RegY_1$, prepare $\MHeight$ internal ancilla registers $\RegY_2,\dots,\RegY_{\MHeight+1}$, and apply the same unitary $\MTUnitary_{i}$ on $(\RegO_1,\dots,\RegO_{\MHeight},\RegY_1,\dots,\RegY_{\MHeight+1})$.
  \item Measure the bit indicating whether $\RegY_{\MHeight+1}$ equals the Merkle root $\MRoot$, and aborts if this measurement does not return $1$.
  \item Measure $(\RegO_{\MHeight-j+1},\dots,\RegO_{\MHeight})$ and $(\RegY_{\MHeight-j+1},\dots,\RegY_{\MHeight+1})$.
  \item Forward the contents of $(\RegO_{\MHeight-j},\RegY_{\MHeight-j})$ to the challenger as the hash function input, and forward $\RegY_{\MHeight-j}$ as the classical output. (If $b = 0$, the challenger in the collapsing experiment will not disturb the state on $(\RegO_{\MHeight-j},\RegY_{\MHeight-j})$; if instead $b = 1$, the challenger measures $(\RegO_{\MHeight-j},\RegY_{\MHeight-j})$ before returning these registers to $\Adversary_{j}$.)
  \item Apply $\MTUnitary_{i}$ again and return the registers $(\RegV,\RegO_1,\dots,\RegO_{\MHeight})$ to $\Adversary$.
  \item Output whatever $\Adversary$ outputs.
\end{enumerate}
The proof is concluded by observing that $\Adversary$'s view when inside the experiment $\HCollapsingExp{0}{\secp}{\Adversary_{j}}$ corresponds to hybrid $\Hybrid_j$ and $\Adversary$'s view when inside the experiment $\HCollapsingExp{1}{\secp}{\Adversary_{j}}$ corresponds to hybrid $\Hybrid_{j+1}$.
\end{proof}

\doclearpage
\section{Post-quantum security of Kilian's protocol}
\label{sec:kilian}

Denote by $\GetKilian{\PCPSystem}{\VCScheme}$ the instantiation of Kilian's protocol with PCP system $\PCPSystem$ and vector commitment scheme $\VCScheme$ (see \cref{sec:kilian-protocol} below). We prove the following theorem.

\begin{theorem}
\label{theorem:kilian}
Let $\PCPSystem$ be a PCP system for $\Relation$ with negligible soundness error, and let $\VCScheme$ be a collapsing vector commitment. Then $\GetKilian{\PCPSystem}{\VCScheme}$ is a post-quantum succinct argument for $\Relation$. Moreover, if $\PCPSystem$ has negligible knowledge error, then $\GetKilian{\PCPSystem}{\VCScheme}$ is also a post-quantum succinct argument of knowledge for $\Relation$.
\end{theorem}

\begin{corollary}
\label{corollary:from-lwe}
Assuming the post-quantum hardness of LWE, there exist post-quantum succinct arguments for $\NP$.
\end{corollary}

\begin{proof}
Collapsing vector commitments can be obtained from collapsing hash functions (\cref{claim:merkle-trees-are-collapsing}), which in turn exist based on the post-quantum hardness of LWE \cite{Unruh16-asiacrypt}. The corollary follows from \Cref{theorem:kilian} applied to a PCP for $\NP$ with suitable efficiency (e.g., \cite{BabaiFLS91}).
\end{proof}

The rest of this section is organized as follows:
in \Cref{sec:pcps} we recall the definition of a PCP;
in \cref{sec:kilian-protocol} we describe Kilian's protocol and prove that $\GetKilian{\PCPSystem}{\VCScheme}$ is collapsing if $\VCScheme$ is a collapsing vector commitment;
in \cref{sec:kilian-proof} we prove \cref{theorem:kilian}.

\subsection{Probabilistically checkable proofs}
\label{sec:pcps}

A \emph{probabilistically checkable proof} (PCP) for a relation $\Relation$ with soundness error $\PCPError$, alphabet $\Alphabet$, and proof length $\ProofLength$, is a pair of polynomial-time algorithms $\PCPSystem=\PCPTuple$ satisfying the following.
\begin{itemize}

  \item \textbf{Completeness.}
For every instance-witness pair $(\Instance,\Witness) \in \Relation$, $\PCPProver(\Instance,\Witness)$ outputs a proof string $\PCPProof \colon [\ProofLength] \to \Alphabet$ such that $\Pr[\PCPVerifier^{\PCPProof}(1^{\secp},\Instance)=1]=1$.

  \item \textbf{Soundness.}
For every instance $\Instance \not\in \Language(\Relation)$ and proof string $\PCPProof \colon [\ProofLength] \to \Alphabet$, $\Pr[\PCPVerifier^{\PCPProof}(1^{\secp},\Instance)=1] \leq \PCPError$.

\end{itemize}
Probabilities are taken over the randomness $\PCPVRandomness$ of $\PCPVerifier$. The \emph{randomness complexity} $\PCPRandComplexity$ is the number of random bits used by $\PCPVerifier$, and the \emph{query complexity} $\PCPQuery$ is the number of locations of $\PCPProof$ read by $\PCPVerifier$. The quantities $\PCPError,\ProofLength,\Alphabet,\PCPRandComplexity,\PCPQuery$ can be functions of the instance size $|\Instance|$.

We also consider PCPs that achieve a \emph{proof of knowledge} property, which is a strengthening of the soundness property.
\begin{itemize}

  \item \textbf{Proof of knowledge.}
$\PCPSystem$ has knowledge error $\PCPKnowledge$ if there exists a polynomial-time extractor algorithm $\PCPExtractor$ such that, for every instance $\Instance$ and proof string $\PCPProof \colon [\ProofLength] \to \Alphabet$, if $\Pr[\PCPVerifier^{\PCPProof}(\Instance)=1]>\PCPKnowledge$ then $\PCPExtractor(\Instance, \PCPProof)$ outputs $\Witness$ such that $(\Instance,\Witness) \in \Relation$.

\end{itemize}

\subsection{Kilian's protocol}
\label{sec:kilian-protocol}

Kilian's protocol \cite{Kilian92} is a public-coin four-message interactive argument $\ARGSystem=(\ARGProver,\ARGVerifier)$ obtained by combining two ingredients:
\begin{itemize}[noitemsep]
  \item a PCP system $\PCPSystem=\PCPTuple$ with alphabet $\Alphabet$, proof length $\ProofLength$, randomness complexity $\PCPRandComplexity$, and query complexity $\PCPQuery$; and
  \item a VC scheme $\VCScheme=\VCTuple$ over alphabet $\Alphabet$.
\end{itemize}
The construction of the interactive argument, which we denote by $(\ARGProver,\ARGVerifier) \eqdef \GetKilian{\PCPSystem}{\VCScheme}$, is specified below. The argument prover $\ARGProver$ and argument verifier $\ARGVerifier$ receive as input a security parameter $\secp$ (in unary) and an instance $\Instance$, while $\ARGProver$ additionally receives as input a witness $\Witness$ for $\Instance$.
\begin{enumerate}[noitemsep]

  \item $\ARGVerifier$ samples a commitment key $\ck \gets \VCGen(\secp,\ProofLength)$ and sends $\ck$ to $\ARGProver$.

  \item $\ARGProver$ computes a PCP string $\PCPProof \gets \PCPProver(\Instance,\Witness)$, computes a commitment to it $(\VCcm,\VCaux) \gets \VCCommit(\ck,\PCPProof)$, and sends $\VCcm$ to $\ARGVerifier$. 

  \item $\ARGVerifier$ samples PCP randomness $\PCPVRandomness \gets \Bits^{\PCPRandComplexity}$ and sends $\PCPVRandomness$ to $\ARGProver$.

  \item $\ARGProver$ runs the PCP verifier $\PCPVerifier^{\PCPProof}(\Instance;\PCPVRandomness)$ to deduce a set $\PCPQuerySet \subseteq [\ProofLength]$ of queries made by $\PCPVerifier$, computes an opening proof $\VCauth \gets \VCOpen(\ck,\VCaux,\PCPQuerySet)$, and sends $(\PCPProof[\PCPQuerySet], \VCauth)$ to $\ARGVerifier$.

  \item $\ARGVerifier$ checks that $\PCPVerifier(\Instance;\PCPVRandomness)$ accepts when answering its PCP queries via $\PCPProof[\PCPQuerySet] \in \Alphabet^{\PCPQuerySet}$ and that $\VCVerify(\ck,\VCcm, \PCPQuerySet, \PCPProof[\PCPQuerySet], \VCauth)=1$. (If $\PCPVerifier$ makes any query outside of $\PCPQuerySet$ then reject.)

\end{enumerate}

We show that $\GetKilian{\PCPSystem}{\VCScheme}$ is a collapsing protocol when $\VCScheme$ is collapsing.
\begin{claim}
\label{claim:kilian-collapsing}
    If $\VCScheme$ is a collapsing vector commitment then for all $\PCPSystem$, $\GetKilian{\PCPSystem}{\VCScheme}$ is a collapsing protocol.
\end{claim}
\begin{proof}
    Consider an adversary $\Adversary$ for $\ProtocolCollapseExp$ for $\Kilian$. We construct an $\Adversary'$ for $\mathsf{VCCollapseExp}$ with the same advantage as follows:
    \begin{enumerate}[nolistsep]
        \item Obtain $\ck$ from the challenger and send it to $\Adversary$. Measure the response $\VCcm$.
        \item Choose $r \gets \Bits^{\PCPRandComplexity}$ and send it to $\Adversary$. Send $(\VCcm,Q)$ and the (unmeasured) state on $\RegResp_2$ to the challenger, where $Q$ is the query set corresponding to $r$.
        \item Receive a state on $\RegResp_2$ and pass it to $\Adversary$. Return the output of $\Adversary$. \qedhere
    \end{enumerate}
\end{proof}

\subsection{Proof of \cref{theorem:kilian}}
\label{sec:kilian-proof}

Since Kilian's protocol instantiated with a collapsing vector commitment $\VCScheme$ is collapsing (\cref{claim:kilian-collapsing}), there exists an algorithm $\Extract^{\Malicious{\ARGProver}}$ making black-box queries to any malicious prover $\Malicious{\ARGProver}$ for Kilian's protocol that satisfies the guarantees of  \cref{theorem:quantum-forking}. We use $\Extract^{\Malicious{\ARGProver}}$ to implement an extractor $\ARGExtractor^{\Malicious{\ARGProver}}$ that makes black-box queries to $\Malicious{\ARGProver}$ and outputs a PCP string $\PCPProof \in \Alphabet^{\ProofLength}$.
\begin{enumerate}[noitemsep]
\item[] $\ARGExtractor^{\Malicious{\ARGProver}(x;\ket{\psi})}(1^{\secp},x,1^{1/\Advantage})$:
\item Sample a commitment key $\ck \gets \VCGen(\secp,\ProofLength)$ and query $\Malicious{\ARGProver}$ on $\ck$ to obtain a commitment $\VCcm$. Let $\tau \eqdef (\ck,\VCcm)$, and let $\DMatrix$ denote the intermediate state of $\Malicious{\ARGProver}$.
\item \label[step]{step:fork} Set $n \eqdef 60 \ProofLength \cdot \log(2|\Alphabet|)/\Advantage$, sample $\vec{r} = (r_1,\ldots,r_n)$ uniformly at random from $(\{ 0,1 \}^{\PCPRandComplexity})^n$, and run $(\tau,(r_1,\LastMsg_1),\ldots,(r_k,\LastMsg_k)) \gets \Extract^{\Malicious{\ARGProver}}(1^{\secp},1^{3/\Advantage},\Transcript,\vec{r},\DMatrix)$. Abort if $k < 6\ProofLength \cdot \log(2|\Alphabet|)$.
\item Parse each $\LastMsg_i$ as  $(\PCPProof[\PCPQuerySet_{r_i}], \VCauth)$, where $\PCPQuerySet_{r_i}$ is defined to be the set of indices that $\PCPVerifier(\Instance)$ queries on random coins $r_i$.
\item\label[step]{step:consistency-check} Check that $\{(\PCPQuerySet_{r_i},\PCPProof[\PCPQuerySet_{r_i}]\}_{i \in [k]}$ are \emph{consistent}, meaning that there does not exist a PCP index $t$ with two different values. If this check fails, abort and output $\bot$.
\item Output a $\PCPProof$ obtained by combining the answers given in $\{(\PCPQuerySet_{r_i},\PCPProof[\PCPQuerySet_{r_i}])\}_{i \in [k]}$ and filling in any unanswered indices arbitrarily.
\end{enumerate}

\begin{claim}
\label{claim:kilian-low-abort-prob}
    $\Pr[\bot \gets \ARGExtractor^{\Malicious{\ARGProver}}] \leq 1- \Omega(\varepsilon) + \negl(\secp)$
\end{claim}

\begin{proof}
    We first bound the probability that $\ARGExtractor$ aborts in \cref{step:fork}. Define $\SuccProb_{\ck}$ to be the probability that $\Malicious{\ARGProver}$ wins when $\ck$ is sampled in the first round; note that $\Expectation_{\ck}[\SuccProb_{ck}] \geq \Advantage$. By \cref{theorem:quantum-forking}, $\Expectation[k|\ck] \geq \SuccProb_{\ck} - \gamma \cdot \Advantage$ for some $\gamma < 1$. Hence by Markov's inequality,
    \begin{equation*}
        \Pr[k < 6\ProofLength \cdot \log(2|\Alphabet|)] = 1-\Omega(\Advantage) \enspace.
    \end{equation*}
    By the position-binding property of $\VCScheme$, the probability that $\ARGExtractor^{\Malicious{\ARGProver}}$ aborts in \cref{step:consistency-check} is $\negl(\secp)$.
    
    It follows that $\Pr[\PCPProof \gets \ARGExtractor^{\Malicious{\ARGProver}}] \geq \Omega(\varepsilon) - \negl(\secp)$. 
\end{proof}

For a PCP $\PCPProof$, let 
$\WIN{\PCPVerifier}{\Instance}{\PCPProof} \coloneqq \Pr[\PCPVerifier^{\PCPProof}(\Instance) =1].$ We prove that conditioned on the event that $\PCPProof \gets \ARGExtractor^{\Malicious{\ARGProver}}$, we have $\WIN{\PCPVerifier}{\Instance}{\PCPProof} \geq k/(2n)$ with overwhelming probability.

\begin{claim}
\label{claim:kilian-pcp-good}
$\Pr[(\PCPProof \gets \ARGExtractor^{\Malicious{\ARGProver}}) \wedge (\PCPProof \neq \bot) \wedge  (\WIN{\PCPVerifier}{\Instance}{\PCPProof} < k/(2n))] \leq \negl(\secp)$.
\end{claim}

\begin{proof}
    We first argue that for any fixed string $\PCPProof^* \in \Alphabet^{\ProofLength}$ where $\WIN{\PCPVerifier}{\Instance}{\PCPProof^*} < k/(2n)$, we have:
    \[\Pr[\PCPProof^* = \PCPProof \wedge \PCPProof \gets \ARGExtractor^{\Malicious{\ARGProver}}] \leq (2|\Alphabet|)^{-\ProofLength}.\]
    The probability  $\ARGExtractor^{\Malicious{\ARGProver}}$ outputs such a $\PCPProof^*$ is upper bounded by the probability that for randomly sampled $(r_1,\dots,r_n)$, there \emph{exist} $k$ distinct $r_i$ such that $\PCPVerifier^{\PCPProof^*}(x;r_i) = 1$. For each $r_i$, the probability $\Pr[\PCPVerifier^{\PCPProof^*}(x;r_i) = 1] < k/(2n)$, so by a multiplicative Chernoff bound (\cref{prop:chernoff-2}) we have
    \begin{equation*}
	\Pr_{r_1,\dots,r_n}[\text{exists }k\text{ distinct }i\in[n]\text{ such that } \PCPVerifier^{\PCPProof^*}(x;r_i)=1] \leq e^{-k/6} = (2|\Alphabet|)^{-\ProofLength} \enspace.
    \end{equation*}
    
    A union bound over all $\PCPProof^*$ completes the proof: 
    \begin{align*}
        &\Pr[(\PCPProof \gets \ARGExtractor^{\Malicious{\ARGProver}}) \wedge (\PCPProof \neq \bot) \wedge  (\WIN{\PCPVerifier}{\Instance}{\PCPProof} < k/(2n))] \\
        &= \sum_{\PCPProof^*,\WIN{\PCPVerifier}{\Instance}{\PCPProof^*} < k/(2n) }\Pr[\PCPProof^* = \PCPProof \wedge \PCPProof \gets \ARGExtractor^{\Malicious{\ARGProver}}]\\
        &\leq |\Alphabet|^{\ProofLength}/(2|\Alphabet|)^{\ProofLength} = \negl(\secp). \qedhere
    \end{align*}
\end{proof}

By combining \cref{claim:kilian-low-abort-prob,claim:kilian-pcp-good} with the fact that $k/(2n) \geq \varepsilon/20$, we obtain 
\[\Pr[(\PCPProof \gets \ARGExtractor^{\Malicious{\ARGProver}}) \wedge (\WIN{\PCPVerifier}{\Instance}{\PCPProof} \geq \varepsilon/20)] \geq \Omega(\varepsilon) - \negl(\secp).\]

If $\PCPSystem$ has negligible soundness error, then this implies that $\varepsilon = \negl(\secp)$.

If $\PCPSystem$ is a proof of knowledge with negligible knowledge error $\PCPKnowledge = \negl(\secp)$ with witness extractor $\PCPExtractor$, then the following extractor achieves knowledge error $\KnowledgeError = \negl(\secp)$: run $\PCPProof \gets \ARGExtractor^{\Malicious{\ARGProver}}$ and output $\Witness \gets \PCPExtractor(\Instance,\PCPProof)$.

\doclearpage
\section*{Acknowledgements}

Part of this work was done while FM was visiting UC Berkeley and the Simons Institute for the Theory of Computing from Fall 2019 to Spring 2020.
AC is supported by the Ethereum Foundation.
FM thanks Justin Holmgren for helpful discussions.
NS is supported by DARPA under Agreement No.\ HR00112020023. NS thanks Dominique Unruh for helpful discussions.
\bibliographystyle{alpha}
\bibliography{references}

\end{document}